\documentclass[trsc,nonblindrev]{informs3} 

\OneAndAHalfSpacedXI 

\usepackage{graphicx}
\usepackage{enumitem}
\newlist{steps}{enumerate}{1}
\setlist[steps, 1]{label = Step \arabic*:}
\usepackage{makecell}
\usepackage{natbib}
 \bibpunct[, ]{(}{)}{;}{a}{,}{,}%
 \def\bibfont{\small}%
\usepackage{amsmath}
\usepackage{subfigure}
\usepackage{threeparttable}

\TheoremsNumberedThrough     

\EquationsNumberedThrough    

\usepackage{ulem}
\usepackage{multirow}
\DeclareMathOperator*{\minimize}{minimize}
\usepackage{soul}
\newcommand*\diff{\mathop{}\!\mathrm{d}}

\usepackage[colorlinks=true,
            linkcolor=blue,
            citecolor=blue,
            urlcolor=blue]{hyperref}

\begin{document}



\RUNTITLE{Heterogeneous network designs}

\TITLE{Optimal designs of heterogeneous grid transit networks}

\ARTICLEAUTHORS{
\AUTHOR{Wenbo Fan$^{1,*}$, Haoyang Mao$^{1}$, Li Zhen, Weihua Gu$^{*}$}
\AFF{Department of Electrical and Electronic Engineering, The Hong Kong Polytechnic University, Hong Kong SAR}

\AFF{\vspace{0.3cm} $^{1}$ Wenbo Fan and Haoyang Mao contributed equally.}
\AFF{$^{*}$ Corresponding authors: \EMAIL{wenbo.fan@polyu.edu.hk} (W. Fan); \EMAIL{weihua.gu@polyu.edu.hk} (W. Gu)}
}

\ABSTRACT{A general Continuum Approximation (CA) model is proposed for optimizing transit network designs (TND) in grid cities under spatially heterogeneous demand. While conventional studies often assume rigid geometric line configurations (e.g., unbranched orthogonal grids), our framework allows the grid bus lines to route more flexibly by making lateral movements and to form network configurations with line detouring, merging, and diverging. The resulting line and stop densities, as well as service headways, vary continuously across both directions of the city—constrained solely by vehicle flow conservation. By respecting non-uniform demand distributions, our heterogeneous networks substantially enlarge the class of heterogeneous network designs that can be represented and optimized within a tractable CA framework. 

To efficiently solve the optimization problem, we develop a sequential geometric programming framework that transforms the model into a sequence of standard geometric programming problems. Numerical experiments validate the accuracy of the proposed model and the solution method by comparing system metrics estimated by the CA models against the actual values computed from the discretized network designs. Under representative spatially heterogeneous demand scenarios, comparisons demonstrate that our model effectively reduces generalized costs by over 7\% against existing homogeneous and restricted heterogeneous TND models. Key findings indicate that: (i) the proposed framework consistently outperforms these conventional counterparts across all tested scenarios; (ii) the fully heterogeneous structure becomes particularly advantageous when patron demand exhibits strong spatial heterogeneity; and (iii) these flexible designs yield the greatest benefits in high-demand, low-wage, and large-area cities. 
}%



\maketitle
%

\section{Introduction}\label{intro} 
As a basic component of public transportation planning, transit network design (TND) aims to optimally determine the physical layout and service frequencies of routes to meet travel demand, laying the foundation for subsequent operational design and control \citep{cederBusNetworkDesign1986}. In doing so, considerable research effort has been devoted to developing a variety of mathematical programming models involving different levels of design detail, as well as their solution methods (see the detailed reviews by \citealp{kepaptsoglouTransitRouteNetwork2009,ibarra-rojasPlanningOperationControl2015}). Nonetheless, for most of these models, the complexity rapidly grows with voluminous inputs, e.g., discrete origin-destination matrices and street network details. These mathematical problems, often formulated as nonlinear integer programming models, are NP-hard and rely on heuristic solution algorithms that incur high computational costs and offer no guarantee of global optimality \citep{asadibagloeeTransitnetworkDesignMethodology2011}.

Continuum Approximation (CA) methods have also been widely used for optimizing transit network designs (e.g. \citealp{daganzoStructureCompetitiveTransit2010,chenOptimalTransitService2015,fanOptimalDesignIntersecting2018a}), in which models’ inputs and outputs were expressed as continuous quantities, rather than the discrete OD table and links and nodes in graphs. The TND problem formulated in the CA fashion contains only a few design variables and functions, e.g., line and stop densities as functions of spatial coordinates and service headways as functions of time and space, and can often be solved to global or near-global optimality via analytical or numerical techniques. Most works in this realm, however, have focused on corridor designs (e.g., \citealp{meiPlanningSkipstopTransit2021}) or simple network designs with uniform or axisymmetric demand and symmetric network layouts. Therefore, they cannot apply to the TND problem for real cities with spatially heterogeneous demand patterns \citep{daganzoStructureCompetitiveTransit2010,badiaCompetitiveTransitNetwork2014,sivakumaranAccessChoiceTransit2014a,chenOptimalTransitService2015}. 

Among the few studies that addressed non-uniform demand with heterogeneous network layouts, the pioneering work is by \cite{ouyangContinuumApproximationApproach2014}, who proposed a grid network design where local lines are added to increase the service frequency and line density in high-demand areas. These local lines are connected with the trunk lines that travel through the city. Nonetheless, this design only has limited flexibility, since it features equally spaced trunk lines with a constant headway. For cities with ring-radial street networks, \cite{vaughanOptimumPolarNetworks1986} proposed a design where the radial lines can diverge and merge anywhere along the route. The design also allows for non-uniform spacings between the ring lines and varying headways. However, Vaughan’s model is flawed as it does not consider the vehicle flow conservation constraint along the radial lines.

In light of the above, this paper aims to develop a more general CA model that accepts spatially nonuniform demand as input and that describes heterogeneous network designs for cities where the streets form a dense grid network. In our design, all lines can be routed laterally, allowing merging and diverging to occur anywhere within the city. Unlike the design in \cite{ouyangContinuumApproximationApproach2014}, such merging and diverging operations are not subject to the power-of-two constraint; therefore, line densities can be flexibly specified at any location. This flexibility is made tractable by the cumulative-flow representation, which provides a measurable and optimizable characterization of arbitrary local merging and diverging decisions. Overall, this work contributes to addressing more flexible TND problems arising in real cities through parsimonious and efficient CA models.

This modeling flexibility comes at the price of additional formulation complexity. In particular, allowing lines to merge, diverge, and make lateral detours anywhere in the city requires us to explicitly characterize the induced vehicle detour distance and the associated spatial redistribution of line flows. These effects lead to nontrivial expressions with absolute-value operators in the CA formulation; see Section \ref{sec_veh_detour} for the detailed modeling. To solve the resulting nonconvex optimization problem, we develop a sequential geometric programming (SGP) solution approach that approximates the original model as a series of standard geometric programming (GP) problems. While GP has been widely used in engineering optimization, its application to continuum approximation-based TND models remains limited in the existing literature \citep{zhenJointOptimizationTrunkfeeder2025}. The proposed SGP framework, therefore, provides an effective and computationally tractable approach for solving complex CA-based TND problems.

Numerical experiments reported later confirm both the accuracy and the practical value of the proposed framework. Compared with the corresponding discretized network evaluations, the CA model keeps the approximation error of generalized cost below 3\%, while the proposed HetNet design achieves 7.03\%--10.43\% lower generalized costs than the benchmark networks under strongly heterogeneous demand patterns.

The remainder of this paper is organized as follows. The next section introduces the fundamental concepts and assumptions. Section \ref{sec_method} presents the modeling framework, followed by the solution algorithm in Section \ref{sec_solution}. Numerical experiments are reported in Section \ref{sec_numerical}, and the final section concludes the paper.

\section{Basic concepts and assumptions} \label{sec_concept}
This paper examines TND in a square city, as outlined in Sections \ref{sec_HetNet} and \ref{sec_assumptions} for the network and demand, respectively. The proposed framework can also be directly applied to rectangular cities with minor adjustments. The notations utilized in this paper are summarized in Appendix \ref{Appen_A}.

\subsection{The city form and transit networks} \label{sec_HetNet}
We begin by considering a square city with dense grid street networks, as shown in Figure \ref{fig_HetNet}. 
Its transit network is assumed to follow the grid pattern accordingly. For convenience, we label the transit lines by their operational direction $i \in \mathbf{I}$, where set $\mathbf{I}$ includes E (\underline{E}astbound), W (\underline{W}estbound), N (\underline{N}orthbound), and S (\underline{S}outhbound). 

While routing predominantly along one direction (termed the longitudinal direction in this paper), the transit lines of that direction can deviate laterally at any location. These lateral movements enable the lines to redistribute themselves along the lateral direction, sometimes by merging or diverging at certain points. For instance,  E lines stretch mainly eastbound across the entire city, but some may turn perpendicular (northbound or southbound) at certain locations before resuming the easterly route, allowing the inter-line spacings to vary over space; see the green lines in Figure \ref{fig_HetNet}. In the present paper, this type of layout is called the Heterogeneous Network, or HetNet for short.
\begin{figure}[htbp]
    \centering
    \includegraphics[width=0.45\linewidth]{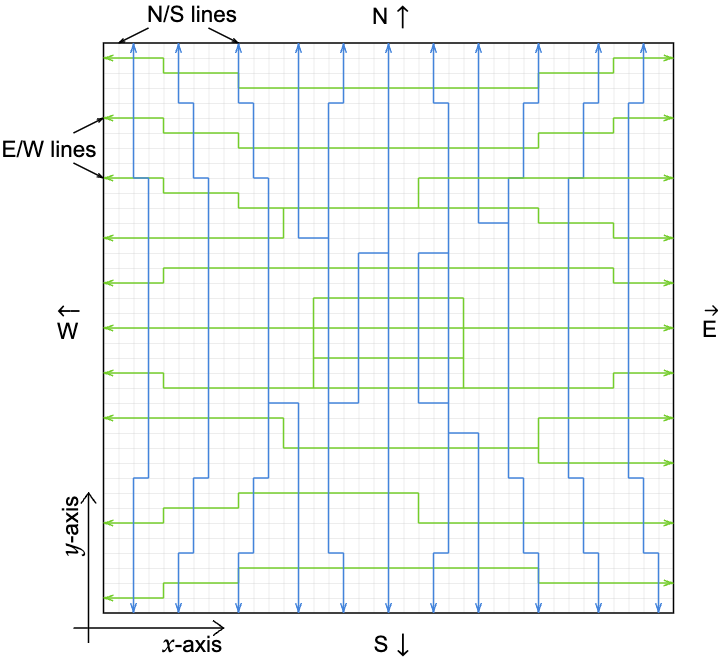}
    \caption{Illustration of heterogeneous transit network (HetNet) in a square city with a grid street network.}
    \label{fig_HetNet}
\end{figure} 

\begin{figure}[htbp]
    \centering
    \subfigure[H-HetNet \citep{ouyangContinuumApproximationApproach2014}]{
        \includegraphics[width=0.5\linewidth]{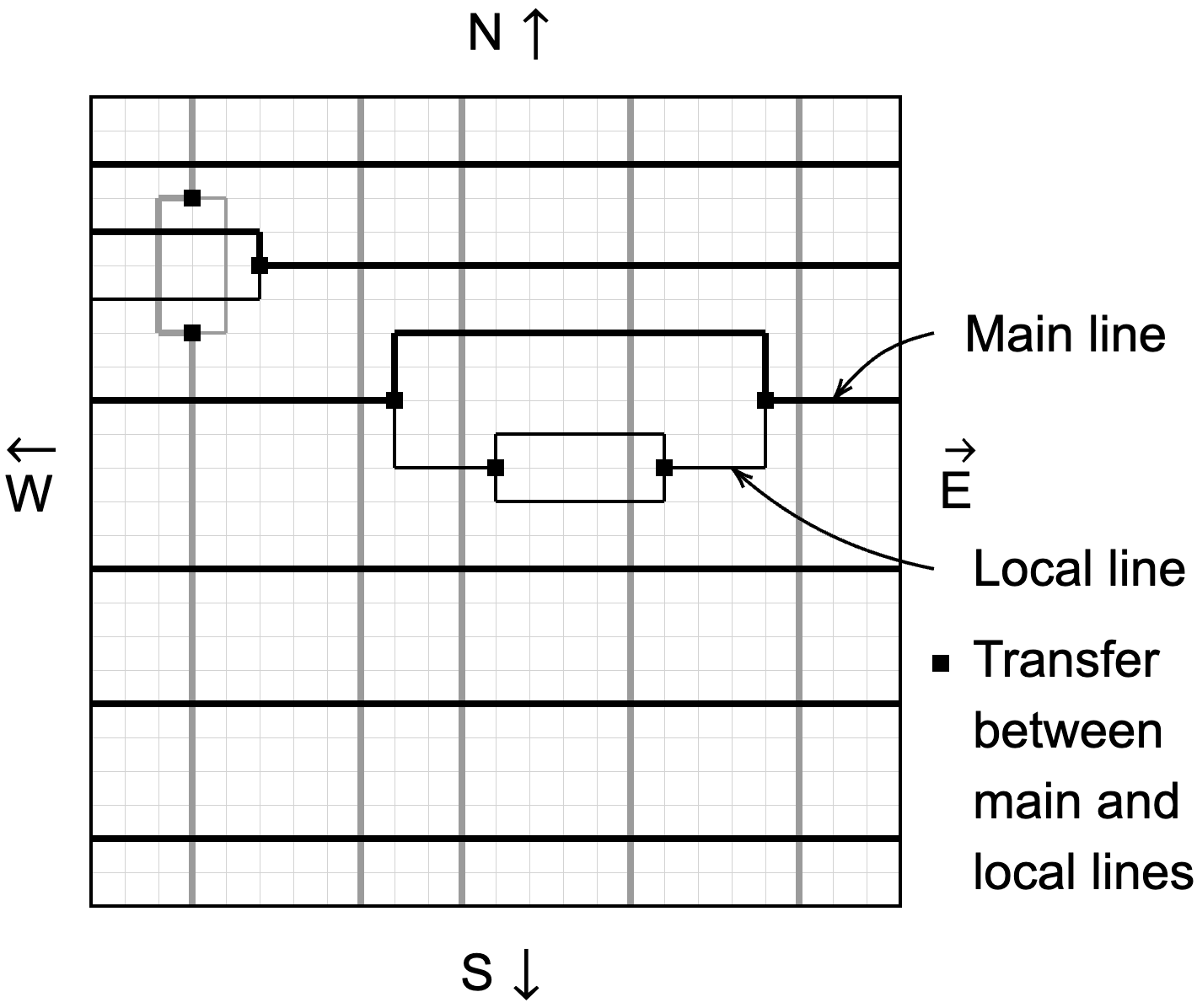}
        \label{Fig_H_HetNet}
    }
    \subfigure[P-HetNet \citep{chienOptimizationGridTransit1997}]{
        \includegraphics[width=0.41\linewidth]{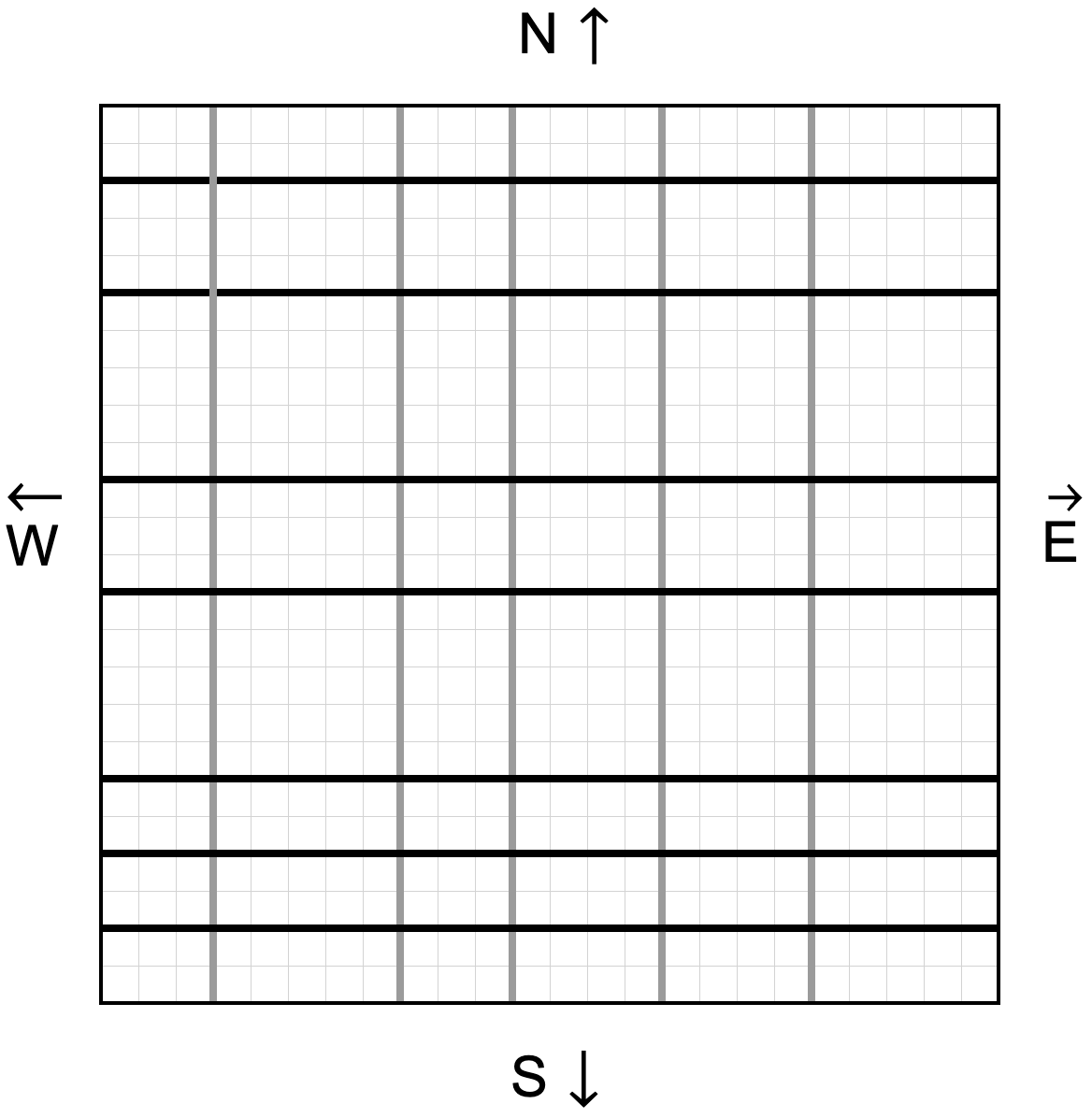}
        \label{Fig_P_HetNet}
    }
    \caption{Other heterogeneous grid transit network designs.}
    \label{fig_other_HetNet}
\end{figure}

To clarify ideas, we contrast HetNet with \cite{ouyangContinuumApproximationApproach2014}'s layout, hereafter called the Hierarchical HetNet or H-HetNet for brevity, and \cite{chienOptimizationGridTransit1997}'s network layout, called the Partial HetNet or P-HetNet; see Figures \ref{Fig_H_HetNet} and \ref{Fig_P_HetNet} for illustration. They are detailed as follows.
\begin{itemize}
    \item H-HetNet, as shown in Figure \ref{Fig_H_HetNet}, is a hierarchical network with two types of transit lines, namely the main lines spanning the entire city and the local lines operating in parallel within certain local zones only. They meet at transfer stations. The local lines can further bifurcate into multi-level local lines in the hierarchy. 
    Consequently, in the {H-HetNet}, a trip may entail transfers between hierarchical lines operating in the same direction. \\
    On the contrary, all lines in our HetNet span the entire city, beginning from one city boundary and ending at the opposite boundary. They may merge into common line segments (i.e., line segments that overlap on the same street) to provide higher-frequency, larger-spacing services along the predominant direction, or diverge to reduce access distances with increased headways in certain areas. 
    \item P-HetNet, shown in Figure \ref{Fig_P_HetNet}, can be viewed as a special case of H-HetNet without local lines or our HetNet without lateral movements, where $i$ lines' spacings can vary in the perpendicular direction but remain invariant along direction $i$.
\end{itemize}

For design purposes, we describe HetNet using two sets of functions: (i) \textit{transit line density} of direction $i \in \mathbf{I}$, denoted $\delta_i(\vec{x})$, meaning the number of {laterally separated lines} (or the number of streets where transit lines go through) per unit of lateral distance at the location denoted by $\vec{x}$; here, $\Vec{x}$ is a two-dimensional vector marking the location in the city domain $\mathbf{R}$; for instance in the Cartesian coordinate system, $\delta_{\rm E}(\vec{x})$ means the number of separated E-lines per $y$-axis distance at $\vec{x} = (x,y), x,y \in [0, R]$, where $R$ denotes the side length of the square city, i.e., $\mathbf{R}=[0,R] \times [0,R]$; and (ii) \textit{service headway} of a common-line segment that traverses location $\vec{x}$ in direction $i \in \mathbf{I}$, $h_i(\vec{x})$.
We assume that the two classes of decision functions, $\delta_i(\vec{x})$ and $h_i(\vec{x})$, are integrable, as is common in the literature of CA models \citep{ouyangContinuumApproximationApproach2014}.

Note that changes in transit line densities and headways reflect fuller heterogeneous network designs, and the two design functions enjoy a degree of design freedom at local areas: Specifically, line spacing expands from high- to low-density areas and vice versa, while buses converge onto common segments from large- to small-headway areas and vice versa. 
This \textit{local design freedom} means line densities and headways can vary independently in local areas, provided total vehicle flows are conserved for lines in the same direction across the entire network. This property of HetNet allows us to produce networks that best fit the needs of spatially non-uniform demand. However, detours happen for lines with small changes in either line densities or headways in neighboring areas. Therefore, the key tradeoff in optimizing HetNet lies between the benefits of heterogeneous line spacings and headways and the associated detour costs for both the operator and passengers. 

We further specify that the network is \textit{bidirectionally symmetrical}, i.e., $\delta_{\rm E}(\Vec{x}) = \delta_{\rm W}(\Vec{x})$ and $\delta_{\rm N}(\Vec{x}) = \delta_{\rm S}(\Vec{x}), \forall \Vec{x} \in \mathbf{R}$. In other words, each E (N) line segment overlaps with a W (S) line segment. We believe this is a reasonable choice in practice for a clear network image with bidirectional accessibility. This symmetrical network also benefits us by simplifying the following model formulations. Although bidirectional symmetry is enforced for the network, we allow \textit{asymmetrical headways} between the two opposing directions on the same line to adapt to unbalanced demand during various operation hours. For fully asymmetrical networks, we leave them for future study. 


In addition, we consider that stops are located at the intersections of the two perpendicular sets of lines, and vehicles do not skip stops. At each stop, a vehicle experiences delays due to acceleration, deceleration, passenger boarding, and alighting. We assume that all the lines merging into a common line segment share the same set of stops in that segment. Patrons can transfer between perpendicular lines at these stops. For instance, a patron can transfer via an E-line stop to either ${\rm N}$ or ${\rm S}$ lines. 
Intermediate stops between transfer stops are ignored in this paper and left for future exploration. 

\subsection{Transit demand and patrons' travel behaviors} \label{sec_assumptions}
The transit demand is assumed to be continuously distributed across the city region \citep{daganzoPublicTransportationSystems2019}. We express the demand as an integrable density function $\lambda(\Vec{x}_o,\Vec{x}_d )$ ($\rm trip \cdot km^{-4} \cdot hr^{-1}$), where $\Vec{x}_o$ and $\Vec{x}_d$ denote the origin and destination coordinates, respectively. Their travel behaviors are assumed to follow a few principles as adopted by previous studies (e.g., \citealp{daganzoPublicTransportationSystems2019, ouyangContinuumApproximationApproach2014}): 
\begin{itemize}
    \item The closest-stop principle: Each patron walks to/from the stop closest to her origin/destination \citep{fanOptimalDesignIntersecting2018a}. 
    \item The random-arrival principle: A patron arrives at the origin stop randomly without consulting the schedule of transit lines \citep{daganzoPublicTransportationSystems2019}. 
    \item The first-arrival vehicle principle (e.g., \citealp{chriquiCommonBusLines1975a}): 
    \begin{itemize}
        \item At her origin stop, a patron would board the first-arrival vehicle of all common lines that connect her destination with at most one transfer. 
        \item Onboard a vehicle, she decides to transfer to the perpendicular line with the first vehicle arrival times of all perpendicular lines that pass her destination. (This assumption implies that some intelligent transit systems can aid onboard patrons' transfer decisions with good quality information.) 
    \end{itemize}
\end{itemize}


Based on the above assumptions, the high-dimensional demand density function can be integrated into local lower-dimensional demand functions at location $(\Vec{x})$. They include the boarding and alighting densities, $\lambda_{\rm bo}^{i} (\Vec{x})$ and $\lambda_{\rm al}^{i} (\Vec{x})$ ($\rm trip \cdot km^{-2} \cdot hr^{-1}$), the cross-sectional on-board flux, $\lambda^{i}_{\rm fl} (\Vec{x})$ ($\rm trip \cdot km^{-1} \cdot hr^{-1}$), and the transferring densities, $\lambda_{\rm tr}^{i} (\Vec{x})$ ($\rm trip \cdot km^{-2} \cdot hr^{-1}$), $ \forall i \in \mathbf{I}$. Detailed derivations of these local demand functions are given in Appendix \ref{Appen_B}.

\section{Methodology} \label{sec_method}
Section \ref{sec_3_1} first presents some preliminary results that facilitate the modeling in Sections \ref{sec_3_2} and \ref{sec_3_3}, which describe the optimization formulation and the cost models, respectively.

\subsection{Preliminaries}\label{sec_3_1}
{In preparation for presenting the HetNet design models, we estimate three key metrics of HetNet: vehicles' detour distance in Section \ref{sec_veh_detour}, patrons' in-vehicle detour time in Section \ref{sec_in_vehicle_time}, and out-of-vehicle wait time in Section \ref{sec_out_vehicle_time}.} The Cartesian coordinate system is used for the ease of derivations. 

\subsubsection{Vehicle's detour distance.}\label{sec_veh_detour}
To estimate vehicles' detour distance, we define $q_i(x,y)=\frac{\delta_i(x,y)}{h_i(x,y)}$ ($\rm veh \cdot hr^{-1} \cdot km^{-1}$) be the vehicle flow of $i$ lines per cross-sectional distance, and $Q_{i}(x,y)$ ($\rm veh \cdot hr^{-1}$) the cumulative cross-sectional vehicle flow of $i$ lines. 

Specifically, E/W lines have $Q_{\rm E/W}(x,y)=\int_{y'=0}^{y}q_{\rm E/W}(x,y') \diff y' $ defined at the cross-section between $(x,0)$ to $(x,y)$; and N/S lines have $Q_{\rm N/S}(x,y)=\int_{x'=0}^{x}q_{\rm N/S}(x',y) \diff x' $ between $(0,y)$ and $(x,y)$. By the definition of HetNet in Section \ref{sec_HetNet}, these $\{Q_i(x,y), i \in \mathbf{I}\}$ satisfy conservation law. This is expressed by (with a slight abuse of notation): $Q_i(x,y|y=R) \equiv \mathrm{Q}_i, \forall x \in [0,R], i \in \{{\rm E,W}\}$ and $Q_j(x,y|x=R) \equiv \mathrm{Q}_j, \forall y \in [0,R], j \in \{{\rm N,S}\}$, where $\mathrm{Q}_{i}$ ($\mathrm{Q}_{j}$) is the total dispatched vehicle flow of $i$ ($j$) lines in the entire city.

Given the above definitions, we have the following proposition.

\begin{proposition} \label{Prop_Dv}
Total vehicle detour distance across the HetNet can be estimated by 
\begin{subequations}
    \begin{align}
        \int_{x=0}^{R}\int_{y=0}^{R}\left[\sum_{i\in\left\{ {\rm E,W}\right\} }\left|\frac{\partial Q_{i}(x,y)}{\partial x}\right|+\sum_{j\in\left\{ {\rm N,S}\right\} }\left|\frac{\partial Q_{j}(x,y)}{\partial y}\right|\right] \diff y \diff x, \notag
    \end{align}
    which is simplified into
    \begin{align}
        \sum_{i \in \mathbf{I}} \iint_{\Vec{x} \in \mathbf{R}} d_i(\Vec{x}) \diff \Vec{x},
    \end{align}
    where the integrand $\{d_i(\Vec{x}) \}$ $\rm (veh \cdot hr^{-1} \cdot km^{-1})$ is defined by 
    \begin{align} \label{eq_di}
        d_i(x,y) = \left|\frac{\partial Q_{i}(x,y)}{\partial x}\right|, i \in \left\{ {\rm E,W}\right\},d_j(x,y) = \left|\frac{\partial Q_{j}(x,y)}{\partial y}\right|, j \in \left\{ {\rm N,S}\right\}.
    \end{align}
\end{subequations}


\end{proposition}
\proof{Proof.} 
    We prove Proposition \ref{Prop_Dv} by deriving the E-line vehicles' detour distance for illustration. 
    
    Firstly, we observe Figure \ref{fig_veh_detour_a} for two cross sections at $x$ and $x+\Delta x$. The passing vehicle flows are indexed by 1,2,...,$Q$,...,$\mathrm{Q}_{\rm E}$ in the $y$-axis direction. Note that the (perpendicular) detour distance of $Q^{\rm th}$ flow is measured by $\left| y^Q_{x+\Delta x} -  y^Q_{x} \right|$, where $y^Q_{x+\Delta x}$ and $y^Q_{x}$ mark $Q^{\rm th}$ vehicle flow's $y$-axis locations at the two cross sections, respectively, satisfying $Q_{\rm E}(x+\Delta x, y|y=y^Q_{x+\Delta x})=Q_{\rm E}(x, |y=y^Q_{x})=Q$.

    \begin{figure}[htbp]
    \centering
    \subfigure[Cross-sectional vehicle flows at $x$ and $x+\Delta x$.]{
        \includegraphics[width=0.45\linewidth]{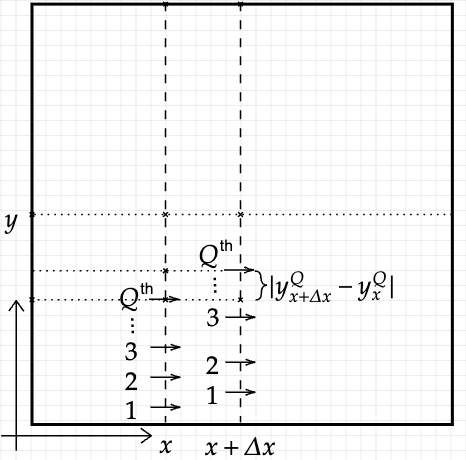}
        \label{fig_veh_detour_a}
    }
    \subfigure[Cumulative vehicle flows in $y$-axis direction at $x$ and $x+\Delta x$.]{
        \includegraphics[width=0.45\linewidth]{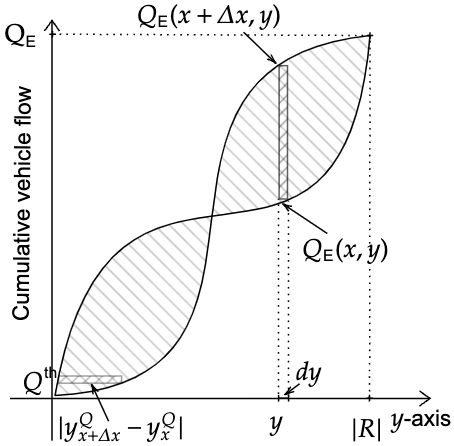}
        \label{fig_veh_detour_b}
    }
    \caption{Vehicle detour distance derivation for E lines.}
    \label{fig_veh_detour}
    \end{figure}
    
    Next, we turn to Figure \ref{fig_veh_detour_b}, which plots cumulative vehicle flows passing the cross-sections at $x$ and $x+\Delta x$, i.e., the curves of $Q_{\rm E}(x,y)$ and $Q_{\rm E}(x+\Delta x,y)$ concerning $y \in [0,R]$. The horizontal distance in Figure \ref{fig_veh_detour_b} corresponds to the perpendicular detour distance of E-line vehicles; c.f. Figure \ref{fig_veh_detour_a}. 
    Thus, integrating the horizontal slices with areas of $\left| y^Q_{x+\Delta x} - y^Q_{x} \right| \diff Q$ over $Q\in [0,\mathrm{Q}_{\rm E}]$ yields the total vehicle detour distance of E lines between $x$ and $x+\Delta x$, i.e., the shaded area in Figure \ref{fig_veh_detour_b}.

    Equivalently, the shaded area can be obtained by integrating the vertical slices between $y$ and $y+dy$, as shown in Figure \ref{fig_veh_detour_b}, for $y\in [0, R]$. The formula is
    \begin{align} \label{eq_detour_dx}
        \int_{y=0}^{R}\left|Q_{\rm E}(x+\Delta x,y) - Q_{\rm E}(x,y) \right| \diff y = \int_{y=0}^{R}\left|\frac{Q_{\rm E}(x,y)-Q_{\rm E}(x+\Delta x,y)}{\Delta x}\right| \diff y \Delta x.
    \end{align}
    
    Finally, when $\Delta x \rightarrow 0$, we integrate Eq. (\ref{eq_detour_dx}) over $x \in [0, R]$, obtaining  the total vehicle detour distance of E lines in the entire HetNet, i.e.,
    \begin{align}
        \int_{x=0}^{R}\int_{y=0}^{R}\left|\frac{\partial Q_{\rm E}(x,y)}{\partial x}\right| \diff y \diff x,
    \end{align}
    where $\left|\frac{\partial Q_{\rm E}(x,y)}{\partial x}\right|$ ($\rm veh \cdot hr^{-1} \cdot km^{-1}$) can be interpreted as the density of detouring vehicle flow per E-line distance; $\left|\frac{\partial Q_{\rm E}(x,y)}{\partial x}\right| \diff x$ yields the number of detouring vehicles in $[x,x+ \diff x]$; and $\left|\frac{\partial Q_{\rm E}(x,y)}{\partial x}\right| \diff x \diff y$ gives vehicle detour distance in the neighborhood area ($\diff x \diff y$) at $(x,y) \in \mathbf{R}$.
    
    The results of W lines can be similarly obtained as
    \begin{align}\label{Eq_vehicle_detour_distance_EW}
        \int_{x=0}^{R}\int_{y=0}^{R}\left|\frac{\partial Q_{\rm W}(x,y)}{\partial x}\right|\diff y \diff x.
    \end{align}
    
    The above derivation process is applied to N and S lines by rotating the coordinate system, which yields,
    \begin{align}\label{Eq_vehicle_detour_distance_NS}
        \int_{x=0}^{R}\int_{y=0}^{R}\left|\frac{\partial Q_{j}(x,y)}{\partial y}\right| \diff y \diff x, \forall j \in \{ {\rm N, S} \}.
    \end{align}
\endproof

\subsubsection{Patrons' in-vehicle detour time.} \label{sec_in_vehicle_time}

Based on Proposition \ref{Prop_Dv}, we have the following corollary,
\begin{corollary} \label{corollary_Dt}
    Total patron in-vehicle detour time can be estimated by 
    \begin{align} 
        \sum_{i \in \mathbf{I}} \iint_{\Vec{x} \in \mathbf{R}} \alpha \frac{\lambda^i_{\rm fl}(\Vec{x})h_i(\Vec{x})}{\delta_i(\vec{x})} \frac{d_i(\Vec{x})}{v} \diff \Vec{x},
    \end{align}
    where $v$ is the transit vehicles' average cruising speed; and parameter $\alpha\in[0,1]$ accounts for the impact of onboard patrons' detours on their in-vehicle trip distance in a uniform grid transit network, which may increase or decrease due to detouring further away or closer in the direction of their destinations. 
\end{corollary}

\begin{remark}
    When $ d_i(\vec{x})=0$, patrons experience no additional in-vehicle distance due to lateral detours. When $d_i(\vec{x})>0$, the parameter $\alpha \in [0,1]$ captures the effective fraction of vehicle detour distance that translates into additional patron in-vehicle distance, shown in Figures \ref{fig_detour_parameter}. The case $\alpha = 0$ represents the situation in which lateral vehicle detours introduce no additional patron detour distance on average, whereas $\alpha = 1$ represents the conservative case in which the full vehicle detour distance is counted as additional patron detour distance.
\end{remark}

\begin{figure}[htbp]
    \centering
    \subfigure[No additional patron detour distance ($\alpha=0$)]{
        \includegraphics[width=0.48\linewidth]{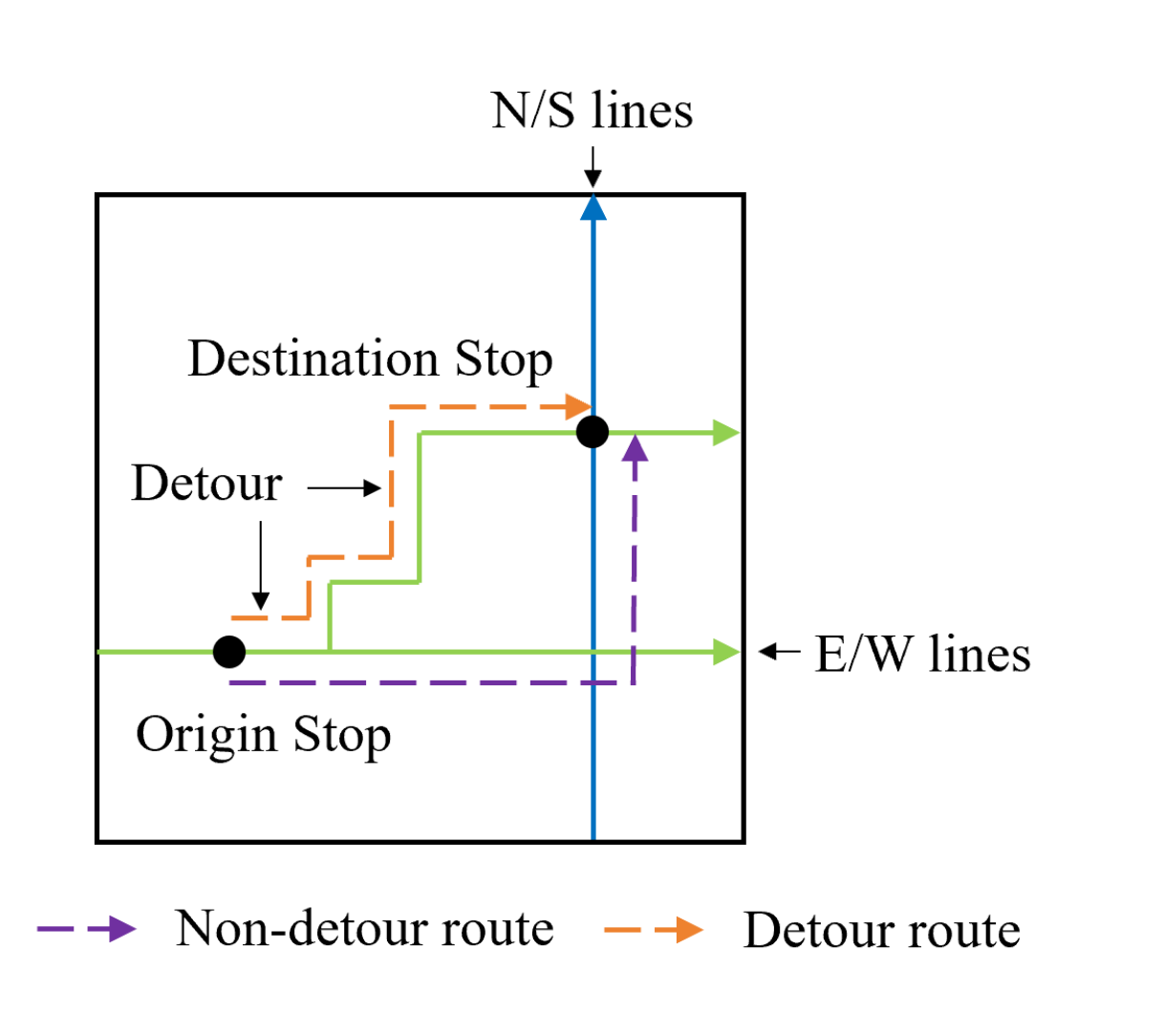}
    }
    \subfigure[Maximum additional patron detour distance ($\alpha=1$)]{
        \includegraphics[width=0.48\linewidth]{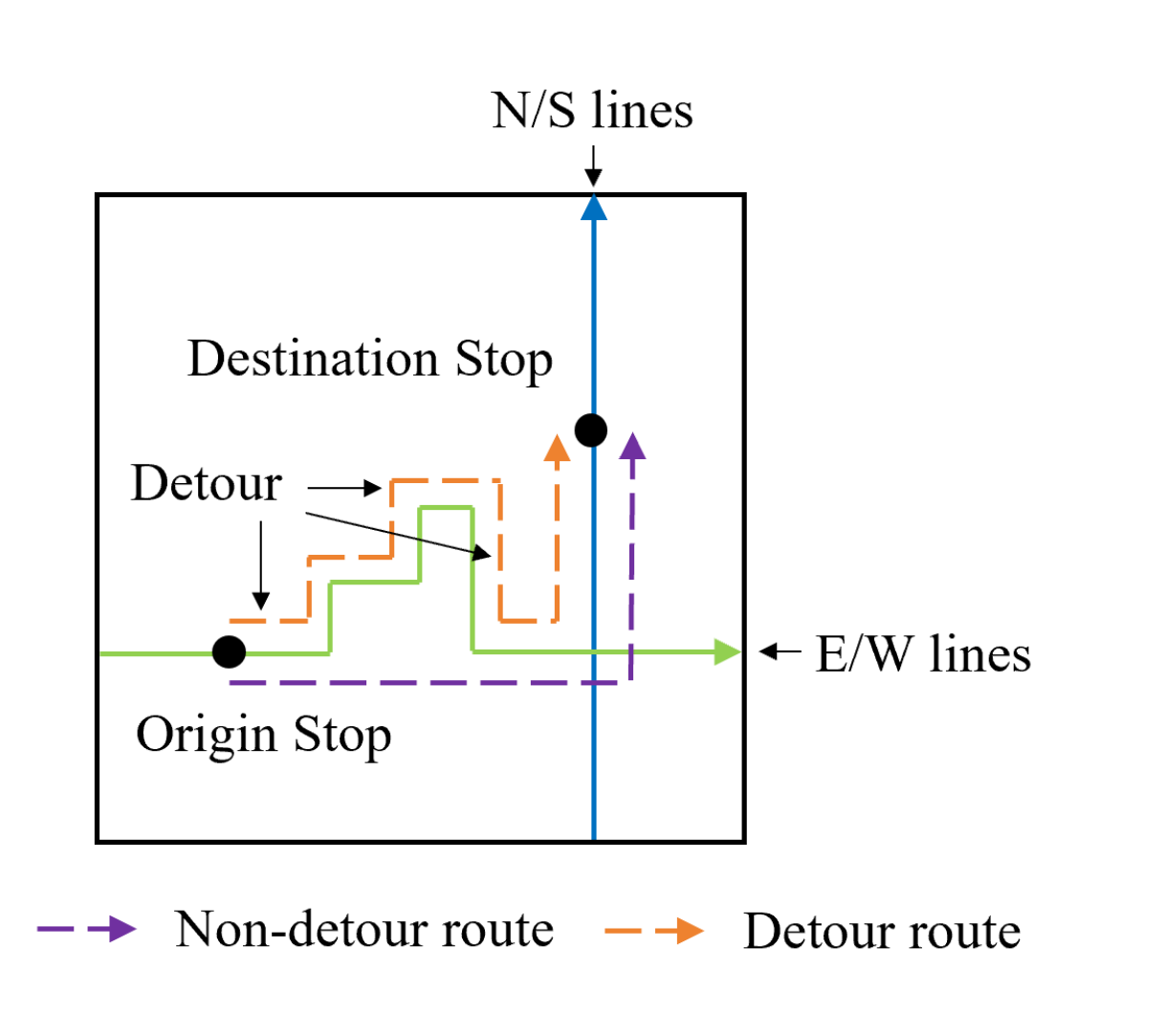}
    }
    \caption{Illustration of detour parameter $\alpha$, which measures the fraction of vehicle detour distance that contributes to additional patron in-vehicle distance.}
    \label{fig_detour_parameter}
\end{figure}

It is worth noting that $\alpha$ is essentially an endogenous parameter, because its realized value depends on the network structure and the resulting demand-flow assignment. As both of these are determined only after the network design is specified, $\alpha$ cannot be estimated a priori during the optimization process. Therefore, in the network design stage, we first fix a candidate value of $\alpha$ and then solve the corresponding optimization problem. In Section \ref{sec_Sensitivity_analysis_of_the_endogenous_parameter}, we further examine how different values of $\alpha$ affect the optimization results and provide a recommended setting for practical use.

\proof{Proof.} Corollary \ref{corollary_Dt} can be proven by noticing that $\frac{f^i_{\rm fl}(\Vec{x})h_i(\Vec{x})}{\delta_i(\vec{x})}$ means average vehicle occupancy. Multiplying this by the local detour distance of all vehicles, $d_i(\Vec{x}) \diff \Vec{x}$, and dividing it by vehicle speed $v$, would yield local patrons' (non-discounted) in-vehicle detour time.
\endproof


Similarly, we have the second corollary that is useful for deriving transit agency cost metrics,
\begin{corollary} \label{corollary_Dl}
    The detoured length of transit line infrastructures (e.g., the guideway) in the HetNet can be estimated by 
    \begin{align}
        \sum_{i \in \mathbf{I}} \iint_{\Vec{x} \in \mathbf{R}} d_i(\Vec{x})h_i(\Vec{x}) \diff \Vec{x}.
    \end{align}
\end{corollary}

\proof{Proof.} Corollary \ref{corollary_Dl} can be proved by noticing that $d_i (\Vec{x}) h_i(\Vec{x}) \diff \Vec{x}$ [$(\rm veh \cdot hr^{-1} \cdot  km^{-1}) \cdot ( hr \cdot veh^{-1}) \cdot km^2$] means the vehicle detour distance per vehicle of $i$ lines in the neighborhood area at $(x,y)$, which is equivalent to the detoured length of local $i$-line infrastructures.
\endproof

\subsubsection{Patrons' out-of-vehicle wait time.} \label{sec_out_vehicle_time}

During transit trips, patrons experience out-of-vehicle waiting at their origin and transfer stops. The first-arrival vehicle principle (as stated in Section \ref{sec_assumptions}) enables us to localize patrons' out-of-vehicle wait time concerning only local variables and estimate the total by the following proposition.

\begin{proposition} \label{Prop_Wt}
    Patrons' out-of-vehicle wait time in the HetNet can be estimated by,
    \begin{align} \label{eq_Tw}
        \sum_{i \in \mathbf{I}} \iint_{\Vec{x} \in \mathbf{R}} \left( \lambda^i_{\rm bo}(\Vec{x}) + \lambda^i_{\rm al}(\Vec{x}) \right) \frac{h_i(\Vec{x})}{2} \diff \Vec{x},
    \end{align}
    where the integrand $\left( \lambda^i_{\rm bo}(\Vec{x}) + \lambda^i_{\rm al}(\Vec{x}) \right) \frac{h_i(\Vec{x})}{2}, \forall i \in \mathbf{I}, \Vec{x} \in \mathbf{R}$ $(\rm trip \cdot km^{-2} \cdot hr^{-1})$ includes localized wait times for patrons boarding at $\Vec{x}$ (of which the number is $ \lambda^i_{\rm bo}(\Vec{x}) \diff \Vec{x} $) and those transferring from somewhere else but alighting at $\Vec{x}$ (of which the number is $ \lambda^i_{\rm al}(\Vec{x}) \diff \Vec{x} $).
\end{proposition}

\proof{Proof.}
We prove Proposition \ref{Prop_Wt} by first observing patrons' typical trip patterns, as visualized in Figure \ref{fig_trip_patterns}, including branch-to-branch, branch-to-trunk, trunk-to-branch, and trunk-to-trunk trips. (Other trip patterns are either simpler ones, e.g., single-line-to-single-line, or can be composed of combinations of the four typical patterns; see the notes following Figure \ref{fig_trip_patterns}.)


\begin{figure}[htbp]
    \centering
    \subfigure[Branch-to-branch trips.]{
        \includegraphics[width=0.45\linewidth]{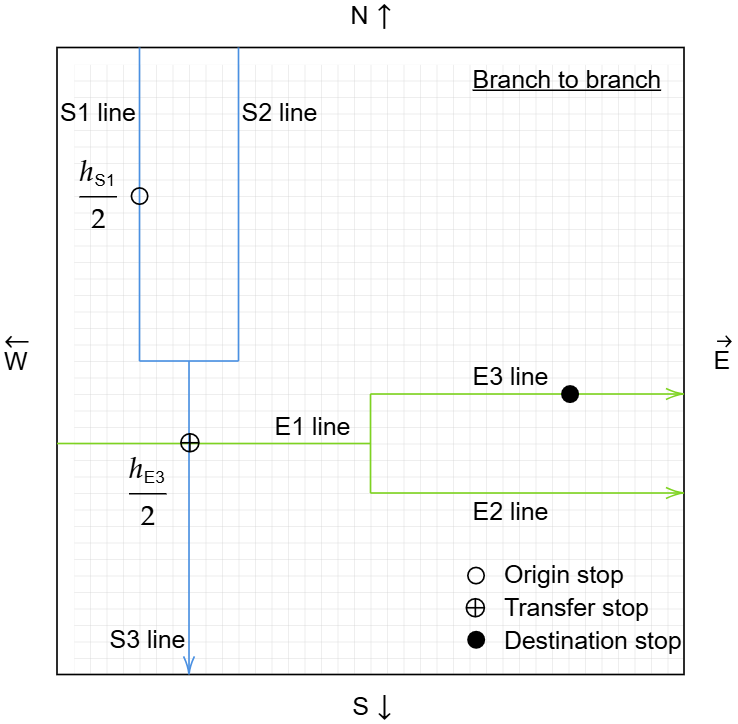}
        \label{fig_trip_patterns_a}
    }
    \subfigure[Branch-to-trunk trips.]{
        \includegraphics[width=0.45\linewidth]{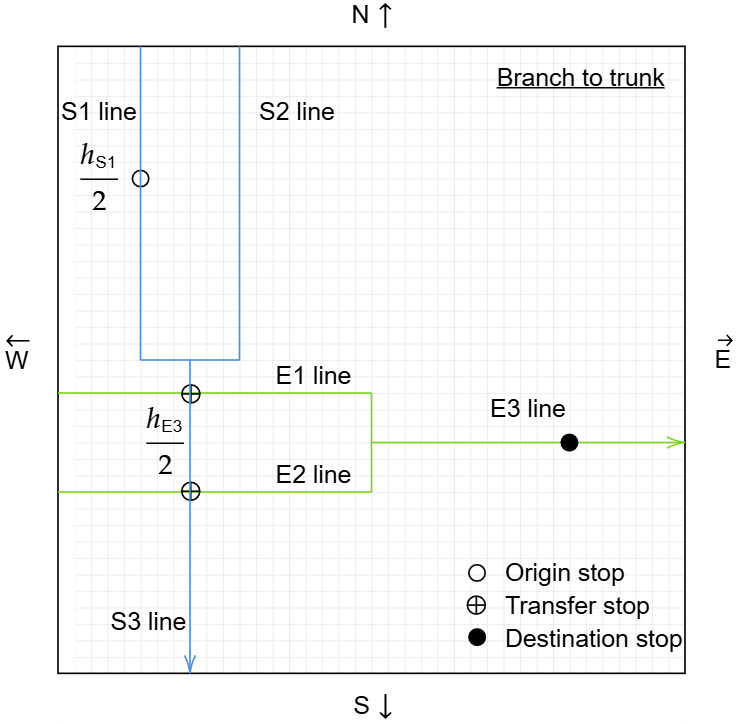}
        \label{fig_trip_patterns_b}
    }
    \subfigure[Trunk-to-branch trips.]{
        \includegraphics[width=0.45\linewidth]{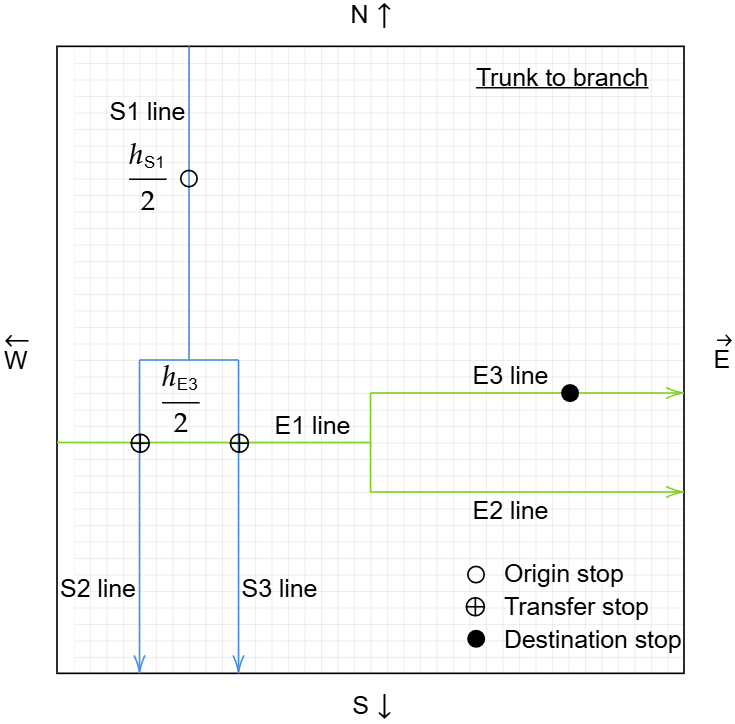}
        \label{fig_trip_patterns_c}
    }
    \subfigure[Trunk-to-trunk trips.]{
        \includegraphics[width=0.45\linewidth]{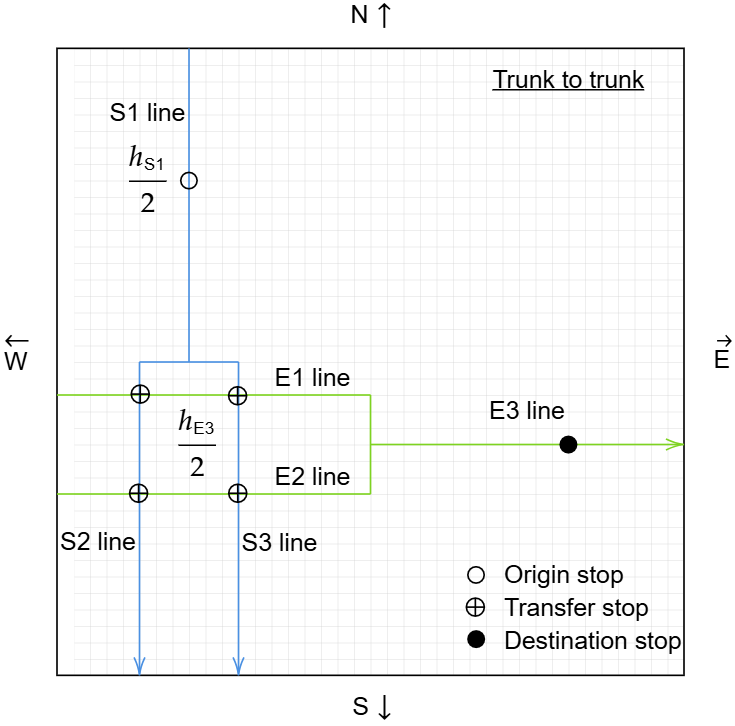}
        \label{fig_trip_patterns_d}
    }
    \caption{Typical trip patterns based on whether patrons' origins and destinations are on a branch/diverged line or a trunk/merged line.\textsuperscript{*}}
    \label{fig_trip_patterns}
    \vspace{0.1cm}
    \begin{minipage}{\linewidth}
    \footnotesize
    \textsuperscript{*}Note how patrons may transfer in these trip patterns: 
    \begin{itemize}
        \item In Figure \ref{fig_trip_patterns_a}, an average patron would board the first-arrival bus of any S1-line buses passing her origin stop, and she then would transfer at the intersected stop and board the first-arrival bus of any E3-line buses passing her destination stop. 
        \item In Figure \ref{fig_trip_patterns_b}, an average patron would experience the same at the origin stop, as in Figure \ref{fig_trip_patterns_a}, but may transfer at any of the branches of E3 lines. This is because the first-arrival bus going to her E3-line stop might appear at any E3-line branch. 
        \item In Figure \ref{fig_trip_patterns_c}, the first-arrival bus of the S1 line may carry the average patron to any of the S1-line branches, where she would board the first-arrival E3-line bus.
        \item In Figure \ref{fig_trip_patterns_d}, an average patron would exhibit boarding behaviors as in Figure \ref{fig_trip_patterns_c} and transferring behaviors as in Figure \ref{fig_trip_patterns_b}.
    \end{itemize}
    \end{minipage}
\end{figure}

Note that a trip may have multiple candidate transfer stops at all branched lines that serve the patron’s origin-destination pair; see Figures \ref{fig_trip_patterns_b}–\ref{fig_trip_patterns_d}. All of these are reasonable transfer options when the distances among branched lines are relatively short compared to the overall trip length, making the difference in the in-vehicle travel time negligible. As a result, the patron can transfer to any vehicle that will visit her destination stop. Under the first-arrival vehicle principle, she will choose the one among them that arrives first.

Since her options include all buses passing her origin and destination stops in all trip patterns, an average patron would experience half of the origin-stop headway (e.g., $ \frac{h_{\rm S1}}{2} $ in Figure \ref{fig_trip_patterns}) for the first boarding and half of the destination-stop headway ($ \frac{h_{\rm E3}}{2} $ in Figure \ref{fig_trip_patterns}) for transferring.
    
    
In other words, the average patron's wait time associated with an $i$-line stop at $\Vec{x}$ is $\frac{h_i(\Vec{x})}{2}$, either for boarding at $\Vec{x}$ or transferring from somewhere else but alighting at $\Vec{x}$. For the two cases, the numbers of patrons are $ \lambda^i_{\rm bo}(\Vec{x}) \diff \Vec{x}$ and $ \lambda^i_{\rm al}(\Vec{x}) \diff \Vec{x}$ , respectively. ({We conservatively assume that all patrons transfer to complete their trips in the grid-network city, ignoring no-transfer trips of those whose origins and destinations must align close to the same line.) Thus, the local patron wait time for $i$ lines at $\Vec{x}$ is 
\begin{align}
    \left( \lambda^i_{\rm bo}(\Vec{x}) + \lambda^i_{\rm al}(\Vec{x}) \right) \frac{h_i(\Vec{x})}{2} \diff \Vec{x}, \forall i \in \mathbf{I}, \Vec{x} \in \mathbf{R}.
\end{align}

Finally, integrating the local results of all lines for $\Vec{x} \in \mathbf{R}$ yields the total patron wait time, i.e., Eq. (\ref{eq_Tw}).
\endproof

\subsection{Network design problem}\label{sec_3_2}
Based on the above results, we quantify the performance of HetNet designs in terms of the system’s generalized cost, denoted by $Z$. This generalized cost includes two cost components, $Z_A$ and $Z_P$, experienced by the transit agency and the patrons, respectively. 

The generalized cost optimization problem is formulated as follows:
\begin{subequations} \label{Network_problem}
    \begin{align}
        \minimize_{ \{\mathrm{Q}_i, \delta_i(\Vec{x}),h_i(\Vec{x}),i \in \mathbf{I}, \Vec{x} \in \mathbf{R} \} } Z &= \frac{1}{\mu}Z_A + Z_P ,
    \end{align}
    subject to:
    \begin{align}
        & \int^{R}_{y=0} \frac{\delta_i(x,y)}{h_i(x,y)} \diff y = \mathrm{Q}_i, \forall i \in \{{\rm E,W}\}, x \in [0,R], \label{cons_flow_conservation_i}\\
        & \int^{R}_{x=0} \frac{\delta_j(x,y)}{h_j(x,y)} \diff x = \mathrm{Q}_j, \forall j \in \{{\rm N,S}\}, y \in [0,R], \label{cons_flow_conservation_j} \\
        & \lambda^i_{\rm fl}(\Vec{x})\frac{h_i(\Vec{x})}{\delta_i(\Vec{x})} \le C, \forall i \in \mathbf{I}, \Vec{x} \in \mathbf{R}, \label{cons_capacity} \\
        & \mathrm{Q}_i, \delta_i(\Vec{x}),h_i(\Vec{x})>0, \forall i \in \mathbf{I}, \Vec{x} \in \mathbf{R}, \label{cons_range}
    \end{align}
\end{subequations}
where $Z_A$ and $Z_P$ will be formulated in the next section, and $\mu$ ($\rm \$ \cdot hr^{-1}$) is the patrons' average value of time, which converts the monetary costs into the temporal unit. 

Constraints (\ref{cons_flow_conservation_i}, \ref{cons_flow_conservation_j}) indicate that the vehicle flows are conserved in the $i$ ($j$) directions, where $\mathrm{Q}_{i}$ ($\mathrm{Q}_{j}$) are scalar variables, representing the total vehicle dispatches of $i$ ($j$) lines per operation hour. 
Constraints (\ref{cons_capacity}) ensure the number of passengers onboard a transit vehicle never exceeds the vehicle capacity, $C$ ($\rm trip \cdot veh^{-1}$). 
Lastly, constraints (\ref{cons_range}) specify the valid ranges of decision functions.

\begin{remark}
    When $\delta_i(\Vec{x})$ of $i$ lines are regulated to be invariant in the operational direction (in other words, no detours is allowed), our model will reduce to lower-dimensional heterogeneous network designs (with decision functions being, e.g., $\delta_{\rm E} (y)=\delta_{\rm W} (y)$, $\delta_{\rm N} (x)=\delta_{\rm S} (x), \forall x, y \in [0, R] $, and $ \{\mathrm{Q}_i\} $ being redundant and dropped), which turns out to be P-HetNet (Chien and Schonfeld, 1997). Our model would ultimately reduce to the uniform network designs with scalar decision variables independent of locations, as in \cite{daganzoPublicTransportationSystems2019}. 
\end{remark}

\subsection{Generalized cost formulation}\label{sec_3_3}
This section presents detailed formulations for $Z_A$ and $Z_P$ in (\ref{Network_problem}).

\subsubsection{Agency cost.}
Transit agency cost, $Z_A$, is related to four network metrics: i) the line infrastructure length, $N_l$; ii) the number of stops, $N_s$; iii) the vehicle distance traveled, $N_k$; and iv) the vehicle time traveled, $N_h$. According to \cite{daganzoPublicTransportationSystems2019}, $Z_A$ can be computed by

\begin{subequations}\label{Network_ZA}
    \begin{align}
        Z_A & = \pi_l N_l + \pi_s N_s + \pi_k N_k + \pi_h N_h,
    \end{align}
    where parameters $\pi_l, \pi_s, \pi_k$, and $\pi_h$ are unit costs associated with the four metrics given below: 
    \begin{align}
        N_l & = \sum_{i\in \mathbf{I}} \iint_{\Vec{x} \in \mathbf{R}}  \left[\delta_i(\Vec{x}) + d_i(\Vec{x}) h_i(\Vec{x})\right]  \diff \Vec{x} , \label{eq_Nl} \\ 
        N_s & = \sum_{i\in \mathbf{I}} \iint_{\Vec{x} \in \mathbf{R}}  \delta_i (\Vec{x})  \delta_{\Bar{i}} (\Vec{x}) \diff \Vec{x}, \label{eq_Ns} \\
        N_k & = \sum_{i\in \mathbf{I}} \iint_{\Vec{x} \in \mathbf{R}}  \left[ \frac{\delta_i(\Vec{x})}{h_i(\Vec{x})}  + d_i(\Vec{x}) \right] \diff \Vec{x}, \label{eq_Nk} \\ 
        N_h & = \sum_{i\in \mathbf{I}} \iint_{\Vec{x} \in \mathbf{R}} \left[ \frac{\delta_i(\Vec{x})}{h_i(\Vec{x})} \left(\frac{1}{v} + \tau  \delta_{\Bar{i}} (\Vec{x})\right) + \frac{d_i(\Vec{x})}{v} \right] \diff \Vec{x}, \label{eq_Nh} 
    \end{align}
    where $\delta_{\Bar{i}}(\Vec{x})$ denotes the line density in the perpendicular direction of $i$, e.g., $\delta_{\Bar{\rm E}}(\Vec{x}) = \delta_{\rm N}(\Vec{x}) = \delta_{\rm S}(\Vec{x})$ due to network symmetry as assumed in Section \ref{sec_HetNet}.
\end{subequations}

Eq. (\ref{eq_Nl}) is straightforward, of which terms $\delta_i(\Vec{x}) \diff \Vec{x}$ and $d_i(\Vec{x})h_i(\Vec{x}) \diff \Vec{x}$ respectively account for the straight and detoured lengths of the $i$-line infrastructure. 
Similarly, Eq. (\ref{eq_Nk}) accounts for the distances traveled on straight-going vehicles and detouring vehicles.

Eq. (\ref{eq_Ns}) returns the number of stops where lines intersect in the symmetrical network.
The $\tau$ ($\rm hr$) in Eq. (\ref{eq_Nh}) is the delay per stop, assumed fixed \citep{daganzoPublicTransportationSystems2019}. (The $\tau$ can be readily modified to be variable, accounting for delays dependent on boarding and alighting demand as in \cite{fanOptimalDesignIntersecting2018a}. However, we prefer not to complicate the model formulations further in this paper as this is a secondary consideration relative to our primary research focus.) Thus, both vehicle time and delay spent in cruising and stopping are considered.

\subsubsection{Patron cost.}
Patrons’ cost $Z_P$ (in terms of hours) per operation hour is the sum of four components \citep{daganzoStructureCompetitiveTransit2010}: i) the access/egress time, $T_a$; ii) the out-of-vehicle wait time, $T_w$, i.e., Eq. (\ref{eq_Tw}); iii) the in-vehicle ride time, $T_r$; and iv) the transfer penalty, $T_t$. (Patrons' fare cost is excluded from the generalized cost since it cancels out with the transit agency's revenue.) Thus, we have

\begin{subequations}\label{Network_ZP}
    \begin{align}
        Z_P & = T_a + T_w + T_r + T_t,
    \end{align}
    and,
    \begin{align}
        T_a  &= \beta_w\sum_{i \in \mathbf{I}} \iint_{\Vec{x} \in \mathbf{R}}  \left(\lambda^{i}_{\rm bo}(\Vec{x}) + \lambda^{i}_{\rm al}(\Vec{x})\right) \left(\frac{1}{4\delta_i(\Vec{x})v_w} + \frac{1}{4\delta_{\Bar{i}}(\Vec{x})v_w}\right) \diff \Vec{x}, \label{eq_Ta} \\
        T_w &= \text{Eq. (\ref{eq_Tw})}, \\ 
        T_r &= \sum_{i \in \mathbf{I}}  \iint_{\Vec{x} \in \mathbf{R}} \lambda^{i}_{\rm fl}(\Vec{x}) \left[\left(\frac{1}{v} + \tau  \delta_{\Bar{i}} (\Vec{x}) \right)+ \frac{\alpha h_i(\Vec{x}) d_i(\Vec{x}) }{\delta_i(\vec{x}) v}\right] \diff  \Vec{x}, \label{eq_Tr} \\
        T_t &=\sigma \sum_{i \in \mathbf{I}} \iint_{\Vec{x} \in \mathbf{R}} \lambda_{\rm tr}^{i}(\Vec{x}) \diff  \Vec{x}, \label{eq_Tt}
    \end{align}
\end{subequations}
where $v_w$ ($\rm km \cdot hr^{-1}$) is patrons' average walking speed and $\sigma$ ($\rm hr \cdot transfer^{-1}$) is the penalty per transfer accounting for the walking distance between stops of perpendicular lines and patrons' aversion to transfers. Additionally, walking time is multiplied by a perceived walking time factor $\beta_w$, reflecting the higher perceived cost of walking relative to in-vehicle travel \citep{wardmanReviewBritishEvidence2001}.

Eq. (\ref{eq_Ta}) can be comprehended by noting that $ \frac{1}{4 \delta_i(\vec{x})} $ and $ \frac{1}{4 \delta_{\Bar{i}}(\vec{x})} $ are the average walk distance in $i$ and ${\Bar{i}}$ directions.

Eq. (\ref{eq_Tr}) accounts for the in-vehicle travel times in $i$ direction (i.e., the terms in the parentheses including cruising and stopping times) and the perpendicular detour times. 

The estimation of the transfer penalty in Eq. (\ref{eq_Tt}) is also conservative. When formulating $\lambda^i_{\rm tr}(\Vec{x})$ in Appendix \ref{Appen_B}, we overestimate the transfer demand by ignoring the no-transfer trips whose origins and destinations are connected by the same line.

\section{Solution method}\label{sec_solution}
To address the problem (\ref{Network_problem}), we introduce GP, a class of nonlinear and nonconvex optimization problems characterized by several favorable theoretical and computational properties. A primary advantage of GP is its transformability into an equivalent convex optimization problem via a logarithmic change of variables. Consequently, global optimal solutions can be computed efficiently and reliably, independent of initial starting points \citep{nesterovInteriorPointPolynomialAlgorithms1994}.

GP has been successfully applied in various engineering fields, including communication networks \citep{chiangPowerControlGeometric2007}, integrated circuit design \citep{boydDigitalCircuitOptimization2005}, and natural gas infrastructure networks \citep{misraOptimalCompressionNatural2015}. However, this method imposes strict requirements on the functional form of the optimization problem: {A standard GP minimizes a posynomial objective subject to posynomial inequality constraints and monomial equality constraints, where all decision variables must be strictly positive (See \citealp{avrielAdvancesGeometricProgramming1980} for formal definitions of posynomials and monomials).} 

Simplified TND formulations based on CA frequently align with these structural requirements, rendering GP a promising solution approach for CA-based optimization. This compatibility is largely evident even within our more sophisticated model, where the majority of the objective function components and constraints satisfy the above criteria, and all decision variables are naturally constrained to be positive. Nonetheless, comprehensive CA models often introduce certain structural deviations that preclude the direct application of standard GP solvers. Specifically, in our model, the vehicle flow conservation constraints in Eq. (\ref{cons_flow_conservation_i})--(\ref{cons_flow_conservation_j}) involve posynomials in equality constraints, while the vehicle detour distance in Eq. (\ref{Eq_vehicle_detour_distance_EW})--(\ref{Eq_vehicle_detour_distance_NS}) contains absolute value terms; neither conforms to the standard GP format. To address these non-standard features, we employ SGP, which iteratively approximates non-conforming functions with local monomials at each operating point, thereby transforming the original problem into a sequence of standard GP subproblems that can be solved efficiently.

\subsection{SGP Reformulation}
We first divide the city into equal-sized, square cells of size $(\Delta)^2$ and centered at $(x_n,y_{n'}), n,n'=0,1,2,..., \mathcal{N} = \frac{R}{\Delta} $.
Then, the discretized vehicle flow conservation constraints in Eq. (\ref{cons_flow_conservation_i}) become:
\begin{align}\label{cons_flow_conservation_i_discrete}
    \Delta \sum_{n' = 1}^{\mathcal{N}}{q_{i}\left( x_{n},y_{n'} \right)} = Q_{i},\quad \forall i \in \left\{ \rm E,W \right\},\forall n,n' \in \left\{ 1,2,\ldots,\mathcal{N} \right\}.
\end{align}

To reduce the number of monomials within the equality constraints, the summation in Eq. (\ref{cons_flow_conservation_i_discrete}) is replaced by a recursive subtraction form, defined as $q_{i}\left( x_{n},y_{n'} \right) = \frac{Q_{i}\left( x_{n},y_{n'} \right) - Q_{i}\left( x_{n},y_{n' - 1} \right)}{\Delta}, \forall i \in \left\{\rm E,W \right\}, \forall n \in \left\{ 1,2,\ldots,\mathcal{N} \right\}, \forall n' \in \left\{ 2,3,\ldots,\mathcal{N} \right\}$. Here $q_{i}\left( x_{n},y_{n'} \right)$ and $Q_{i}\left( x_{n},y_{n'} \right)$ are new decision variables that enter our SGP formulation. This allows for the following equivalent transformations:

\begin{subequations}\label{eq_qi_group}
    \begin{align}
        & \dfrac{\Delta \cdot q_{i}\left( x_{n},y_{n'} \right)}{Q_{i}\left( x_{n},y_{n'} \right)} + \dfrac{Q_{i}\left( x_{n},y_{n' - 1} \right)}{Q_{i}\left( x_{n},y_{n'} \right)} = 1,  \label{eq_qi_groupa} \\[8pt]
        & \dfrac{Q_{i}\left( x_{n},y_{\mathcal{N}} \right)}{Q_{i}} = 1,
        \dfrac{\Delta \cdot q_{i}\left( x_{n},y_{1} \right)}{Q_i \left( x_{n},y_{1} \right)} = 1, \\[8pt]
        & \forall i \in \left\{\rm E,W \right\},
        \forall n \in \left\{ 1,2,\ldots,\mathcal{N} \right\}, 
        \forall n' \in \left\{ 2,3,\ldots,\mathcal{N} \right\},
    \end{align}
\end{subequations}
where Eq. (\ref{eq_qi_groupa}) is expressed as a sum of two monomials. In accordance with the SGP framework, we apply first-order monomial approximation via the arithmetic–geometric inequality as follows:

\begin{subequations}\label{eq_qi_groupa_approximate}
    \begin{align}
        &\left( \dfrac{\Delta \cdot {q_{i}\left( x_{n},y_{n'} \right)}^{(m-1)}}{{Q_{i}\left( x_{n},y_{n'} \right)}^{(m-1)}} + \dfrac{{Q_{i}\left( x_{n},y_{n' - 1} \right)}^{(m-1)}}{{Q_{i}\left( x_{n},y_{n'} \right)}^{(m-1)}} \right) \times \\ \notag
        &\left( \dfrac{{Q_{i}\left( x_{n},y_{n'} \right)}^{(m-1)}q_{i}\left( x_{n},y_{n'} \right)}{{q_{i}\left( x_{n},y_{n'} \right)}^{(m-1)}Q_{i}\left( x_{n},y_{n'} \right)} \right)^{\gamma_{i}^{Q,1}\left( x_{n},y_{n'} \right)}
        \left( \dfrac{{Q_{i}\left( x_{n},y_{n'} \right)}^{(m-1)}Q_{i}\left( x_{n},y_{n' - 1} \right)}{{Q_{i}\left( x_{n},y_{n' - 1} \right)}^{(m-1)}Q_{i}\left( x_{n},y_{n'} \right)} \right)^{\gamma_{i}^{Q,2}\left( x_{n},y_{n'} \right)} = 1, \\[8pt]
        &\gamma_{i}^{Q,1}\left( x_{n},y_{n'} \right) = \dfrac{\Delta \cdot {q_{i}\left( x_{n},y_{n'} \right)}^{(m-1)}}{\Delta \cdot {q_{i}\left( x_{n},y_{n'} \right)}^{(m-1)} + {Q_{i}\left( x_{n},y_{n' - 1} \right)}^{(m-1)}}, 
        \gamma_{i}^{Q,2}\left( x_{n},y_{n'} \right) = 1 - \gamma_{i}^{Q,1}\left( x_{n},y_{n'} \right),\\[8pt]
        &\forall i \in \left\{\rm E,W \right\},
        \forall n \in \left\{ 1,2,\ldots,\mathcal{N} \right\}, 
        \forall n' \in \left\{ 2,3,\ldots,\mathcal{N} \right\},
    \end{align}
\end{subequations}
where $m$ denotes the iteration index, and $\gamma_{i}^{Q,1}, \gamma_{i}^{Q,2}$ represent the exponents of the monomial approximation. The values obtained from the previous iteration, such as $q_{i}(x_{n},y_{n'})^{(m-1)}$, are referred to as the reference points (or nominal points).

The transformation for the vehicle flow conservation constraints in Eq. (\ref{cons_flow_conservation_j}) follows a similar procedure:

\begin{subequations}\label{eq_qj_groupa_approximate}
    \begin{align}        
        &\left( \dfrac{\Delta \cdot {q_{j}\left( x_{n},y_{n'} \right)}^{(m-1)}}{{Q_{j}\left( x_{n},y_{n'} \right)}^{(m-1)}} + \dfrac{{Q_{j}\left( x_{n - 1},y_{n'} \right)}^{(m-1)}}{{Q_{j}\left( x_{n},y_{n'} \right)}^{(m-1)}} \right) \times \\ \notag
        &\left( \dfrac{{Q_{j}\left( x_{n},y_{n'} \right)}^{(m-1)}q_{j}\left( x_{n},y_{n'} \right)}{{q_{j}\left( x_{n},y_{n'} \right)}^{(m-1)}Q_{j}\left( x_{n},y_{n'} \right)} \right)^{\gamma_{j}^{Q,1}\left( x_{n},y_{n'} \right)}\left( \dfrac{{Q_{j}\left( x_{n},y_{n'} \right)}^{(m-1)}Q_{j}\left( x_{n - 1},y_{n'} \right)}{{Q_{j}\left( x_{n - 1},y_{n'} \right)}^{(m-1)}Q_{j}\left( x_{n},y_{n'} \right)} \right)^{\gamma_{j}^{Q,2}\left( x_{n},y_{n'} \right)} = 1, \\[8pt] 
        &\gamma_{j}^{Q,1}\left( x_{n},y_{n'} \right) = \dfrac{\Delta \cdot {q_{j}\left( x_{n},y_{n'} \right)}^{(m-1)}}{\Delta \cdot {q_{j}\left( x_{n},y_{n'} \right)}^{(m-1)} + {Q_{j}\left( x_{n - 1},y_{n'} \right)}^{(m-1)}}, 
        \gamma_{j}^{Q,2}\left( x_{n},y_{n'} \right) = 1 - \gamma_{j}^{Q,1}\left( x_{n},y_{n'} \right), \\[8pt]
        &\dfrac{Q_{j}\left( x_{\mathcal{N}},y_{n'} \right)}{Q_{j}} = 1,
        \dfrac{\Delta \cdot q_{j}\left( x_{1},y_{n} \right)}{Q_j \left( x_{1},y_{n} \right)} = 1, \\[8pt]
        &\forall j \in \left\{\rm N,S \right\},\forall n \in \left\{ 2,3,\ldots,\mathcal{N} \right\},\forall n' \in \left\{ 1,2,\ldots,\mathcal{N} \right\}.
    \end{align}
\end{subequations}

For the vehicle detour distance in Eq. (\ref{Eq_vehicle_detour_distance_EW}), the discretized formulation is given by:

\begin{align}\label{Eq_vehicle_detour_distance_EW_discretized}
    d_{i}\left( x_{n},y_{n'} \right) = 
    \begin{cases} 
    \left| \dfrac{Q_{i}\left( x_{n + \varsigma_i},y_{n'} \right) - Q_{i}\left( x_{n},y_{n'} \right)}{\Delta} \right|, 
    & (i,n) \in \{({\rm E},n) \mid 1 \leq n \leq \mathcal{N}-1\} \cup \{({\rm W},n) \mid 2 \leq n \leq \mathcal{N}\} \\[8pt]
    \epsilon, & \text{otherwise}
    \end{cases}
\end{align}
where $\varsigma_i$ denotes the directional step size, defined as $\varsigma_E=+1,\varsigma_W=-1$. In our SGP formulation, $d_{i}\left( x_{n},y_{n'} \right)$ are treated as independent decision variables satisfying standardized constraints. To satisfy the requirement of the standard GP that all decision variables must be strictly positive, we impose a lower bound on the vehicle detour distance by setting a small constant $\epsilon=10^{-5}$.

The presence of absolute value terms in the vehicle detour distance calculation, as shown in Eq. (\ref{Eq_vehicle_detour_distance_EW_discretized}), poses a significant challenge for optimization. To resolve this, we first perform an equivalent linearization of the absolute value operator:

\begin{align}\label{Eq_vehicle_detour_distance_EW_linearization}
    &d_{i}\left( x_{n},y_{n'} \right) \geq \max\left(\frac{Q_{i}\left( x_{n + \varsigma_i},y_{n'} \right) - Q_{i}\left( x_{n},y_{n'} \right)}{\Delta}, \frac{Q_{i}\left( x_{n},y_{n'} \right) - Q_{i}\left( x_{n + \varsigma_i},y_{n'} \right)}{\Delta}, \epsilon \right),
    \\
    &(i,n) \in \{({\rm E},n) \mid 1 \leq n \leq \mathcal{N}-1\} \cup \{({\rm W},n) \mid 2 \leq n \leq \mathcal{N}\} .
\end{align}

This formulation establishes a lower bound for $d_{i}\left( x_{n},y_{n'} \right)$. Since an increase in $d_{i}\left( x_{n},y_{n'} \right)$ results in a corresponding increase in the objective function (cost minimization), the optimal solution is guaranteed to satisfy the equality in the original absolute value expression Eq. (\ref{Eq_vehicle_detour_distance_EW_discretized}).

To comply with the standard GP inequality constraint requirements, we apply a reciprocal transformation to Eq. (\ref{Eq_vehicle_detour_distance_EW_linearization}):

\begin{subequations}\label{Eq_vehicle_detour_distance_EW_reciprocal}
    \begin{align}
        &\dfrac{Q_{i}\left( x_{n + \varsigma_i},y_{n'} \right)}{\Delta \cdot d_{i}\left( x_{n},y_{n'} \right) + Q_{i}\left( x_{n},y_{n'} \right)} \leq 1, \label{Eq_vehicle_detour_distance_EW_reciprocal_a} \\[8pt]
        &\dfrac{Q_{i}\left( x_{n},y_{n'} \right)}{\Delta \cdot d_{i}\left( x_{n},y_{n'} \right) + Q_{i}\left( x_{n + \varsigma_i},y_{n'} \right)} \leq 1, \label{Eq_vehicle_detour_distance_EW_reciprocal_b} \\[8pt] 
        &\dfrac{\epsilon}{d_{i}\left( x_{n},y_{n'} \right)} \leq 1, \\[8pt]
        &(i,n) \in \{({\rm E},n) \mid 1 \leq n \leq \mathcal{N}-1\} \cup \{({\rm W},n) \mid 2 \leq n \leq \mathcal{N}\} .
    \end{align}
\end{subequations}

Next, to convert the left-hand sides of Eq. (\ref{Eq_vehicle_detour_distance_EW_reciprocal_a}) and Eq. (\ref{Eq_vehicle_detour_distance_EW_reciprocal_b}) into monomials, we also employ first-order monomial approximation via the arithmetic–geometric inequality, following a procedure analogous to the transformation of the vehicle flow conservation equations. Consequently, the estimated vehicle detour distance for the ${\rm E,W}$ directions at the $m$-th iteration is given by:

\begin{subequations}\label{Eq_vehicle_detour_distance_EW_approximate}
    \begin{align}
        &\dfrac{Q_{i}\left( x_{n + \varsigma_i},y_{n'} \right)}
        {\left[
        \begin{aligned}
            &\left( \Delta \cdot {d_{i}\left( x_{n},y_{n'} \right)}^{(m-1)} + {Q_{i}\left( x_{n},y_{n'} \right)}^{(m-1)} \right) \times \\
            &\left( \dfrac{d_{i}\left( x_{n},y_{n'} \right)}{{d_{i}\left( x_{n},y_{n'} \right)}^{(m-1)}} \right)^{\gamma_{i}^{d,1}\left( x_{n},y_{n'} \right)}\left( \dfrac{Q_{i}\left( x_{n},y_{n'} \right)}{{Q_{i}\left( x_{n},y_{n'} \right)}^{(m-1)}} \right)^{\gamma_{i}^{d,2}\left( x_{n},y_{n'} \right)}
        \end{aligned}
        \right]
        } \leq 1, \\[8pt] 
        &\dfrac{Q_{i}\left( x_{n},y_{n'} \right)}
        {\left[
        \begin{aligned}
            &\left( \Delta \cdot {d_{i}\left( x_{n},y_{n'} \right)}^{(m-1)} + {Q_{i}\left( x_{n + \varsigma_i},y_{n'} \right)}^{(m-1)} \right) \times \\
            &\left( \dfrac{d_{i}\left( x_{n},y_{n'} \right)}{{d_{i}\left( x_{n},y_{n'} \right)}^{(m-1)}} \right)^{\gamma_{i}^{d,3}\left( x_{n},y_{n'} \right)}\left( \dfrac{Q_{i}\left( x_{n + \varsigma_i},y_{n'} \right)}{{Q_{i}\left( x_{n + \varsigma_i},y_{n'} \right)}^{(m-1)}} \right)^{\gamma_{i}^{d,4}\left( x_{n},y_{n'} \right)}
        \end{aligned}
        \right]
        } \leq 1, \\[8pt] 
        &\gamma_{i}^{d,1}\left( x_{n},y_{n'} \right) = \dfrac{\Delta \cdot {d_{i}\left( x_{n},y_{n'} \right)}^{(m-1)}}{\Delta \cdot {d_{i}\left( x_{n},y_{n'} \right)}^{(m-1)} + {Q_{i}\left( x_{n},y_{n'} \right)}^{(m-1)}}, 
        \gamma_{i}^{d,2}\left( x_{n},y_{n'} \right) = 1 - \gamma_{i}^{d,1}\left( x_{n},y_{n'} \right), \\[8pt] 
        &\gamma_{i}^{d,3}\left( x_{n},y_{n'} \right) = \dfrac{\Delta \cdot {d_{i}\left( x_{n},y_{n'} \right)}^{(m-1)}}{\Delta \cdot {d_{i}\left( x_{n},y_{n'} \right)}^{(m-1)} + {Q_{i}\left( x_{n + \varsigma_i},y_{n'} \right)}^{(m-1)}}, 
        \gamma_{i}^{d,4}\left( x_{n},y_{n'} \right) = 1 - \gamma_{i}^{d,3}\left( x_{n},y_{n'} \right), \\[8pt] 
        &\dfrac{\epsilon}{d_{i}\left( x_{n},y_{n'} \right)} \leq 1, \\[8pt]
        &(i,n) \in \{({\rm E},n) \mid 1 \leq n \leq \mathcal{N}-1\} \cup \{({\rm W},n) \mid 2 \leq n \leq \mathcal{N}\} ,
    \end{align}
\end{subequations}
where $\gamma_{i}^{d,1}, \gamma_{i}^{d,2}, \gamma_{i}^{d,3}, \gamma_{i}^{d,4}$ represent the exponents of the monomial approximation. Through an analogous derivation, the estimation for the vehicle detour distance in the ${\rm N, S}$ directions (where $\varsigma_N=+1,\varsigma_S=-1$) at the $m$-th iteration is obtained as:

\begin{subequations}\label{Eq_vehicle_detour_distance_NS_approximate}
    \begin{align}
        &\dfrac{Q_{j}\left( x_{n},y_{n' + \varsigma_j} \right)}
        {\left[
        \begin{aligned}
            &\left( \Delta \cdot {d_{j}\left( x_{n},y_{n'} \right)}^{(m-1)} + {Q_{j}\left( x_{n},y_{n'} \right)}^{(m-1)} \right) \times \\
            &\left( \dfrac{d_{j}\left( x_{n},y_{n'} \right)}{{d_{j}\left( x_{n},y_{n'} \right)}^{(m-1)}} \right)^{\gamma_{j}^{d,1}\left( x_{n},y_{n'} \right)}\left( \dfrac{Q_{j}\left( x_{n},y_{n'} \right)}{{Q_{j}\left( x_{n},y_{n'} \right)}^{(m-1)}} \right)^{\gamma_{j}^{d,2}\left( x_{n},y_{n'} \right)}
        \end{aligned}
        \right]
        } \leq 1, \\[8pt] 
        &\dfrac{Q_{j}\left( x_{n},y_{n'} \right)}
        {\left[
        \begin{aligned}
        &\left( \Delta \cdot {d_{j}\left( x_{n},y_{n'} \right)}^{(m-1)} + {Q_{j}\left( x_{n},y_{n' + \varsigma_j} \right)}^{(m-1)} \right) \times \\
        &\left( \dfrac{d_{j}\left( x_{n},y_{n'} \right)}{{d_{j}\left( x_{n},y_{n'} \right)}^{(m-1)}} \right)^{\gamma_{j}^{d,3}\left( x_{n},y_{n'} \right)}\left( \dfrac{Q_{j}\left( x_{n},y_{n' + \varsigma_j} \right)}{{Q_{j}\left( x_{n},y_{n' + \varsigma_j} \right)}^{(m-1)}} \right)^{\gamma_{j}^{d,4}\left( x_{n},y_{n'} \right)}
        \end{aligned}
        \right]
        } \leq 1, \\[8pt] 
        &\gamma_j^{d,1}\left( x_{n},y_{n'} \right) = \dfrac{\Delta \cdot {d_{j}\left( x_{n},y_{n'} \right)}^{(m-1)}}{\Delta \cdot {d_{j}\left( x_{n},y_{n'} \right)}^{(m-1)} + {Q_{j}\left( x_{n},y_{n'} \right)}^{(m-1)}}, 
        \gamma_{j}^{d,2}\left( x_{n},y_{n'} \right) = 1 - \gamma_{j}^{d,1}\left( x_{n},y_{n'} \right), \\[8pt] 
        &\gamma_{j}^{d,3}\left( x_{n},y_{n'} \right) = \dfrac{\Delta \cdot {d_{j}\left( x_{n},y_{n'} \right)}^{(m-1)}}{\Delta \cdot {d_{j}\left( x_{n},y_{n'} \right)}^{(m-1)} + {Q_{j}\left( x_{n},y_{n' + \varsigma_j} \right)}^{(m-1)}}, 
        \gamma_{j}^{d,4}\left( x_{n},y_{n'} \right) = 1 - \gamma_{j}^{d,3}\left( x_{n},y_{n'} \right), \\[8pt] 
        &\dfrac{\epsilon}{d_{j}\left( x_{n},y_{n'} \right)} \leq 1, \\[8pt]
        &(j,n) \in \{({\rm N},n) \mid 1 \leq n \leq \mathcal{N}-1\} \cup \{({\rm S},n) \mid 2 \leq n \leq \mathcal{N}\} .
    \end{align}
\end{subequations}

It should be noted that, in this section, monomial approximation is applied only to those constraints in the original problem that do not satisfy the standard GP form. The remaining components in the objective function and constraints are retained in their original form, since they already conform to the structural requirements of standard GP, such as $q_i(x_n,y_{n'})=\delta_i(x_n,y_{n'})/h_i(x_n,y_{n'})$.

\subsection{Solution algorithm}
To improve the numerical robustness of the SGP procedure, we further introduce multiplicative trust-region constraints around the current accepted iterate. At the $m$-th accepted iteration, the trial variables in the GP subproblem are restricted by
\begin{subequations}\label{eq_trust_region}
    \begin{align}
        &\frac{1}{\omega_{\delta}^{(m)}} \le
        \frac{\widetilde{\delta}_i(x_n,y_{n'})}{\delta_i(x_n,y_{n'})^{(m)}}
        \le \omega_{\delta}^{(m)}, \\[6pt]
        &\frac{1}{\omega_{h}^{(m)}} \le
        \frac{\widetilde{h}_i(x_n,y_{n'})}{h_i(x_n,y_{n'})^{(m)}}
        \le \omega_{h}^{(m)}, \\[6pt]
        &\forall i \in \mathbf{I}, \forall n,n' \in \{1,2,\ldots,\mathcal{N}\},
    \end{align}
\end{subequations}
where $(\widetilde{\delta},\widetilde{h},\widetilde{q},\widetilde{Q},\widetilde{d})$ denotes the trial solution returned by the current GP subproblem. $\omega_{\delta}^{(m)} > 1$ and $\omega_{h}^{(m)} > 1$ denote the trust-region radii for network density and headway, respectively. These constraints limit the relative change of decision variables between two consecutive accepted iterates while preserving the standard GP structure, because each bound can be equivalently written as a monomial inequality.

After solving each GP subproblem, we evaluate the original discrete flow-conservation residual ($\nu^{(m)}$) of the trial solution,
\begin{equation}\label{eq_trust_region_violation}
    \nu^{(m)} =
    \max \left\{
    \max_{\substack{i \in \{\rm E,W\}\\ n \in \{1,\ldots,\mathcal{N}\}}}
    \left|
    \frac{\Delta \sum_{n'=1}^{\mathcal{N}} \widetilde{q}_i(x_n,y_{n'})}
    {\widetilde{Q}_i} - 1
    \right|,
    \max_{\substack{j \in \{\rm N,S\}\\ n' \in \{1,\ldots,\mathcal{N}\}}}
    \left|
    \frac{\Delta \sum_{n=1}^{\mathcal{N}} \widetilde{q}_j(x_n,y_{n'})}
    {\widetilde{Q}_j} - 1
    \right|
    \right\},
\end{equation}
The steps of the trust-region-enhanced SGP algorithm are as follows.
 
\begin{steps}[start = 0]
     \item Initialization. Set $m=0$. We randomly generate a feasible initial point $\{\delta_{i}(x_n,y_{n'})^{(0)}, h_i(x_n,y_{n'})^{(0)}\}$ and set the initial trust-region radii $\omega_{\delta}^{(0)} > 1$ and $\omega_{h}^{(0)} > 1$. We also choose a feasibility tolerance $\varepsilon_{\rm fea}$, an expansion factor $\kappa_{\rm inc} > 1$, a contraction factor $\kappa_{\rm dec} \in (0,1)$, and upper bounds $\bar{\omega}_{\delta}$ and $\bar{\omega}_{h}$.
    \item GP Subproblem Construction. Using the current accepted iterate $\{\delta_i(x_n,y_{n'})^{(m)}, h_i(x_n,y_{n'})^{(m)}\}$ as reference points, all auxiliary variables $\{q_i^{(m)}, Q_i^{(m)}, d_i^{(m)}\}$ are updated, and the exponents of the local monomial approximations $\{\gamma_i^{Q,1 \dots 2}, \gamma_i^{d,1 \dots 4}\}$ for all $i \in \mathbf{I}$ are recalculated. The trust-region constraints in Eq.~(\ref{eq_trust_region}) are then added, and the HetNet model is reformulated into a standard GP subproblem.
    \item GP Subproblem Solution. We solve the reformulated GP and obtain a trial solution $\{\widetilde{\delta}_{i}(x_n,y_{n'}), \widetilde{h}_{i}(x_n,y_{n'}), \widetilde{q}_{i}(x_n,y_{n'}), \widetilde{Q}_{i}, \widetilde{d}_{i}(x_n,y_{n'})\}$.
    \item Feasibility and trust-region update. We evaluate the original objective value $\widetilde{Z}$ and the flow-conservation residual $\nu^{(m)}$ from Eq.~(\ref{eq_trust_region_violation}). If $\nu^{(m)} \le \varepsilon_{\rm fea}$ and $\widetilde{Z} \le Z^{(m)}$, the trial solution is accepted and we set
    \begin{subequations}
    \begin{align}
        \delta_{i}(x_n,y_{n'})^{(m+1)} &= \widetilde{\delta}_{i}(x_n,y_{n'}), \\
        h_{i}(x_n,y_{n'})^{(m+1)} &= \widetilde{h}_{i}(x_n,y_{n'}), \\
        Z^{(m+1)} &= \widetilde{Z}, \\
        \omega_{\delta}^{(m+1)} &= \min \left\{ \kappa_{\rm inc}\omega_{\delta}^{(m)}, \bar{\omega}_{\delta} \right\}, \\
        \omega_{h}^{(m+1)} &= \min \left\{ \kappa_{\rm inc}\omega_{h}^{(m)}, \bar{\omega}_{h} \right\}.
    \end{align}
    \end{subequations}
    Otherwise, the trial solution is rejected, the accepted iterate remains unchanged. Note that intermediate iterates that are infeasible to the original problem are common and generally acceptable in SGP procedures based on local approximations, see \cite{boydTutorialGeometricProgramming2007}; our trust-region-enhanced algorithm instead updates the approximation model and progressively steers the accepted iterates toward feasibility and convergence. The trust-region radii are contracted as
    \begin{subequations}
    \begin{align}
        \omega_{\delta}^{(m)} &\leftarrow 1 + \kappa_{\rm dec}\left(\omega_{\delta}^{(m)} - 1\right), \\
        \omega_{h}^{(m)} &\leftarrow 1 + \kappa_{\rm dec}\left(\omega_{h}^{(m)} - 1\right).
    \end{align}
    \end{subequations}
    \item Convergence check. After an accepted update, stop and report the results if the following criterion is satisfied; otherwise, set $m = m+1$ and return to Step 1.

    \begin{subequations}
    \begin{align}
        \left| \frac{\mathrm{Z}^{(m+1)} - \mathrm{Z}^{(m)}}
        {\mathrm{Z}^{(m)}} \right| &\le \xi , \\
        \frac{1}{2\mathcal{N}^2}
        \left(
        \sum_{i \in I,n,n' \in \mathcal{N}}
        \left| \frac{\delta_{i}(x_n,y_{n'})^{(m+1)}
        - \delta_{i}(x_n,y_{n'})^{(m)}}
        {\delta_{i}(x_n,y_{n'})^{(m)}} \right|
        \right. \notag\\
        \left.
        \qquad + \sum_{i \in I,n,n' \in \mathcal{N}}
        \left| \frac{h_{i}(x_n,y_{n'})^{(m+1)}
        - h_{i}(x_n,y_{n'})^{(m)}}
        {h_{i}(x_n,y_{n'})^{(m)}} \right|
        \right)
        &\le \xi .
    \end{align}
    \end{subequations}
    
    where $\xi$ is a pre-specified small value. In practice, $\omega_{\delta}^{(0)} = \omega_{h}^{(0)} = 1.1$, $\kappa_{\rm inc} = 1.2$, $\kappa_{\rm dec} = 0.5$, $\bar{\omega}_{\delta} = \bar{\omega}_{h} = 2$, and $\varepsilon_{\rm fea} = \xi = 10^{-3}$ can be used as default settings. 
\end{steps}

The solution quality of the converged results is confirmed via multiple runs with random initial inputs. We reckon that a solution of good quality is attained when multiple runs always converge to the same results. Step 2 is performed using the commercial GP solver MOSEK. The entire solution algorithm is executed on a desktop computer with a 3.8 GHz CPU (AMD Ryzen 7 9700X) and 16 GB of memory.

\section{Numerical case studies}\label{sec_numerical} 
This section demonstrates the proposed models and solution algorithm in a square city with a dense grid street network. The first subsection presents the demand patterns and parameter values used in the subsequent numerical case studies. Section \ref{sec_Model_accuracy_validation} discretizes the continuous optimization results onto practical transit networks, evaluates the approximation error of the proposed CA model, and examines how the endogenous parameter $\alpha$ affects the optimization results. Section \ref{sec_optimal_design} compares our optimal HetNet design with the optimal designs developed by three benchmark models, i.e., the H-HetNet model \citep{ouyangContinuumApproximationApproach2014}, the P-HetNet model \citep{chienOptimizationGridTransit1997}, and the Homogeneous Network (HomNet) model, of which the latter two serve as two special cases of the HetNet model. Since total patron demand, patrons' average value of time, and city size are three key features that characterize different cities \citep{daganzoStructureCompetitiveTransit2010}, Section \ref{sec_parametric_analysis} varies $D$, $\mu$, and $R$ to examine how the optimization results differ across urban contexts and demonstrate the advantages and robustness of the proposed approach.

\subsection{Set-up} \label{sec_set_up}
The heterogeneity of demand can be characterized with respect to various dimensions, such as demand density, trip distance, and travel directions. In this study, we first adopt the monocentric and commute demand density functions from \cite{ouyangContinuumApproximationApproach2014}. However, these two patterns remain relatively smooth and exhibit low heterogeneity in the present testbed. To examine network-design performance under sharper local demand imbalances and multiple separated high-demand areas, we further introduce four checkerboard demand patterns (termed Checkerboard 1-4), making a total of six demand patterns as illustrated in Figure \ref{fig_ouyang_demand}. Together, these six patterns allow us to compare HetNet with the benchmark networks across both low- and high-heterogeneity regimes. The detailed generation procedures for these heterogeneous demand distributions are provided in Appendix \ref{Appen_C}.

\begin{figure}[htbp]
    \centering
    \subfigure[Monocentric demand]{
        \includegraphics[width=0.45\linewidth]{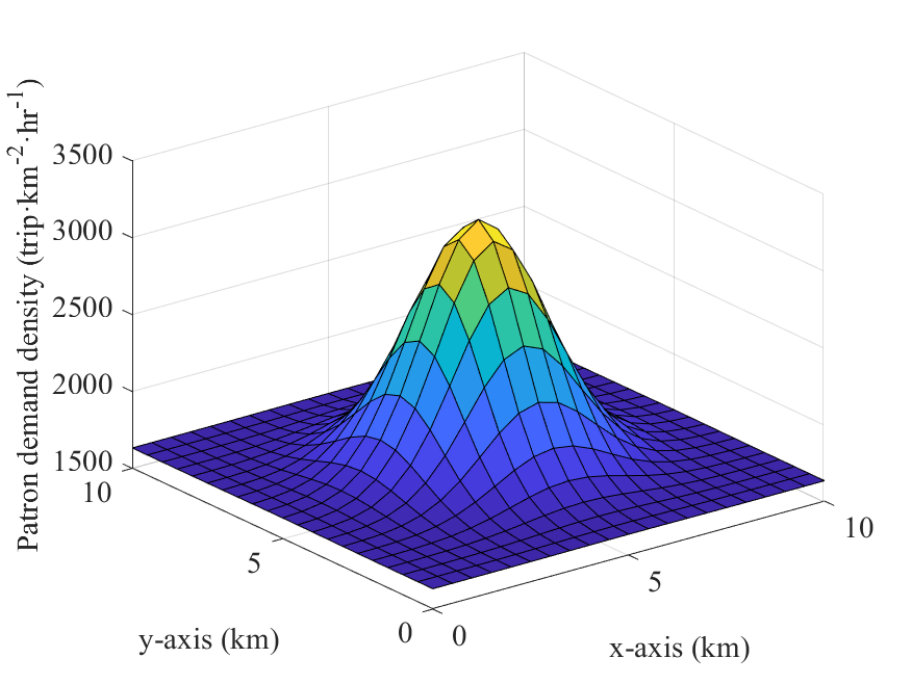}
        \label{Fig_monocentric_demand}
    }
    \subfigure[Commute demand]{
        \includegraphics[width=0.45\linewidth]{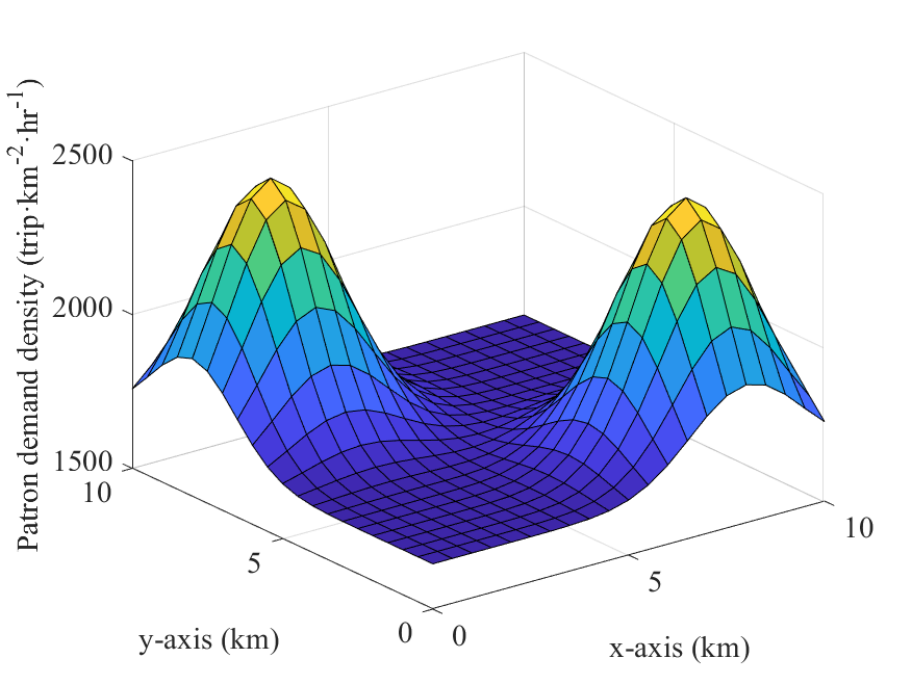}
        \label{Fig_commute_demand}
    }
    \subfigure[Checkerboard 1 demand]{
        \includegraphics[width=0.45\linewidth]{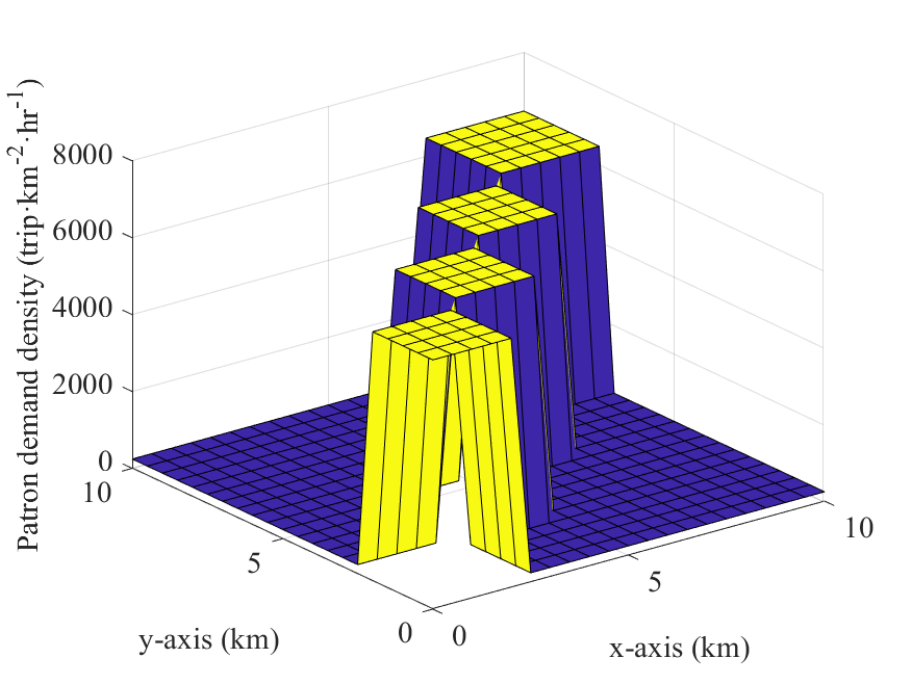}
        \label{Fig_checkerboard0_demand}
    }
    \subfigure[Checkerboard 2 demand]{
        \includegraphics[width=0.45\linewidth]{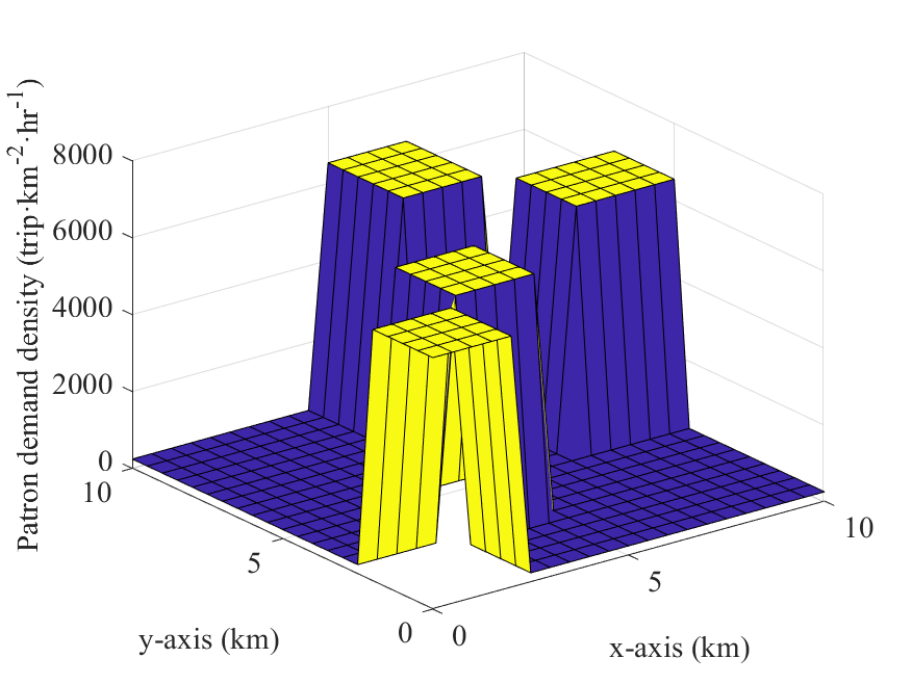}
        \label{Fig_checkerboard1_demand}
    }
    \subfigure[Checkerboard 3 demand]{
        \includegraphics[width=0.45\linewidth]{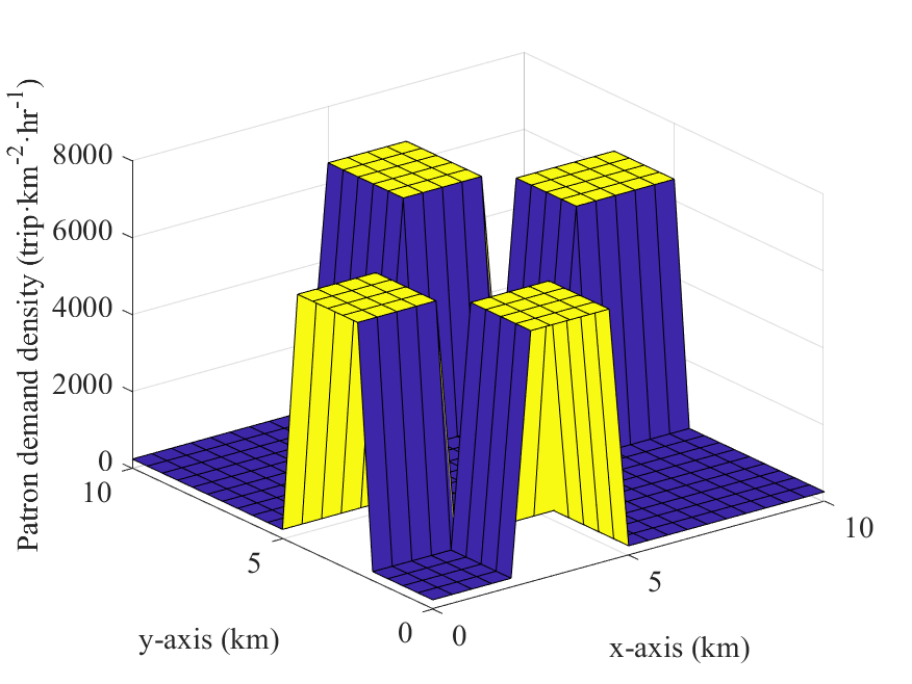}
        \label{Fig_checkerboard2_demand}
    }
    \subfigure[Checkerboard 4 demand]{
        \includegraphics[width=0.45\linewidth]{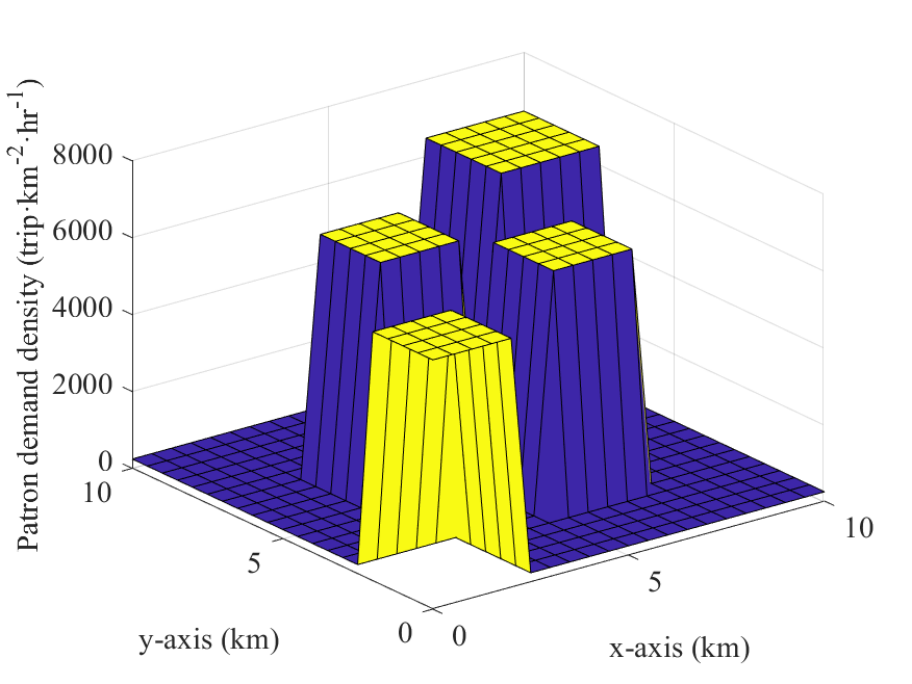}
        \label{Fig_checkerboard3_demand}
    }
    \caption{Boarding and alighting demand density ($D=100,000$ $\rm trip \cdot hr^{-1}$).}
    \label{fig_ouyang_demand}
\end{figure}

Unless otherwise specified, the default values are set as follows: total patron demand $D=100{,}000$ ($\rm trip \cdot hr^{-1}$), value of time $\mu = 25$ ($\rm \$ \cdot hr^{-1}$), city side length $R=10$ ($\rm km$), and patron detour impact factor $\alpha = 0.5$. Other parameter values are borrowed from \cite{ouyangContinuumApproximationApproach2014} for bus transit systems, including $C = 80$ ($\rm trip \cdot veh^{-1}$), $v = 25$ ($\rm km \cdot hr^{-1}$), $v_w = 2$ ($\rm km \cdot hr^{-1}$), $\tau = \frac{30}{3600}$ ($\rm hr \cdot stop^{-1}$), $\sigma = \frac{60}{3600}$ ($\rm hr \cdot transfer^{-1}$), and the parameters in agency cost metrics, $\pi_l=\pi_s=0$, $\pi_k = 2$ ($\rm \$ \cdot veh^{-1} \cdot km^{-1}$), $\pi_h = 40$ ($\rm \$ \cdot veh^{-1} \cdot hr^{-1}$).

In addition to the parameters adopted from the literature, several new parameters are introduced in this study. The perceived walking time factor is $\beta_w = 2$, and the spatial discretization parameter is $\Delta = 0.5$ ($\rm km$), reflecting typical average street block lengths observed in urban areas.

To measure how the degree of heterogeneity ($\eta$) of different demand patterns affects the optimal design, we quantify it using a distance-weighted Jensen–Shannon (JS) divergence, ranging from 0 to 1, with values closer to 1 indicating stronger spatial heterogeneity. The JS divergence measures structural dissimilarity between probability distributions by comparing each distribution with the average distribution, making it suitable for characterizing demand heterogeneity \citep{huangMobilityNetworkApproach2018}. We further incorporate a demand-distance weighting scheme to better capture the burden imposed by spatial demand patterns on the network \citep{liuExploringSpatiallyHeterogeneous2022}. The detailed computation procedure of $\eta$ is provided in Appendix \ref{Appen_D}.

\subsection{Model accuracy validation}\label{sec_Model_accuracy_validation}
\subsubsection{Network discretization.}\label{sec_network_discretization}
To validate the accuracy of the proposed CA model in estimating generalized costs, the continuous optimization results are discretized onto actual transit networks. Following the assumptions in Section \ref{sec_HetNet}, patron demand is assigned to the nearest boarding and alighting stops, and the shortest-path algorithm is used to determine individual travel routes. Based on these routes, the exact generalized cost and its individual components are calculated. The details of the network discretization procedure are provided in Appendix \ref{Appen_F}. This evaluation result is reported in Table \ref{tab_network_discretization_alpha_0.5}.  

\begin{table}[!ht]
\centering
\caption{Comparison between CA optimization results and corresponding discretized transit networks ($\alpha=0.5$).}
\label{tab_network_discretization_alpha_0.5}
\resizebox{\columnwidth}{!}{%
\renewcommand{\arraystretch}{1.3}
\begin{tabular}{|c|ccc|ccc|ccc|}
\hline
\multirow{2}{*}{\textbf{\begin{tabular}[c]{@{}c@{}}Cost term\\ ($\rm min \cdot trip^{-1}$)\end{tabular}}} & \multicolumn{3}{c|}{\textbf{\begin{tabular}[c]{@{}c@{}}Monocentric demand\\ $\eta=0.0439$\end{tabular}}}    & \multicolumn{3}{c|}{\textbf{\begin{tabular}[c]{@{}c@{}}Commute demand\\ $\eta=0.0792$\end{tabular}}}        & \multicolumn{3}{c|}{\textbf{\begin{tabular}[c]{@{}c@{}}Checkerboard 1 demand\\ $\eta=0.5486$\end{tabular}}} \\ \cline{2-10} 
                                                                                                          & \multicolumn{1}{c|}{\textbf{CA}}     & \multicolumn{1}{c|}{\textbf{Discretization}}    & \textbf{Gap}       & \multicolumn{1}{c|}{\textbf{CA}}     & \multicolumn{1}{c|}{\textbf{Discretization}}    & \textbf{Gap}       & \multicolumn{1}{c|}{\textbf{CA}}     & \multicolumn{1}{c|}{\textbf{Discretization}}    & \textbf{Gap}       \\ \hline
$N_k$                                                                                                     & \multicolumn{1}{c|}{1.41}            & \multicolumn{1}{c|}{1.51}                       & 6.62\%             & \multicolumn{1}{c|}{1.40}            & \multicolumn{1}{c|}{1.56}                       & 10.26\%            & \multicolumn{1}{c|}{1.62}            & \multicolumn{1}{c|}{1.75}                       & 7.43\%             \\ \hline
$N_h$                                                                                                     & \multicolumn{1}{c|}{2.05}            & \multicolumn{1}{c|}{2.11}                       & 2.84\%             & \multicolumn{1}{c|}{2.03}            & \multicolumn{1}{c|}{2.14}                       & 5.14\%             & \multicolumn{1}{c|}{2.06}            & \multicolumn{1}{c|}{2.12}                       & 2.83\%             \\ \hline
$T_a$                                                                                                     & \multicolumn{1}{c|}{16.04}           & \multicolumn{1}{c|}{14.14}                      & 11.85\%            & \multicolumn{1}{c|}{16.40}           & \multicolumn{1}{c|}{14.67}                      & 10.55\%            & \multicolumn{1}{c|}{13.79}           & \multicolumn{1}{c|}{14.02}                      & 1.64\%             \\ \hline
$T_w$                                                                                                     & \multicolumn{1}{c|}{3.46}            & \multicolumn{1}{c|}{3.47}                       & 0.29\%             & \multicolumn{1}{c|}{3.44}            & \multicolumn{1}{c|}{3.56}                       & 3.37\%             & \multicolumn{1}{c|}{3.53}            & \multicolumn{1}{c|}{3.31}                       & 6.23\%             \\ \hline
$T_r$                                                                                                     & \multicolumn{1}{c|}{27.32}           & \multicolumn{1}{c|}{28.09}                      & 2.74\%             & \multicolumn{1}{c|}{29.23}           & \multicolumn{1}{c|}{29.37}                      & 0.48\%             & \multicolumn{1}{c|}{26.81}           & \multicolumn{1}{c|}{26.70}                      & 0.41\%             \\ \hline
$T_t$                                                                                                     & \multicolumn{1}{c|}{0.90}            & \multicolumn{1}{c|}{0.92}                       & 2.17\%             & \multicolumn{1}{c|}{0.91}            & \multicolumn{1}{c|}{0.84}                       & 7.69\%             & \multicolumn{1}{c|}{0.91}            & \multicolumn{1}{c|}{0.92}                       & 1.09\%             \\ \hline
$Z$                                                                                                       & \multicolumn{1}{c|}{51.18}           & \multicolumn{1}{c|}{50.24}                      & \textbf{1.84\%}    & \multicolumn{1}{c|}{53.41}           & \multicolumn{1}{c|}{52.14}                      & \textbf{2.38\%}    & \multicolumn{1}{c|}{48.72}           & \multicolumn{1}{c|}{48.82}                      & \textbf{0.20\%}    \\ \hline
\multirow{2}{*}{\textbf{\begin{tabular}[c]{@{}c@{}}Cost term\\ ($\rm min \cdot trip^{-1}$)\end{tabular}}} & \multicolumn{3}{c|}{\textbf{\begin{tabular}[c]{@{}c@{}}Checkerboard 2 demand\\ $\eta=0.5504$\end{tabular}}} & \multicolumn{3}{c|}{\textbf{\begin{tabular}[c]{@{}c@{}}Checkerboard 3 demand\\ $\eta=0.5524$\end{tabular}}} & \multicolumn{3}{c|}{\textbf{\begin{tabular}[c]{@{}c@{}}Checkerboard 4 demand\\ $\eta=0.5522$\end{tabular}}} \\ \cline{2-10} 
                                                                                                          & \multicolumn{1}{c|}{\textbf{CA}}     & \multicolumn{1}{c|}{\textbf{Discretization}}    & \textbf{Gap}       & \multicolumn{1}{c|}{\textbf{CA}}     & \multicolumn{1}{c|}{\textbf{Discretization}}    & \textbf{Gap}       & \multicolumn{1}{c|}{\textbf{CA}}     & \multicolumn{1}{c|}{\textbf{Discretization}}    & \textbf{Gap}       \\ \hline
$N_k$                                                                                                     & \multicolumn{1}{c|}{1.59}            & \multicolumn{1}{c|}{1.73}                       & 8.09\%             & \multicolumn{1}{c|}{1.58}            & \multicolumn{1}{c|}{1.74}                       & 9.20\%             & \multicolumn{1}{c|}{1.59}            & \multicolumn{1}{c|}{1.76}                       & 9.66\%             \\ \hline
$N_h$                                                                                                     & \multicolumn{1}{c|}{2.04}            & \multicolumn{1}{c|}{2.13}                       & 4.23\%             & \multicolumn{1}{c|}{2.03}            & \multicolumn{1}{c|}{2.12}                       & 4.25\%             & \multicolumn{1}{c|}{2.03}            & \multicolumn{1}{c|}{2.13}                       & 4.69\%             \\ \hline
$T_a$                                                                                                     & \multicolumn{1}{c|}{13.67}           & \multicolumn{1}{c|}{13.80}                      & 0.94\%             & \multicolumn{1}{c|}{13.66}           & \multicolumn{1}{c|}{13.61}                      & 0.37\%             & \multicolumn{1}{c|}{13.76}           & \multicolumn{1}{c|}{14.13}                      & 2.62\%             \\ \hline
$T_w$                                                                                                     & \multicolumn{1}{c|}{3.49}            & \multicolumn{1}{c|}{3.40}                       & 2.58\%             & \multicolumn{1}{c|}{3.48}            & \multicolumn{1}{c|}{3.44}                       & 1.15\%             & \multicolumn{1}{c|}{3.48}            & \multicolumn{1}{c|}{3.37}                       & 3.16\%             \\ \hline
$T_r$                                                                                                     & \multicolumn{1}{c|}{26.27}           & \multicolumn{1}{c|}{26.11}                      & 0.61\%             & \multicolumn{1}{c|}{26.24}           & \multicolumn{1}{c|}{26.35}                      & 0.42\%             & \multicolumn{1}{c|}{26.89}           & \multicolumn{1}{c|}{26.90}                      & 0.04\%             \\ \hline
$T_t$                                                                                                     & \multicolumn{1}{c|}{0.91}            & \multicolumn{1}{c|}{0.92}                       & 1.09\%             & \multicolumn{1}{c|}{0.91}            & \multicolumn{1}{c|}{0.91}                       & 0.00\%             & \multicolumn{1}{c|}{0.91}            & \multicolumn{1}{c|}{0.92}                       & 1.09\%             \\ \hline
$Z$                                                                                                       & \multicolumn{1}{c|}{47.97}           & \multicolumn{1}{c|}{48.09}                      & \textbf{0.25\%}    & \multicolumn{1}{c|}{47.90}           & \multicolumn{1}{c|}{48.17}                      & \textbf{0.56\%}    & \multicolumn{1}{c|}{48.66}           & \multicolumn{1}{c|}{49.21}                      & \textbf{1.12\%}    \\ \hline
\end{tabular}
}
\end{table}

The results show that the differences between the generalized costs of the CA model and the discretized network, which reflect the approximation error of our CA approach are minimal across all demand patterns, with deviations generally below 2.5\%. The discrepancies in individual cost components are moderately larger but also limited, with the majority remaining within 12\%. As is common in CA transit network design models (see \citealp{daganzoPublicTransportationSystems2019}), these component-level deviations partly offset one another when aggregated into the generalized cost, resulting in the smaller deviations observed for $Z$. These results indicate that the proposed CA model can reliably approximate the operational characteristics and generalized costs of realistic transit networks while maintaining analytical tractability.

It is worth noting that the gap in the access time component $T_a$ is occasionally larger than that of other cost components. This difference mainly arises from the continuous approximation adopted in the CA model. In the CA formulation, access distance is approximated as a continuous average based on route spacing, whereas in the discretized network, patrons are assigned to the nearest actual stops. Such discrete stop assignments can shorten or lengthen individual access distances depending on the local network geometry, leading to slightly larger deviations in $T_a$. Nevertheless, since access time represents only one component of the total generalized cost and the deviations remain within a moderate range, the overall impact on the generalized cost $Z$ remains very limited.

Figure~\ref{fig_network_discretization_EW} presents the discretized transit network structures in the ${\rm E}$ and ${\rm W}$ directions under the six demand patterns. For brevity, the N--S lines are not shown, as they display similar spatial patterns; that is, line density generally increases in high-demand-density areas and decreases in low-demand-density areas. 

The results show that the proposed transit network design responds effectively to spatially heterogeneous demand distributions (see Figure~\ref{fig_ouyang_demand}). In areas with stronger demand intensity, transit lines become more closely spaced, enabling service capacity to be allocated more efficiently and reducing overall operating costs. By contrast, the monocentric and commute demand patterns exhibit relatively low heterogeneity ($\eta$ = 0.0439, 0.0792), implying that demand is more evenly distributed across the study area. In such cases, a relatively homogeneous network configuration naturally emerges as the optimal design.

\begin{figure}[htbp]
    \centering
    \subfigure[Monocentric demand]{
        \includegraphics[width=0.45\linewidth]{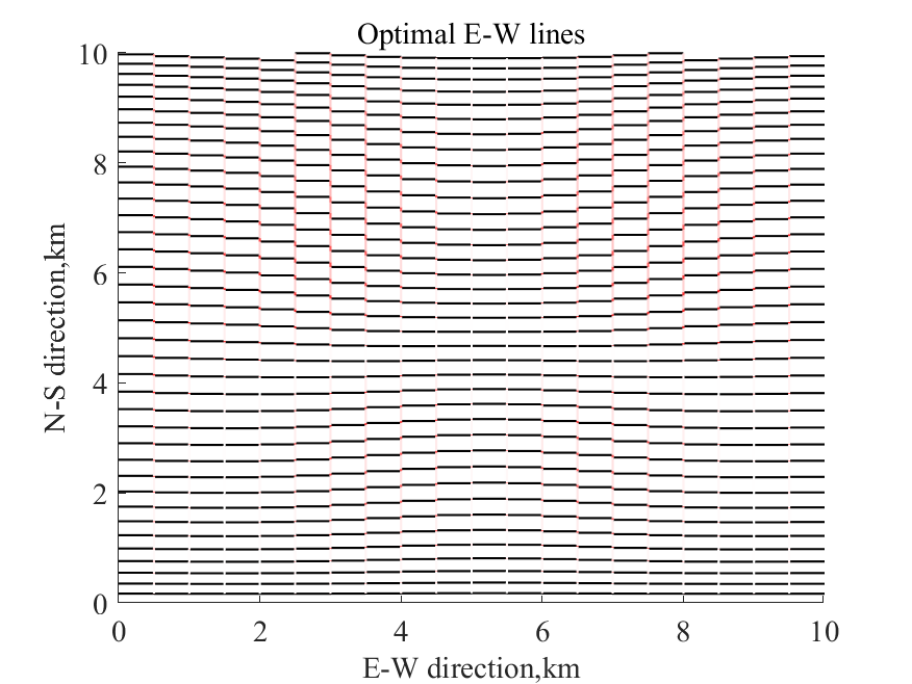}
    }
    \subfigure[Commute demand]{
        \includegraphics[width=0.45\linewidth]{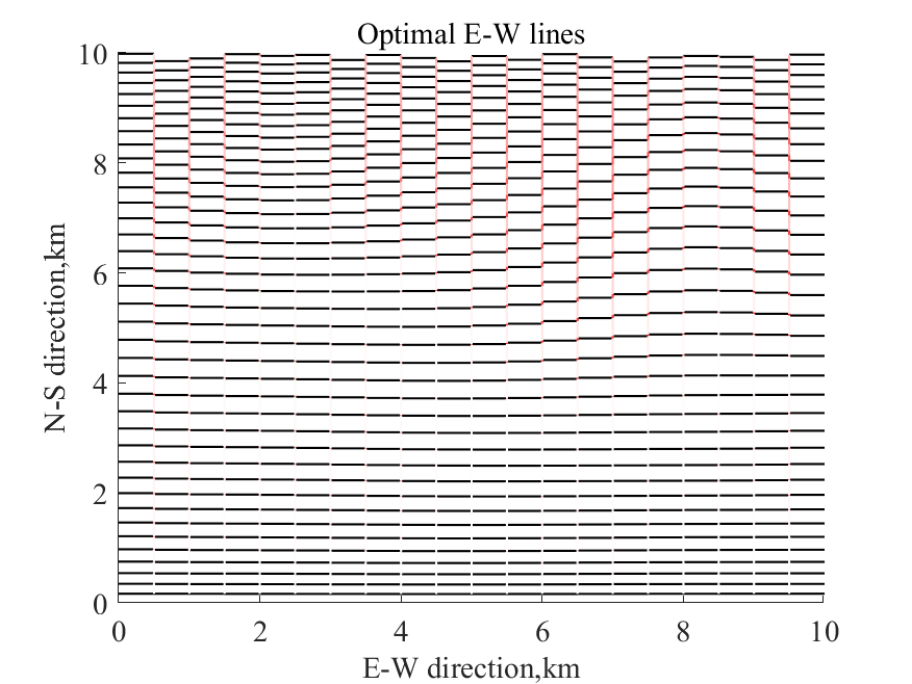}
    }
    \subfigure[Checkerboard 1 demand]{
        \includegraphics[width=0.45\linewidth]{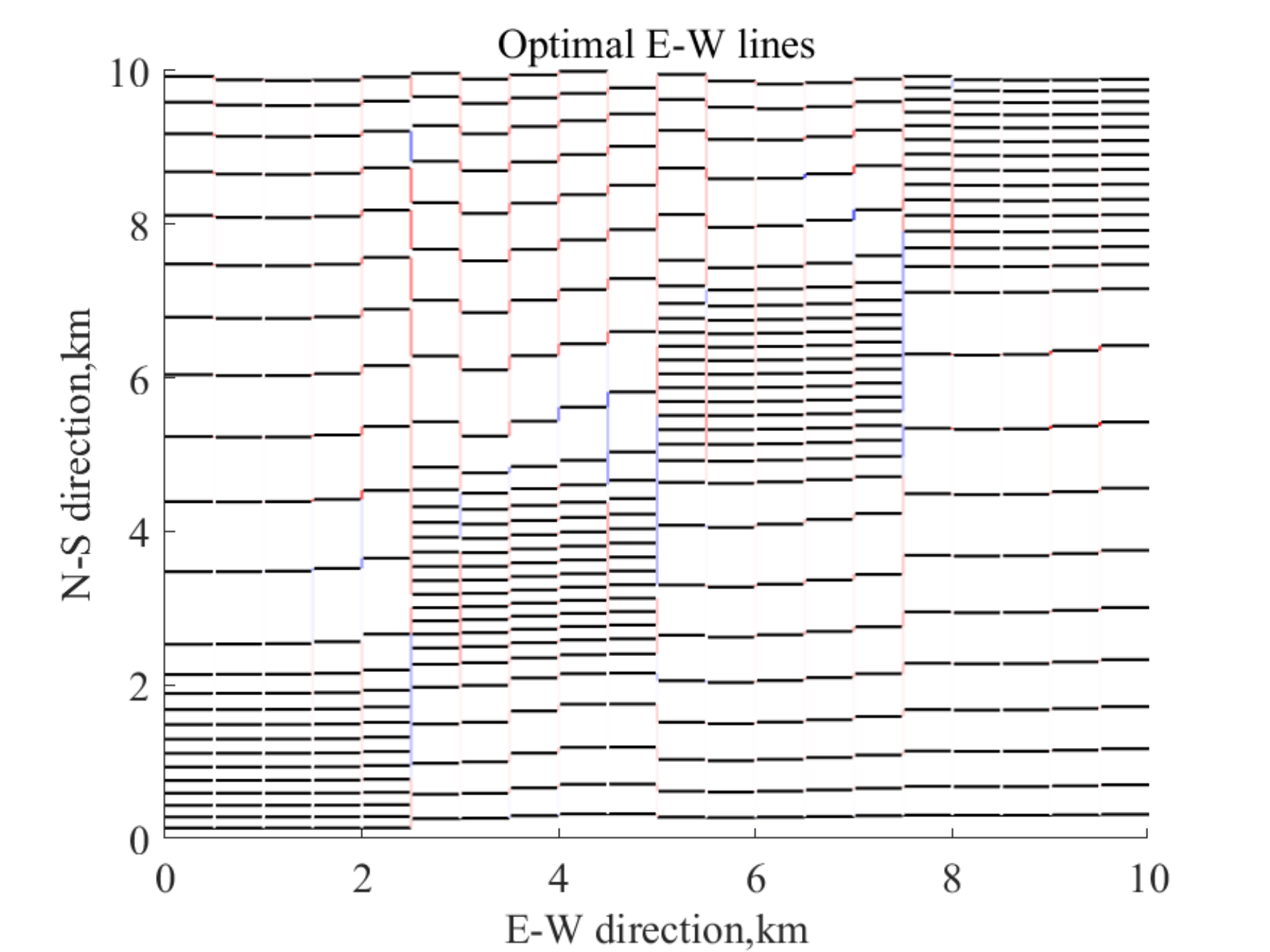}
    }
    \subfigure[Checkerboard 2 demand]{
        \includegraphics[width=0.45\linewidth]{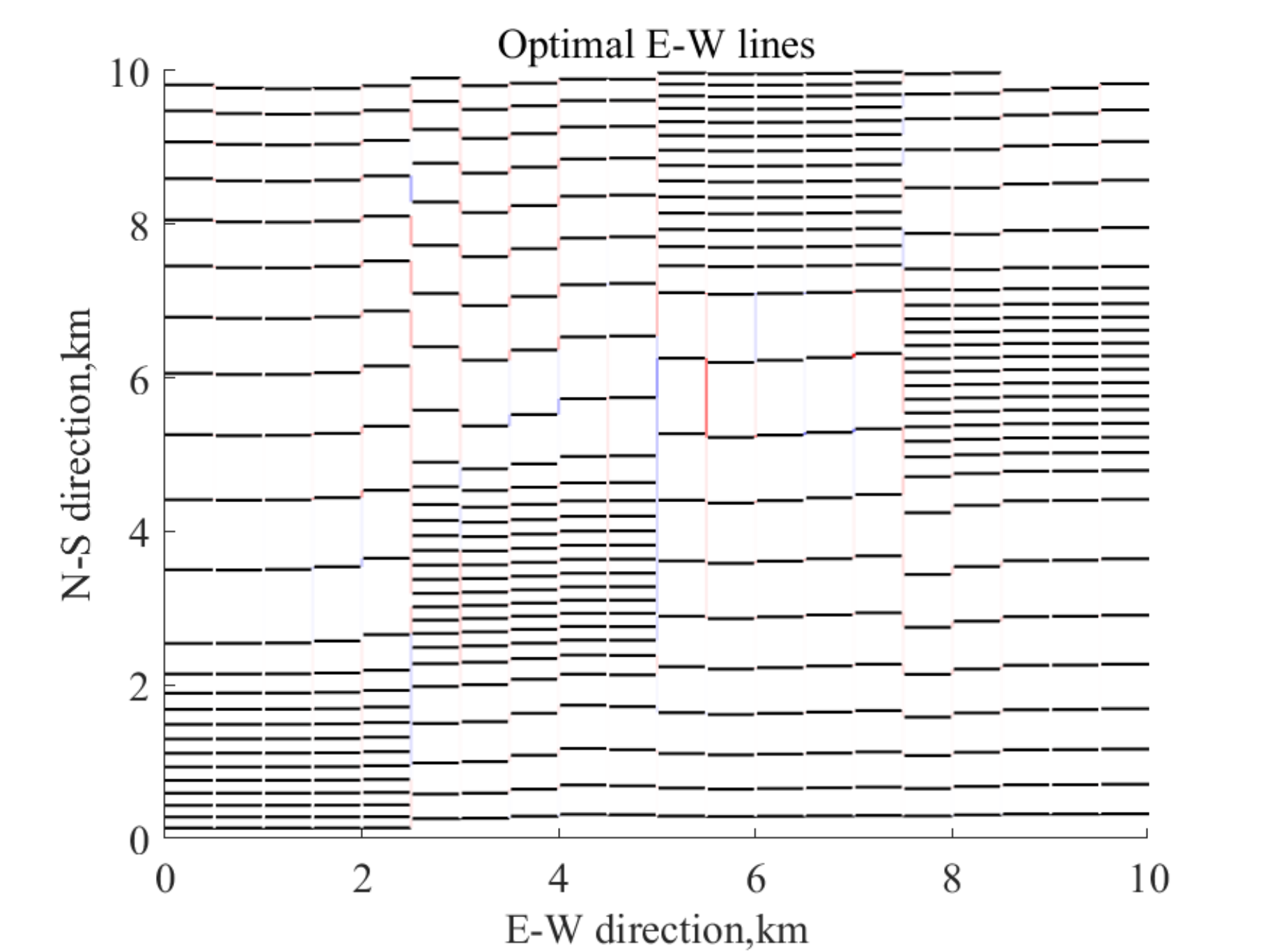}
    }
    \subfigure[Checkerboard 3 demand]{
        \includegraphics[width=0.45\linewidth]{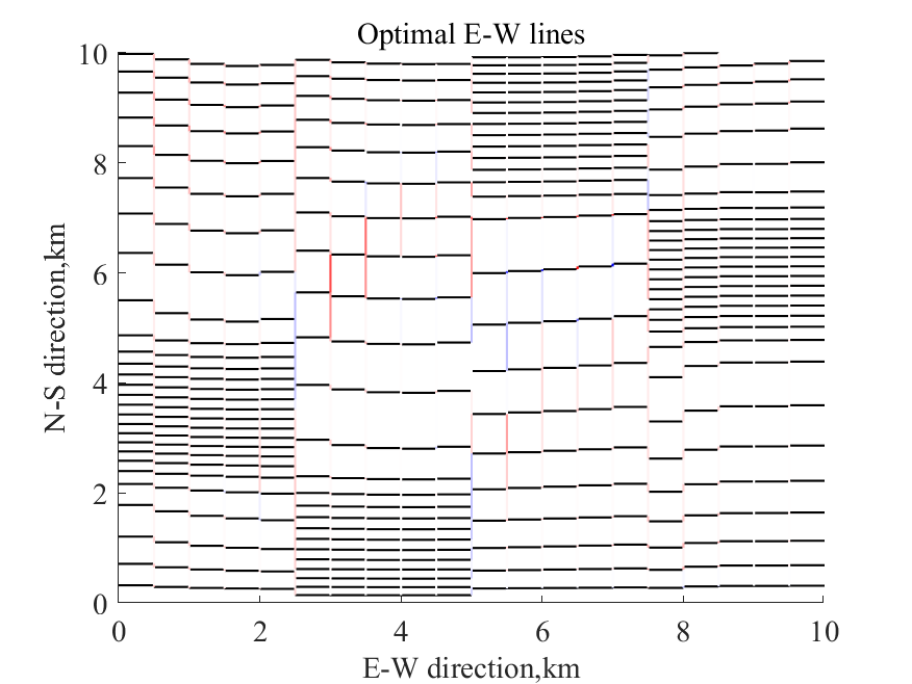}
    }
    \subfigure[Checkerboard 4 demand]{
        \includegraphics[width=0.45\linewidth]{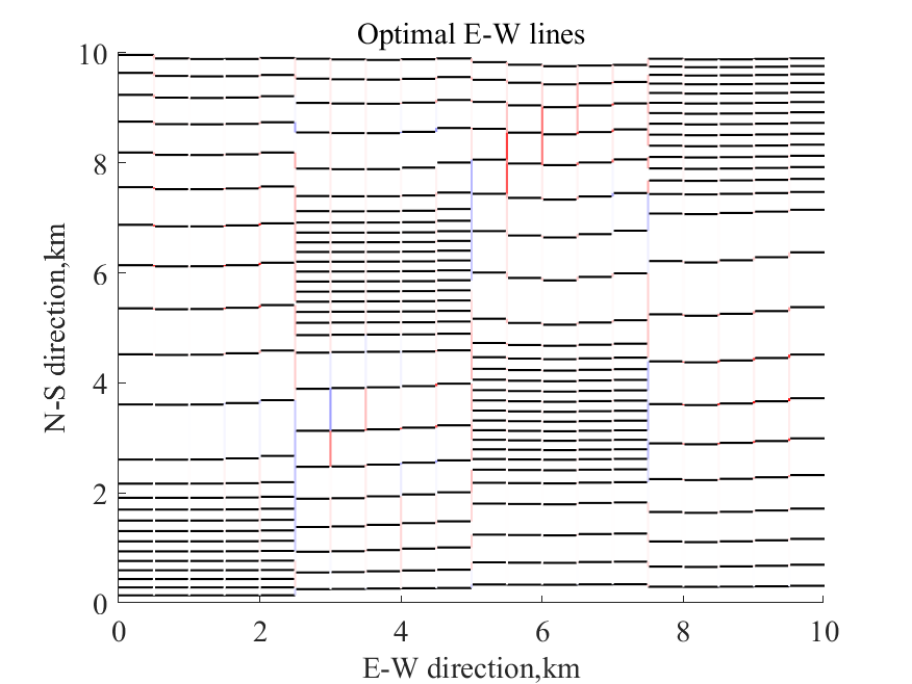}
    }
    \caption{Discretized transit networks in the ${\rm E}$ and ${\rm W}$ directions. The red lines represent vehicular detour flows, while the color intensity indicates the magnitude of the corresponding flow.}
    \label{fig_network_discretization_EW}
\end{figure} 

\subsubsection{Sensitivity analysis of the endogenous parameter $\alpha$.}\label{sec_Sensitivity_analysis_of_the_endogenous_parameter}
As discussed in Section \ref{sec_in_vehicle_time}, the detour impact factor $\alpha$ is endogenous, as its realized value is jointly determined by the network configuration and the induced demand-flow assignment. Since neither is known before the design problem is solved, $\alpha$ cannot be calibrated a priori. It is therefore necessary to assess whether treating $\alpha$ as an exogenous parameter materially affects the optimization outcomes.

To this end, we vary $\alpha$ over the interval $[0,1]$ under six demand patterns and evaluate the resulting optimal network designs and generalized costs. The CA optimization results show that the model is largely insensitive to $\alpha$. Specifically, for each demand pattern, the standard deviation of the average generalized cost across different $\alpha$ levels remains below 0.5 $\rm min\cdot trip^{-1}$, which is less than 1\% of the corresponding mean generalized cost, suggesting that the optimized system performance changes only marginally as $\alpha$ varies. This robustness implies that changes in the effective detour burden do not materially affect the main design conclusions.

We further compare the CA solutions with their corresponding discretized network evaluations. Across all tested cases, the gap in generalized cost remains within 3\%, demonstrating that the CA model consistently provides an accurate approximation of the discretized network costs over the full range of $\alpha$. Moreover, once the discretized network and the associated demand assignment are obtained, the realized value of $\alpha$ can be computed ex post. These values are found to lie predominantly in the range $[0.3, 0.7]$.

Overall, these results indicate that $\alpha$ has little influence on either the optimization outcomes or the estimation accuracy of the CA model. Furthermore, setting $\alpha = 0.5$ avoids systematically overestimating or underestimating the magnitude of the effective detour, and is therefore adopted in the subsequent analysis.

\subsection{Optimized designs of heterogeneous grid network} \label{sec_optimal_design}
Table \ref{tab: versus_designs_OuyangDemand} presents the results of four network designs in terms of the average agency cost, patron cost, and generalized cost under six demand patterns. For fair comparison, the benchmark layouts are implemented as harmonized versions of the original models under the same demand functions, cost components, and capacity constraints used in this paper. P-HetNet and HomNet are formulated as restricted special cases of the proposed HetNet model and are solved to global optimality using GP. The detailed formulations are provided in Appendix~\ref{Appen_E}, while the complete methodological development can be found in \cite{maoDesignTransitNetworks2026}. In addition, H-HetNet is solved using \texttt{fmincon} to obtain a high-quality near-global solution. These strong and effective baselines provide a rigorous reference for evaluating the performance of the proposed HetNet model. To ensure a fair comparison, all network designs are evaluated using continuous optimization solutions.

\begin{table}[!ht]
\centering
\caption{Comparison with other grid networks under six demand patterns.}
\label{tab: versus_designs_OuyangDemand}
\resizebox{\columnwidth}{!}{%
\renewcommand{\arraystretch}{1.3}
\begin{tabular}{|ccccc|cccc|}
\hline
\multicolumn{1}{|c|}{\multirow{2}{*}{\textbf{\begin{tabular}[c]{@{}c@{}}Cost term\\ ($\rm min \cdot trip^{-1}$)\end{tabular}}}} & \multicolumn{4}{c|}{\textbf{Monocentric demand $\eta=0.0439$}}                                                                           & \multicolumn{4}{c|}{\textbf{Commute demand $\eta=0.0792$}}                                                                               \\ \cline{2-9} 
\multicolumn{1}{|c|}{}                                                                                                          & \multicolumn{1}{c|}{\textbf{HetNet}} & \multicolumn{1}{c|}{\textbf{H-HetNet}} & \multicolumn{1}{c|}{\textbf{P-HetNet}} & \textbf{HomNet} & \multicolumn{1}{c|}{\textbf{HetNet}} & \multicolumn{1}{c|}{\textbf{H-HetNet}} & \multicolumn{1}{c|}{\textbf{P-HetNet}} & \textbf{HomNet} \\ \hline
\multicolumn{1}{|c|}{$N_k$}                                                                                                     & \multicolumn{1}{c|}{1.41}            & \multicolumn{1}{c|}{1.39}              & \multicolumn{1}{c|}{1.41}              & 1.39            & \multicolumn{1}{c|}{1.40}            & \multicolumn{1}{c|}{1.37}              & \multicolumn{1}{c|}{1.40}              & 1.37            \\ \hline
\multicolumn{1}{|c|}{$N_h$}                                                                                                     & \multicolumn{1}{c|}{2.05}            & \multicolumn{1}{c|}{1.97}              & \multicolumn{1}{c|}{2.05}              & 1.97            & \multicolumn{1}{c|}{2.03}            & \multicolumn{1}{c|}{1.92}              & \multicolumn{1}{c|}{2.02}              & 1.92            \\ \hline
\multicolumn{1}{|c|}{$T_a$}                                                                                                     & \multicolumn{1}{c|}{16.04}           & \multicolumn{1}{c|}{16.26}             & \multicolumn{1}{c|}{16.07}             & 16.26           & \multicolumn{1}{c|}{16.40}           & \multicolumn{1}{c|}{16.84}             & \multicolumn{1}{c|}{16.47}             & 16.84           \\ \hline
\multicolumn{1}{|c|}{$T_w$}                                                                                                     & \multicolumn{1}{c|}{3.46}            & \multicolumn{1}{c|}{3.36}              & \multicolumn{1}{c|}{3.45}              & 3.36            & \multicolumn{1}{c|}{3.44}            & \multicolumn{1}{c|}{3.29}              & \multicolumn{1}{c|}{3.43}              & 3.29            \\ \hline
\multicolumn{1}{|c|}{$T_r$}                                                                                                     & \multicolumn{1}{c|}{27.32}           & \multicolumn{1}{c|}{27.69}             & \multicolumn{1}{c|}{27.35}             & 27.69           & \multicolumn{1}{c|}{29.23}           & \multicolumn{1}{c|}{29.90}             & \multicolumn{1}{c|}{29.31}             & 29.90           \\ \hline
\multicolumn{1}{|c|}{$T_t$}                                                                                                     & \multicolumn{1}{c|}{0.90}            & \multicolumn{1}{c|}{0.90}              & \multicolumn{1}{c|}{0.90}              & 0.90            & \multicolumn{1}{c|}{0.91}            & \multicolumn{1}{c|}{0.91}              & \multicolumn{1}{c|}{0.91}              & 0.91            \\ \hline
\multicolumn{1}{|c|}{$Z$}                                                                                                       & \multicolumn{1}{c|}{51.18}           & \multicolumn{1}{c|}{51.57}             & \multicolumn{1}{c|}{51.23}             & 51.57           & \multicolumn{1}{c|}{53.41}           & \multicolumn{1}{c|}{54.23}             & \multicolumn{1}{c|}{53.54}             & 54.23           \\ \hline
\multicolumn{2}{|c|}{\textbf{\begin{tabular}[c]{@{}c@{}}HetNet's cost savings\\ against the benchmarks\end{tabular}}}                                                  & \multicolumn{1}{c|}{0.76\%}            & \multicolumn{1}{c|}{0.10\%}            & 0.76\%          & \multicolumn{1}{c|}{/}               & \multicolumn{1}{c|}{1.51\%}            & \multicolumn{1}{c|}{0.24\%}            & 1.51\%          \\ \hline
\multicolumn{1}{|c|}{\multirow{2}{*}{\textbf{\begin{tabular}[c]{@{}c@{}}Cost term\\ ($\rm min \cdot trip^{-1}$)\end{tabular}}}} & \multicolumn{4}{c|}{\textbf{Checkerboard 1 demand $\eta=0.5486$}}                                                                        & \multicolumn{4}{c|}{\textbf{Checkerboard 2 demand $\eta=0.5504$}}                                                                        \\ \cline{2-9} 
\multicolumn{1}{|c|}{}                                                                                                          & \multicolumn{1}{c|}{\textbf{HetNet}} & \multicolumn{1}{c|}{\textbf{H-HetNet}} & \multicolumn{1}{c|}{\textbf{P-HetNet}} & \textbf{HomNet} & \multicolumn{1}{c|}{\textbf{HetNet}} & \multicolumn{1}{c|}{\textbf{H-HetNet}} & \multicolumn{1}{c|}{\textbf{P-HetNet}} & \textbf{HomNet} \\ \hline
\multicolumn{1}{|c|}{$N_k$}                                                                                                     & \multicolumn{1}{c|}{1.62}            & \multicolumn{1}{c|}{2.46}              & \multicolumn{1}{c|}{1.48}              & 1.70            & \multicolumn{1}{c|}{1.59}            & \multicolumn{1}{c|}{2.57}              & \multicolumn{1}{c|}{1.54}              & 1.82            \\ \hline
\multicolumn{1}{|c|}{$N_h$}                                                                                                     & \multicolumn{1}{c|}{2.06}            & \multicolumn{1}{c|}{3.64}              & \multicolumn{1}{c|}{2.13}              & 2.38            & \multicolumn{1}{c|}{2.04}            & \multicolumn{1}{c|}{3.76}              & \multicolumn{1}{c|}{2.24}              & 2.57            \\ \hline
\multicolumn{1}{|c|}{$T_a$}                                                                                                     & \multicolumn{1}{c|}{13.79}           & \multicolumn{1}{c|}{11.25}             & \multicolumn{1}{c|}{16.38}             & 16.59           & \multicolumn{1}{c|}{13.67}           & \multicolumn{1}{c|}{11.07}             & \multicolumn{1}{c|}{16.24}             & 16.35           \\ \hline
\multicolumn{1}{|c|}{$T_w$}                                                                                                     & \multicolumn{1}{c|}{3.53}            & \multicolumn{1}{c|}{2.41}              & \multicolumn{1}{c|}{3.28}              & 2.70            & \multicolumn{1}{c|}{3.49}            & \multicolumn{1}{c|}{2.34}              & \multicolumn{1}{c|}{3.18}              & 2.56            \\ \hline
\multicolumn{1}{|c|}{$T_r$}                                                                                                     & \multicolumn{1}{c|}{26.81}           & \multicolumn{1}{c|}{31.60}             & \multicolumn{1}{c|}{29.23}             & 29.95           & \multicolumn{1}{c|}{26.27}           & \multicolumn{1}{c|}{31.16}             & \multicolumn{1}{c|}{28.65}             & 29.27           \\ \hline
\multicolumn{1}{|c|}{$T_t$}                                                                                                     & \multicolumn{1}{c|}{0.91}            & \multicolumn{1}{c|}{1.22}              & \multicolumn{1}{c|}{0.91}              & 0.91            & \multicolumn{1}{c|}{0.91}            & \multicolumn{1}{c|}{1.21}              & \multicolumn{1}{c|}{0.91}              & 0.91            \\ \hline
\multicolumn{1}{|c|}{$Z$}                                                                                                       & \multicolumn{1}{c|}{48.72}           & \multicolumn{1}{c|}{52.58}             & \multicolumn{1}{c|}{53.41}             & 54.23           & \multicolumn{1}{c|}{47.97}           & \multicolumn{1}{c|}{52.11}             & \multicolumn{1}{c|}{52.76}             & 53.48           \\ \hline
\multicolumn{2}{|c|}{\textbf{\begin{tabular}[c]{@{}c@{}}HetNet's cost savings\\ against the benchmarks\end{tabular}}}                                                  & \multicolumn{1}{c|}{7.34\%}            & \multicolumn{1}{c|}{8.78\%}            & 10.16\%         & \multicolumn{1}{c|}{/}               & \multicolumn{1}{c|}{7.94\%}            & \multicolumn{1}{c|}{9.08\%}            & 10.30\%         \\ \hline
\multicolumn{1}{|c|}{\multirow{2}{*}{\textbf{\begin{tabular}[c]{@{}c@{}}Cost term\\ ($\rm min \cdot trip^{-1}$)\end{tabular}}}} & \multicolumn{4}{c|}{\textbf{Checkerboard 3 demand $\eta=0.5524$}}                                                                        & \multicolumn{4}{c|}{\textbf{Checkerboard 4 demand $\eta=0.5522$}}                                                                        \\ \cline{2-9} 
\multicolumn{1}{|c|}{}                                                                                                          & \multicolumn{1}{c|}{\textbf{HetNet}} & \multicolumn{1}{c|}{\textbf{H-HetNet}} & \multicolumn{1}{c|}{\textbf{P-HetNet}} & \textbf{HomNet} & \multicolumn{1}{c|}{\textbf{HetNet}} & \multicolumn{1}{c|}{\textbf{H-HetNet}} & \multicolumn{1}{c|}{\textbf{P-HetNet}} & \textbf{HomNet} \\ \hline
\multicolumn{1}{|c|}{$N_k$}                                                                                                     & \multicolumn{1}{c|}{1.58}            & \multicolumn{1}{c|}{2.49}              & \multicolumn{1}{c|}{1.60}              & 1.82            & \multicolumn{1}{c|}{1.59}            & \multicolumn{1}{c|}{2.45}              & \multicolumn{1}{c|}{1.48}              & 1.70            \\ \hline
\multicolumn{1}{|c|}{$N_h$}                                                                                                     & \multicolumn{1}{c|}{2.03}            & \multicolumn{1}{c|}{3.57}              & \multicolumn{1}{c|}{2.31}              & 2.57            & \multicolumn{1}{c|}{2.03}            & \multicolumn{1}{c|}{3.62}              & \multicolumn{1}{c|}{2.13}              & 2.38            \\ \hline
\multicolumn{1}{|c|}{$T_a$}                                                                                                     & \multicolumn{1}{c|}{13.66}           & \multicolumn{1}{c|}{11.41}             & \multicolumn{1}{c|}{16.21}             & 16.35           & \multicolumn{1}{c|}{13.76}           & \multicolumn{1}{c|}{11.29}             & \multicolumn{1}{c|}{16.38}             & 16.59           \\ \hline
\multicolumn{1}{|c|}{$T_w$}                                                                                                     & \multicolumn{1}{c|}{3.48}            & \multicolumn{1}{c|}{2.25}              & \multicolumn{1}{c|}{3.09}              & 2.56            & \multicolumn{1}{c|}{3.48}            & \multicolumn{1}{c|}{2.39}              & \multicolumn{1}{c|}{3.28}              & 2.70            \\ \hline
\multicolumn{1}{|c|}{$T_r$}                                                                                                     & \multicolumn{1}{c|}{26.24}           & \multicolumn{1}{c|}{30.65}             & \multicolumn{1}{c|}{28.67}             & 29.27           & \multicolumn{1}{c|}{26.89}           & \multicolumn{1}{c|}{31.44}             & \multicolumn{1}{c|}{29.23}             & 29.95           \\ \hline
\multicolumn{1}{|c|}{$T_t$}                                                                                                     & \multicolumn{1}{c|}{0.91}            & \multicolumn{1}{c|}{1.15}              & \multicolumn{1}{c|}{0.91}              & 0.91            & \multicolumn{1}{c|}{0.91}            & \multicolumn{1}{c|}{1.21}              & \multicolumn{1}{c|}{0.91}              & 0.91            \\ \hline
\multicolumn{1}{|c|}{$Z$}                                                                                                       & \multicolumn{1}{c|}{47.90}           & \multicolumn{1}{c|}{51.52}             & \multicolumn{1}{c|}{52.79}             & 53.48           & \multicolumn{1}{c|}{48.66}           & \multicolumn{1}{c|}{52.4}              & \multicolumn{1}{c|}{53.41}             & 54.23           \\ \hline
\multicolumn{2}{|c|}{\textbf{\begin{tabular}[c]{@{}c@{}}HetNet's cost savings\\ against the benchmarks\end{tabular}}}                                                  & \multicolumn{1}{c|}{7.03\%}            & \multicolumn{1}{c|}{9.26\%}            & 10.43\%         & \multicolumn{1}{c|}{/}               & \multicolumn{1}{c|}{7.14\%}            & \multicolumn{1}{c|}{8.89\%}            & 10.27\%         \\ \hline
\end{tabular}
}
\end{table}


The results show that HetNet consistently achieves the lowest generalized cost across all scenarios, and the cost saving is greater as the demand heterogeneity $\eta$ increases. Under relatively homogeneous demand patterns (e.g., $\eta=0.0439, 0.0792$ for monocentric and commute patterns), the performance differences among network designs are relatively small. However, under highly heterogeneous demand patterns ($\eta \approx 0.55$ for checkerboard scenarios), HetNet demonstrates a much more pronounced advantage. In these cases, the generalized cost of HetNet is 7.03\%--10.43\% lower than that of the benchmark networks. 


A closer examination of the cost components provides further insights. Compared with H-HetNet, HetNet sacrifices some of the access- and waiting-time benefits generated by very dense local routes, but substantially reduces in-vehicle and transfer times, leading to a lower overall generalized cost. Compared with P-HetNet and HomNet, the main savings of HetNet arise from simultaneous reductions in access/egress time and in-vehicle riding time. These two components are competing in geometrically restricted grid designs: increasing line density reduces access distances but tends to increase stopping delays and riding time, whereas sparser lines have the opposite effect. By spatially adapting line density and service structure to local demand, HetNet achieves a better trade-off between these competing components, while the agency-cost components remain broadly comparable.

The above results suggest that the proposed heterogeneous network design is particularly effective in capturing and accommodating spatial demand imbalances. By adapting route spacing and service structure to local demand conditions, HetNet can better allocate transit resources and mitigate inefficiencies that arise in conventional, geometrically restricted grid networks when demand is highly uneven. Moreover, even in low-heterogeneity scenarios, HetNet still achieves slightly better performance than the benchmark networks, demonstrating its robustness across different demand patterns.

\subsection{Parametric analysis}\label{sec_parametric_analysis}
This section presents a parametric analysis of the optimized designs across four networks, considering: (i) total patron demand ranging from 5,000 to 200,000 ($\rm trip \cdot hr^{-1}$); (ii) low-, medium-, and high-wage cities with $\mu = 5, 20, 25$ ($\rm \$ \cdot hr^{-1}$), respectively; (iii) city sizes of $R = 10$ and $30$ ($\rm km$); and (iv) six demand patterns. 

For the city-size analysis, comparability across scenarios is ensured by proportionally scaling the total patron demand. Specifically, when $R = 30$ ($\rm km$), the total demand is nine times that at $R = 10$. The demand density distribution is scaled accordingly, while the cruising speed of transit vehicles $v$ is kept constant. The results are shown in Figures \ref{fig_demand_parametric_analysis}–\ref{fig_citysize_parametric_analysis}.

\begin{figure}[htbp]
    \centering
    \subfigure[Monocentric demand]{
        \includegraphics[height=2.8cm]{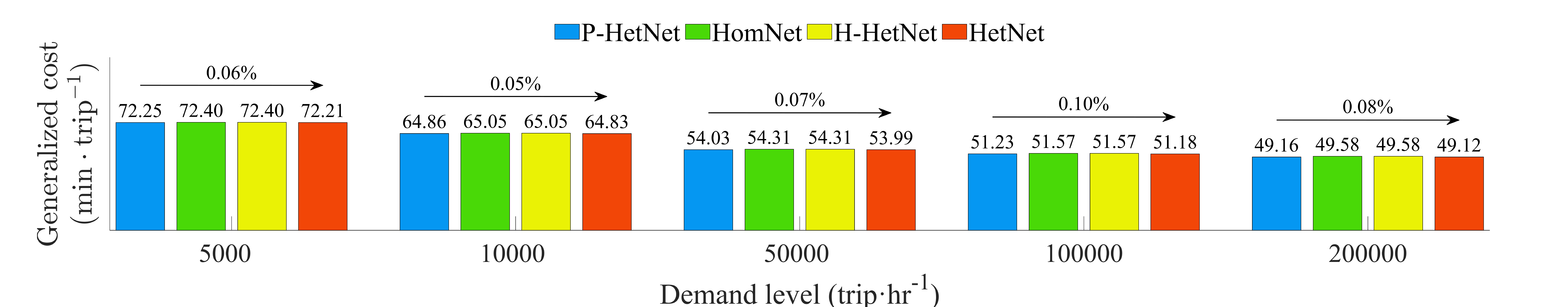}
    }
    \subfigure[Commute demand]{
        \includegraphics[height=2.8cm]{images/Fig_Commuter_demand_parametric_analysis.pdf}
    }
    \subfigure[Checkerboard 1 demand]{
        \includegraphics[height=2.8cm]{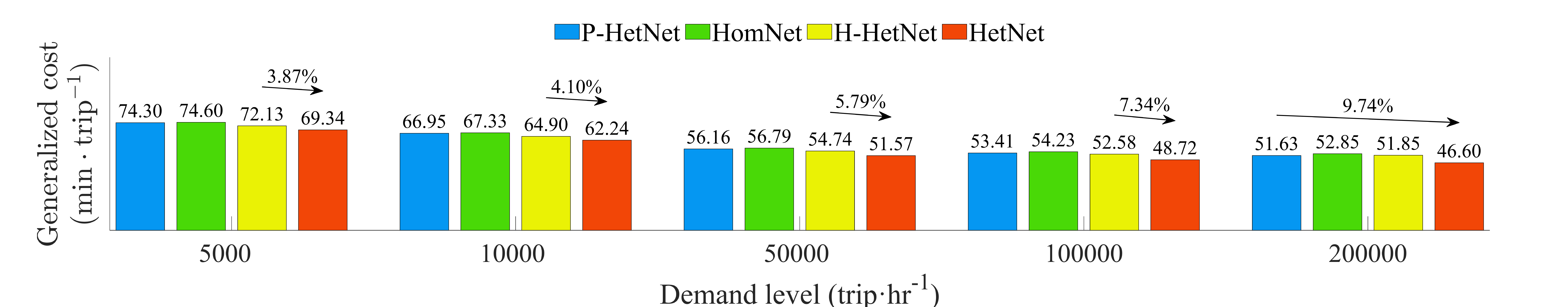}
    }
    \subfigure[Checkerboard 2 demand]{
        \includegraphics[height=2.8cm]{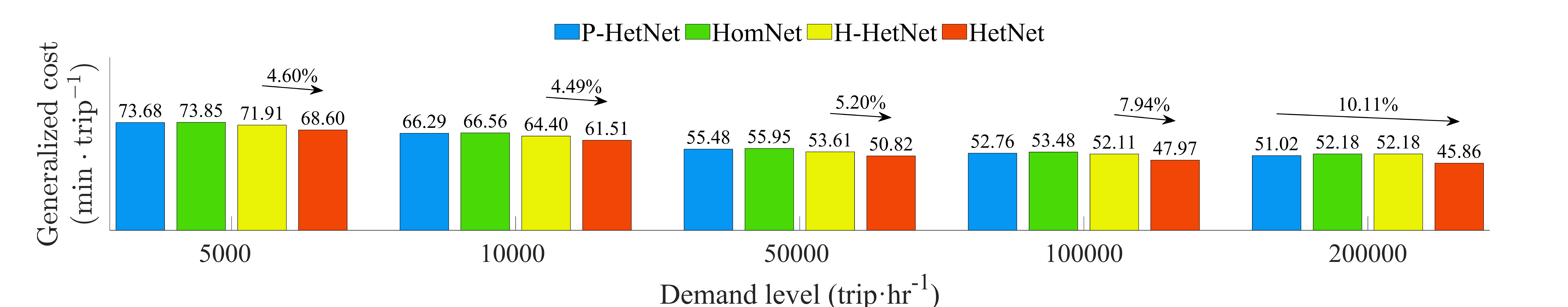}
    }
    \subfigure[Checkerboard 3 demand]{
        \includegraphics[height=2.8cm]{images/Fig_Chessboard2_demand_parametric_analysis.pdf}
    }
    \subfigure[Checkerboard 4 demand]{
        \includegraphics[height=2.8cm]{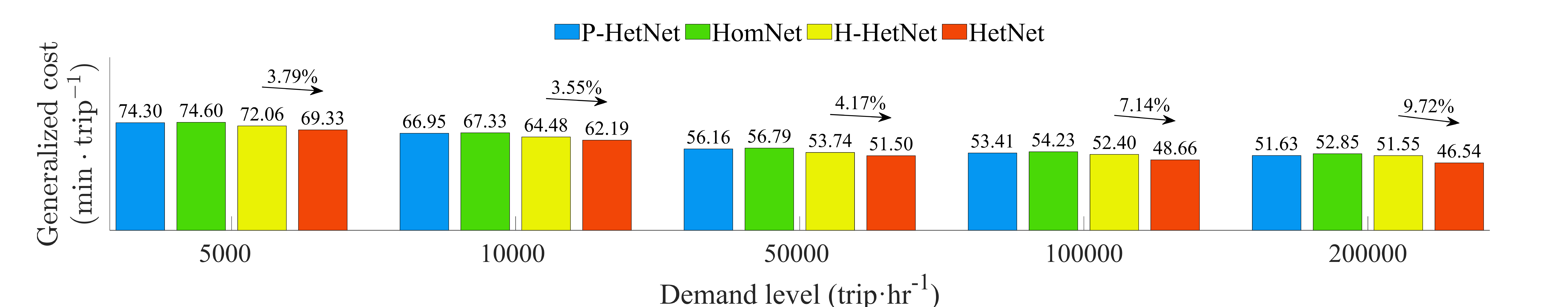}
    }
    \caption{Effects of total patron demand ($D$) on generalized cost under different demand patterns. The percentage above each group reports HetNet’s cost reduction relative to the best-performing benchmark in the same group, i.e., $(Z_{\mathrm{BM}}^{\min}-Z_{\mathrm{HetNet}})/Z_{\mathrm{BM}}^{\min}\times 100\%$, where $Z_{\mathrm{BM}}^{\min}=\min\{Z_{\mathrm{H\text{-}HetNet}},Z_{\mathrm{P\text{-}HetNet}},Z_{\mathrm{HomNet}}\}$.}
    \label{fig_demand_parametric_analysis}
\end{figure}

\begin{figure}[htbp]
    \centering
    \subfigure[Monocentric demand]{
        \includegraphics[height=2.5cm]{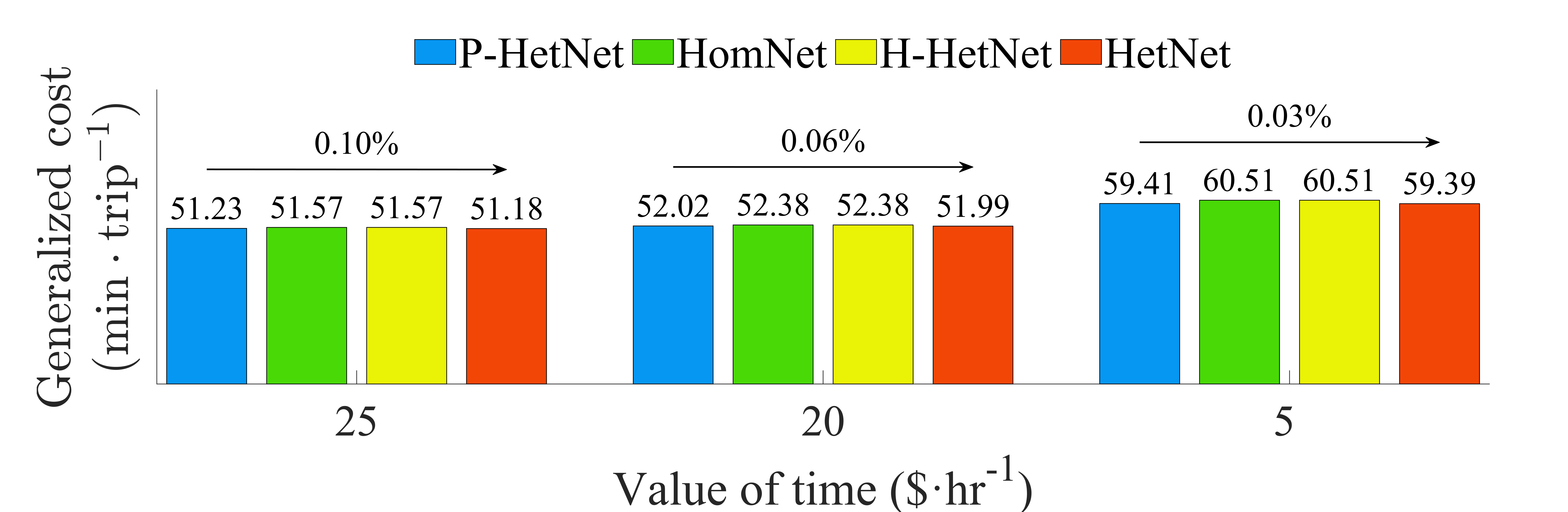}
    }
    \subfigure[Commute demand]{
        \includegraphics[height=2.5cm]{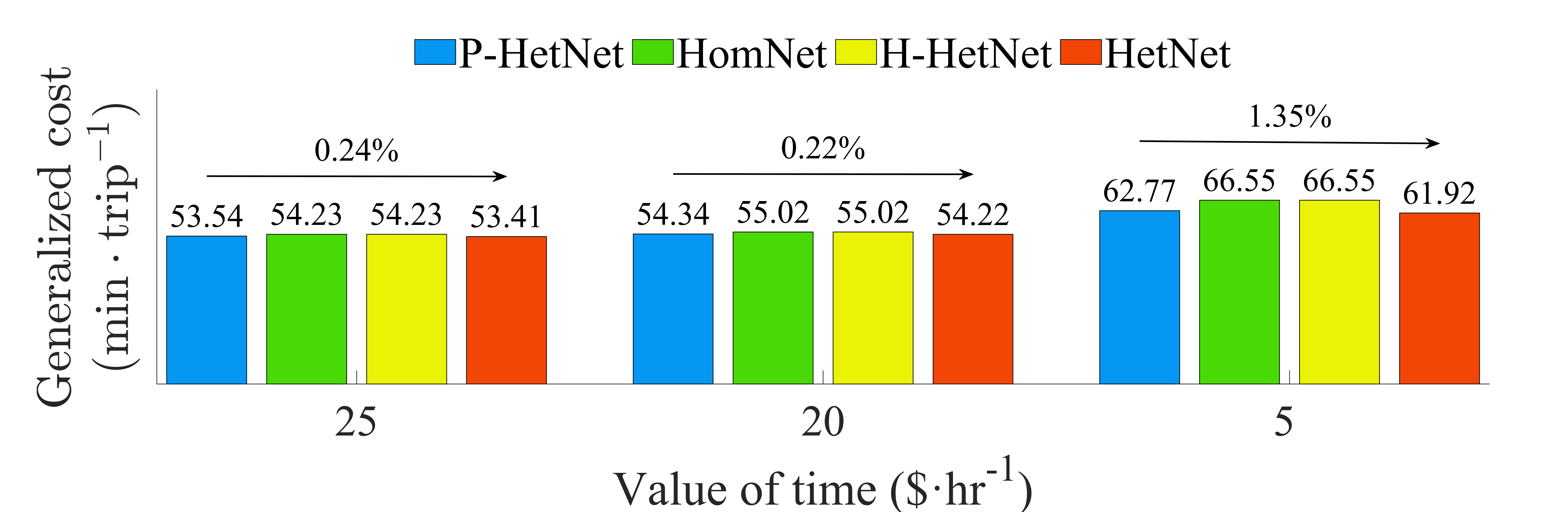}
    }
    \subfigure[Checkerboard 1 demand]{
        \includegraphics[height=2.5cm]{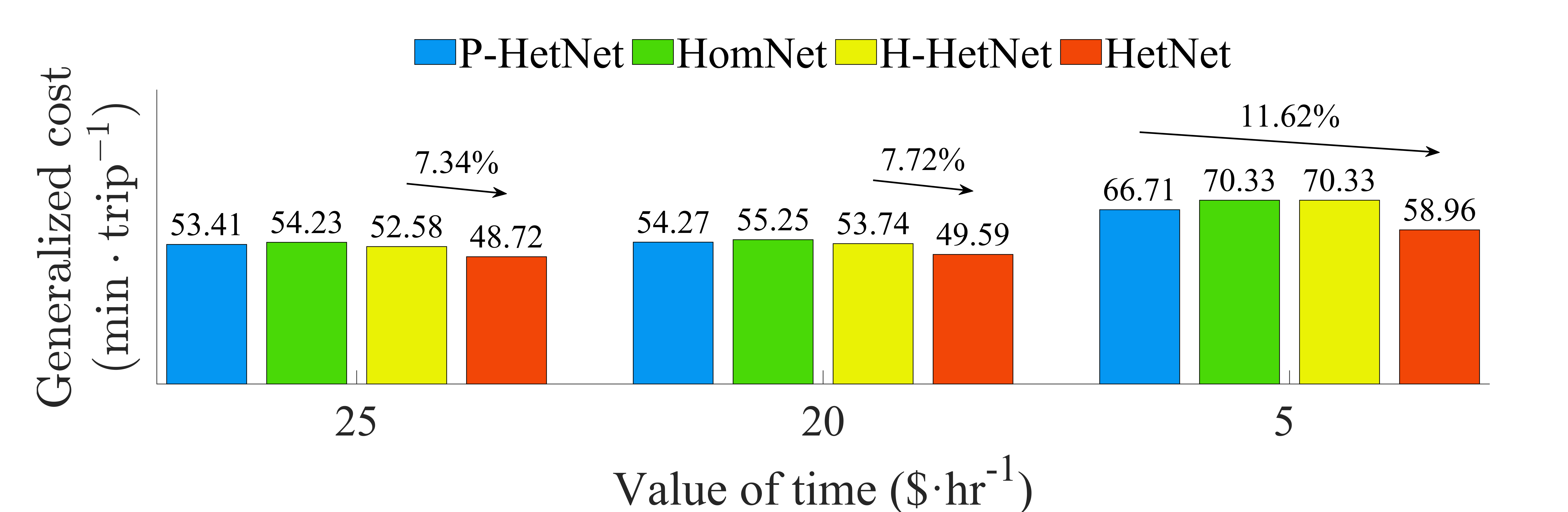}
    }
    \subfigure[Checkerboard 2 demand]{
        \includegraphics[height=2.5cm]{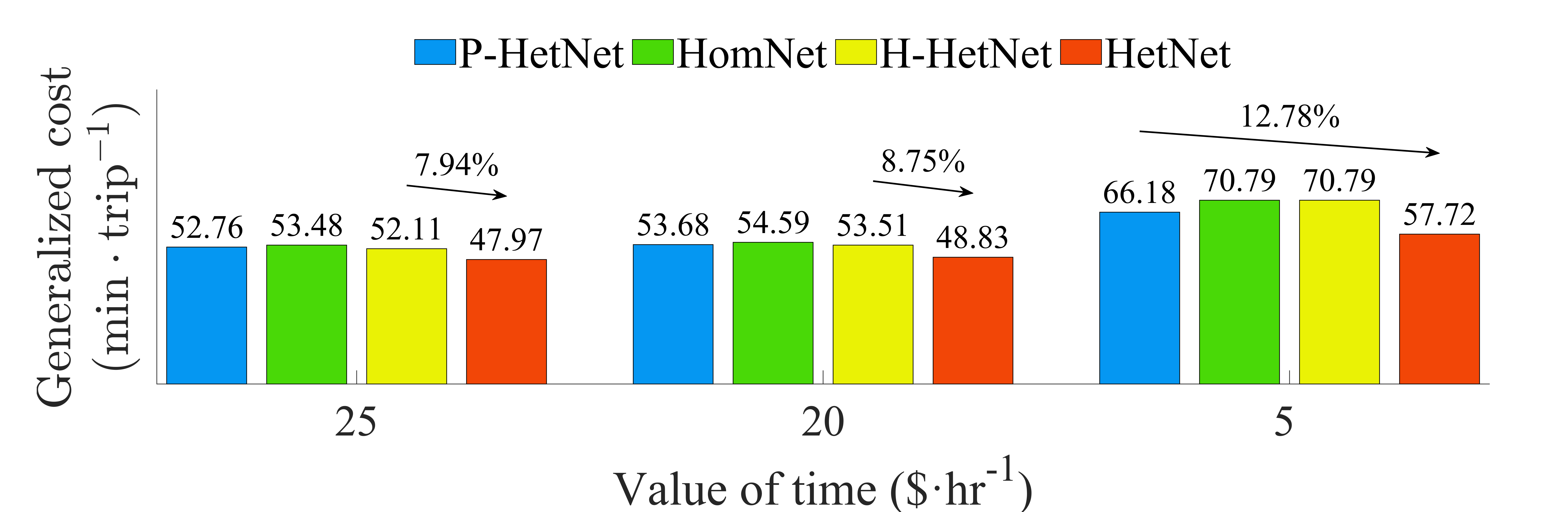}
    }
    \subfigure[Checkerboard 3 demand]{
        \includegraphics[height=2.5cm]{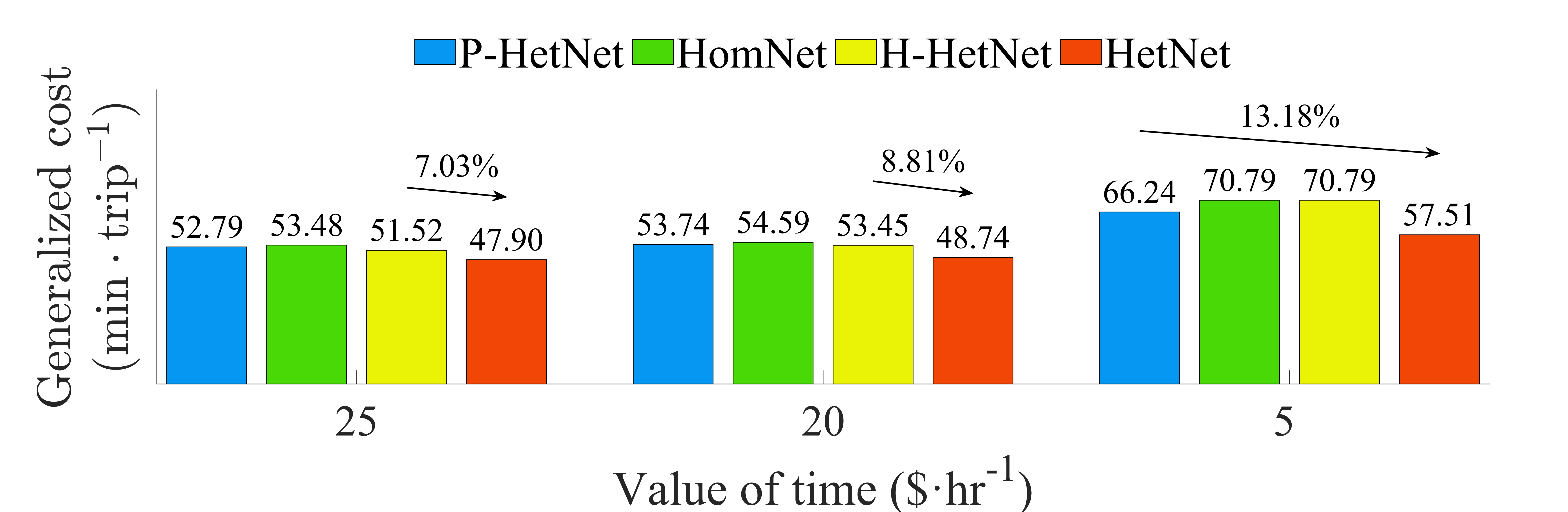}
    }
    \subfigure[Checkerboard 4 demand]{
        \includegraphics[height=2.5cm]{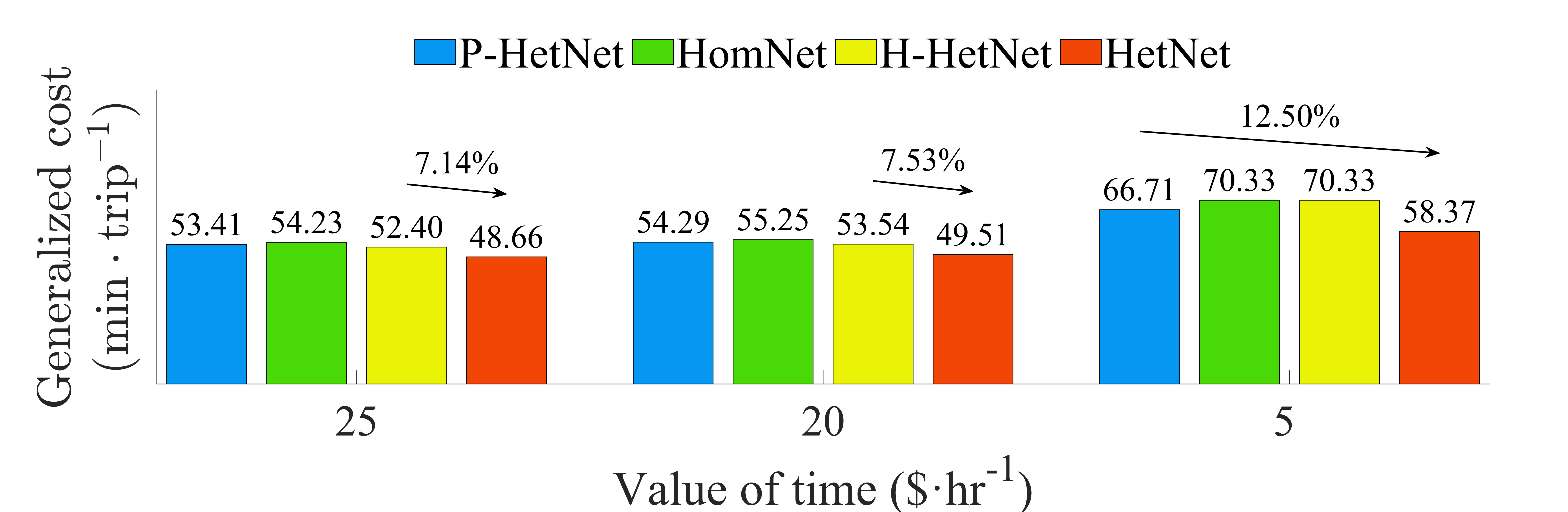}
    }
    \caption{Effects of patrons' average value of time ($\mu$) on generalized cost under different demand patterns. The percentage above each group reports HetNet’s cost reduction relative to the best-performing benchmark in the same group, i.e., $(Z_{\mathrm{BM}}^{\min}-Z_{\mathrm{HetNet}})/Z_{\mathrm{BM}}^{\min}\times 100\%$, where $Z_{\mathrm{BM}}^{\min}=\min\{Z_{\mathrm{H\text{-}HetNet}}, Z_{\mathrm{P\text{-}HetNet}}, Z_{\mathrm{HomNet}}\}$.}
    \label{fig_mu_parametric_analysis}
\end{figure}

\begin{figure}[htbp]
    \centering
    \subfigure[Monocentric demand]{
        \includegraphics[height=2.5cm]{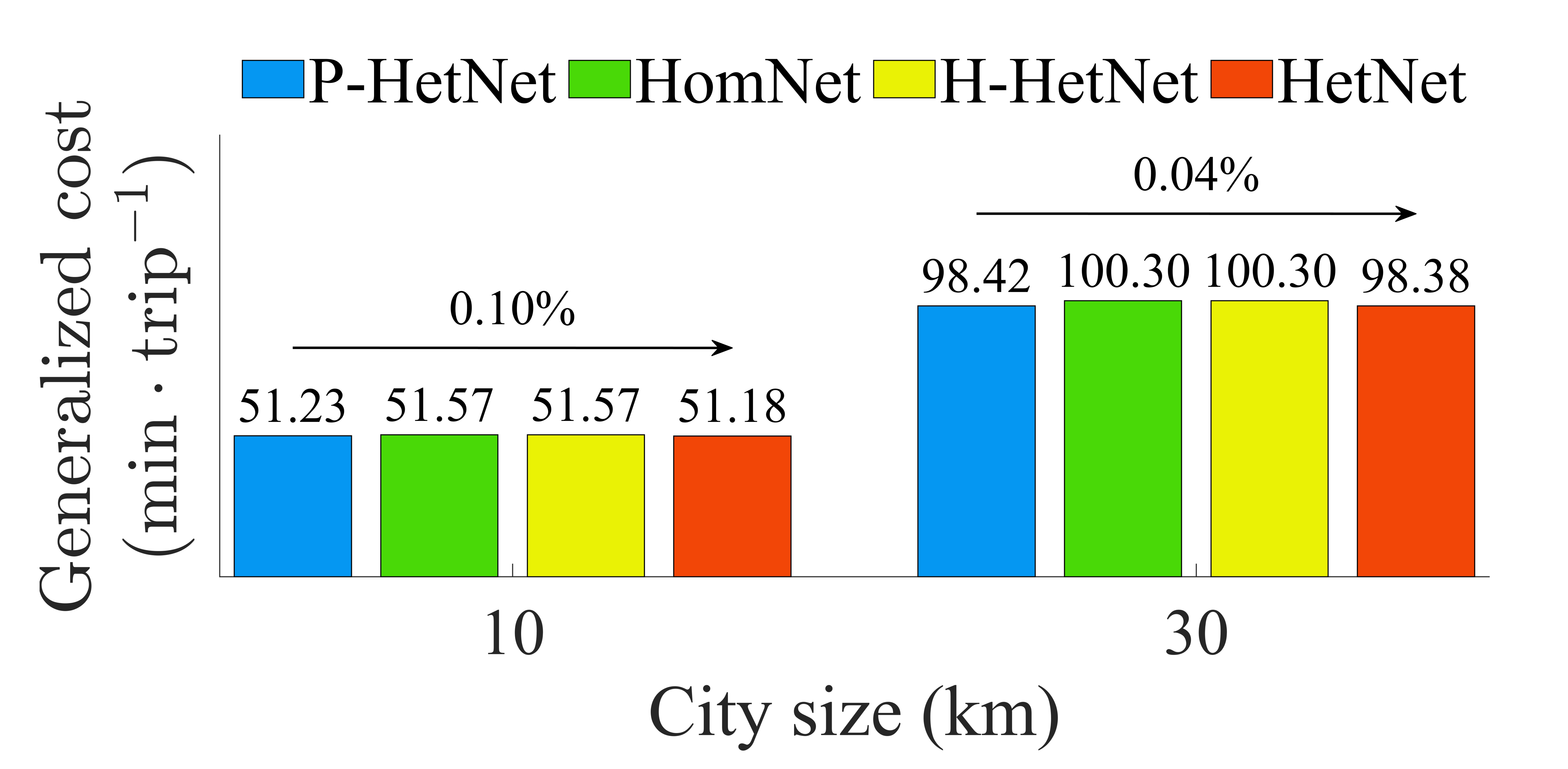}
    }
    \subfigure[Commute demand]{
        \includegraphics[height=2.5cm]{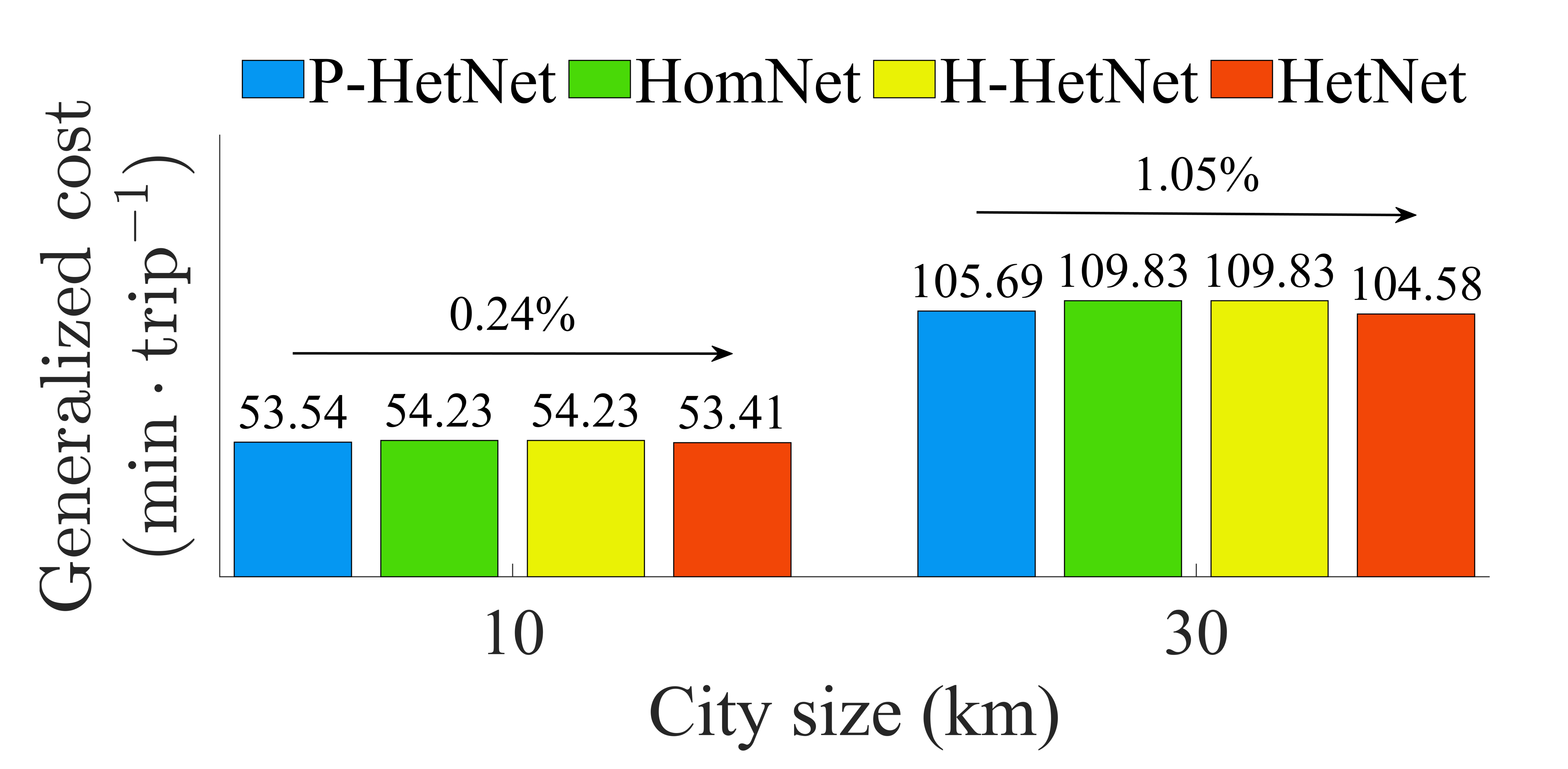}
    }
    \subfigure[Checkerboard 1 demand]{
        \includegraphics[height=2.5cm]{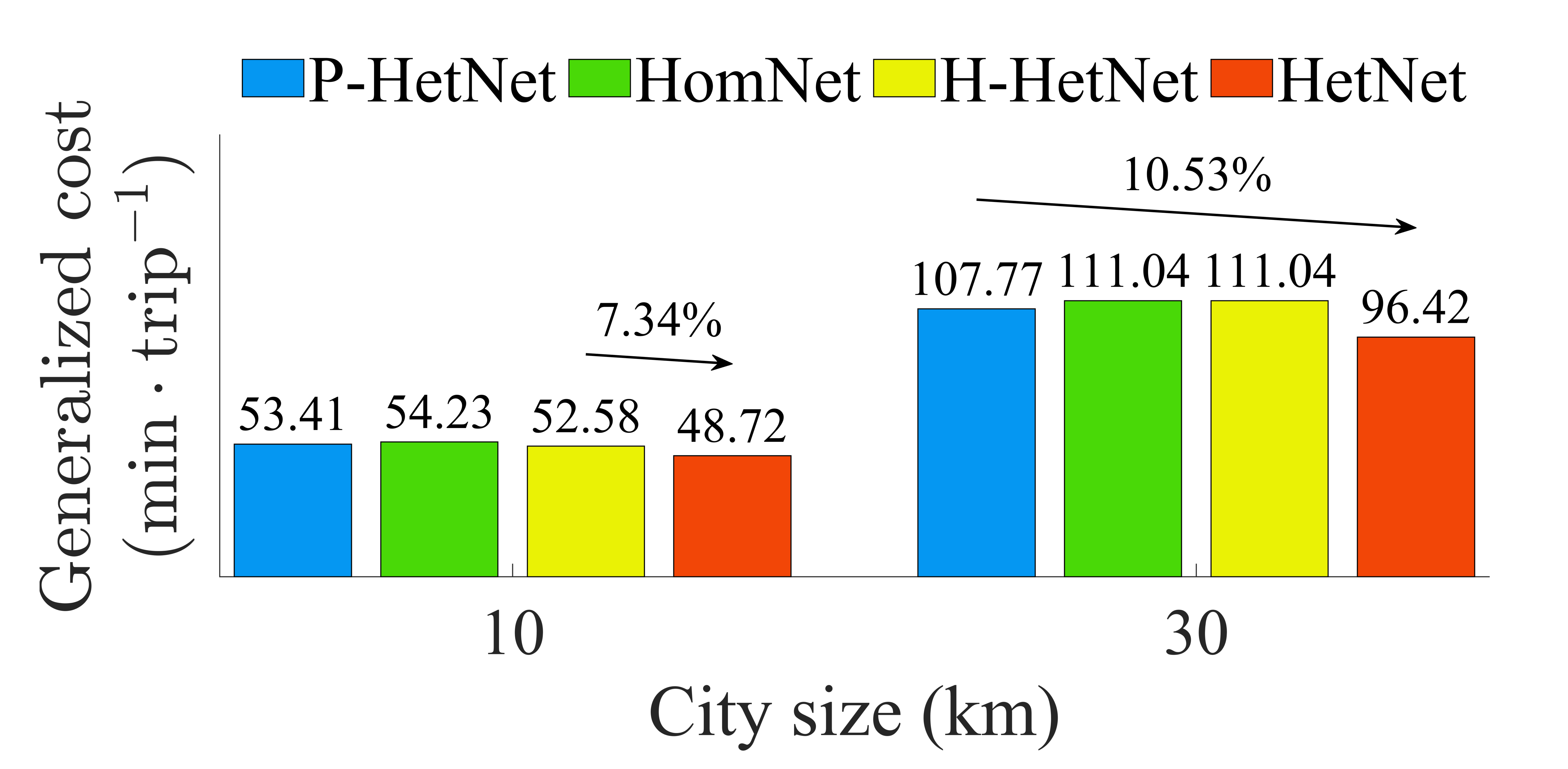}
    }
    \subfigure[Checkerboard 2 demand]{
        \includegraphics[height=2.5cm]{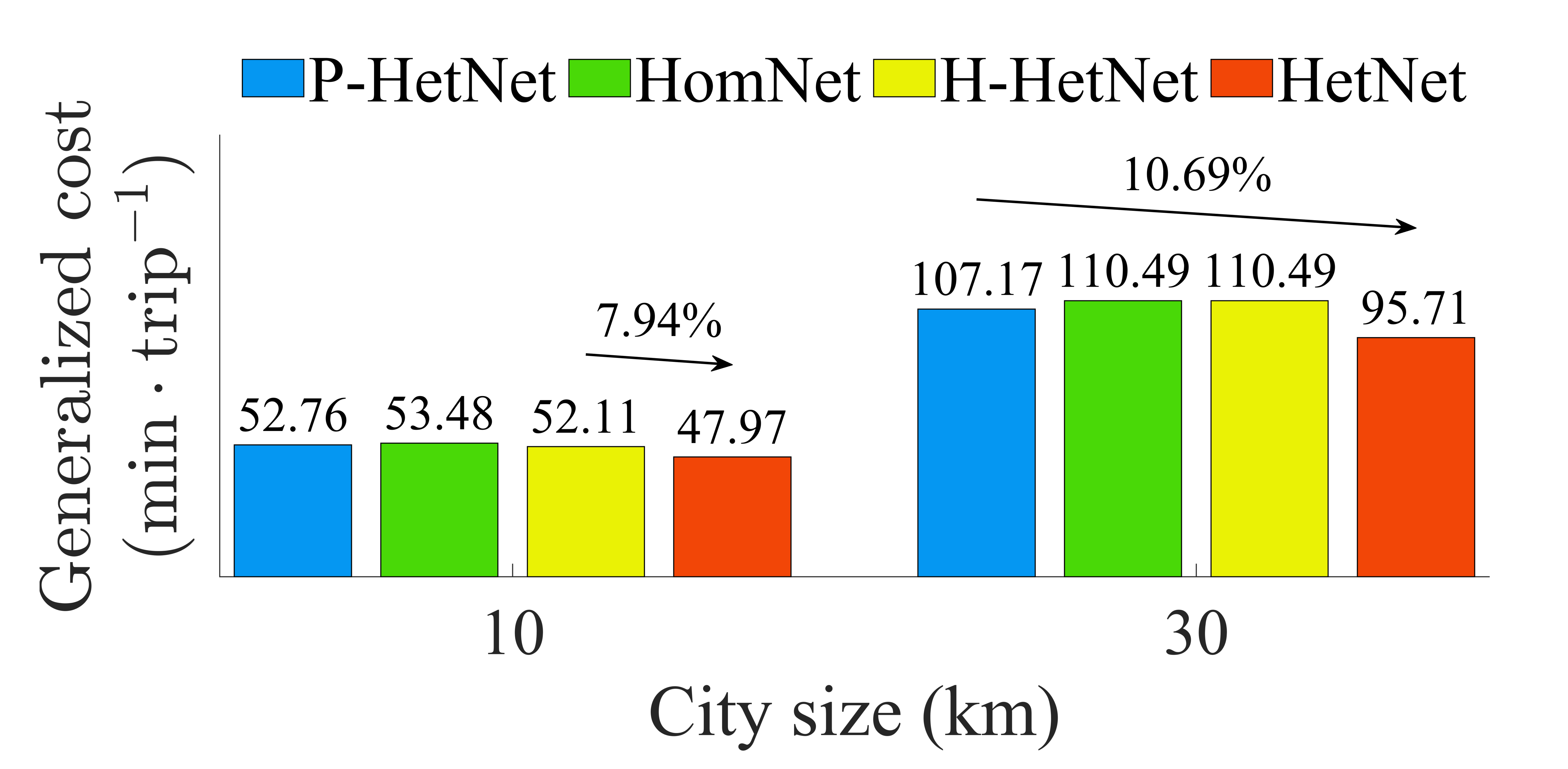}
    }
    \subfigure[Checkerboard 3 demand]{
        \includegraphics[height=2.5cm]{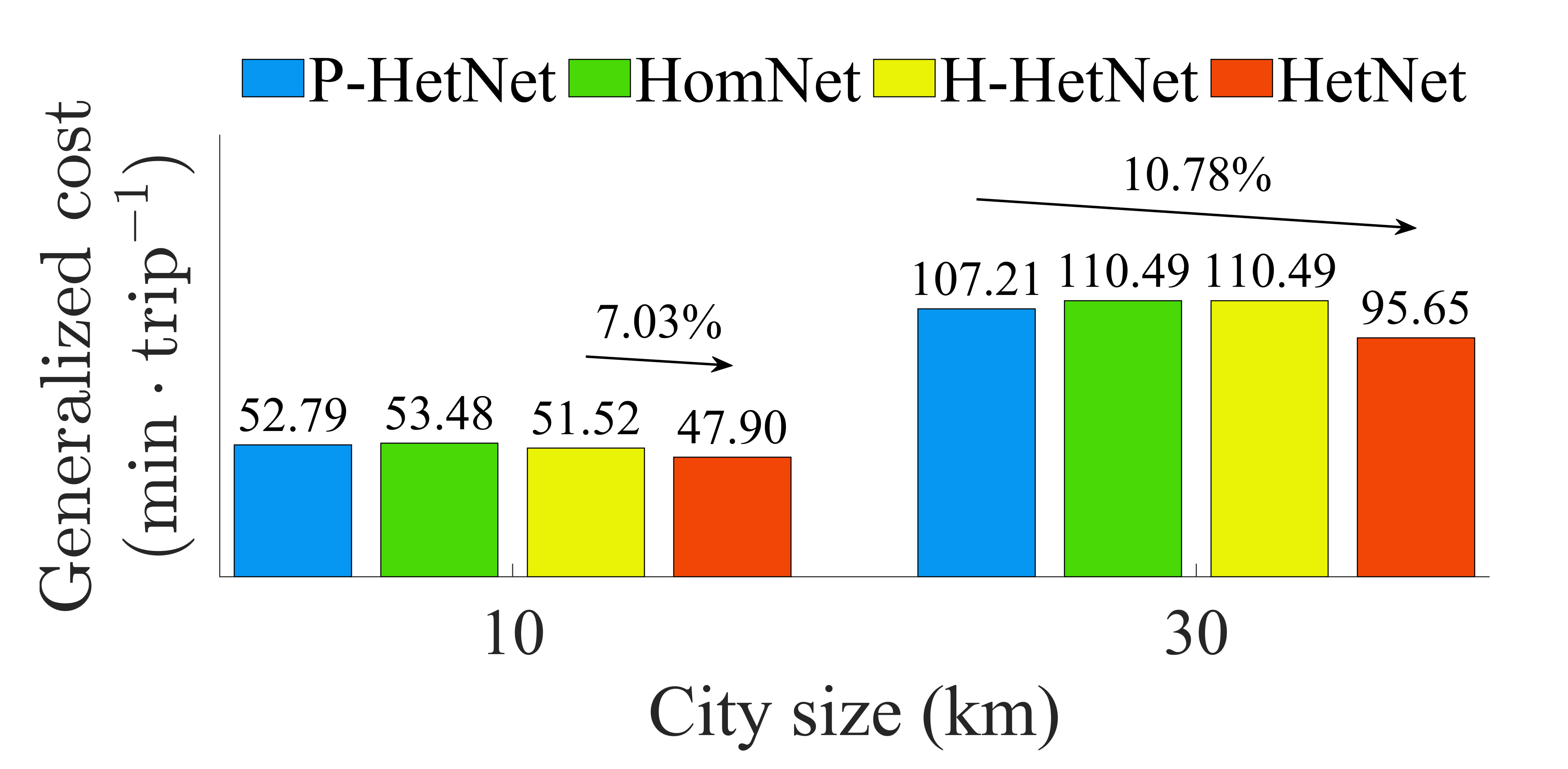}
    }
    \subfigure[Checkerboard 4 demand]{
        \includegraphics[height=2.5cm]{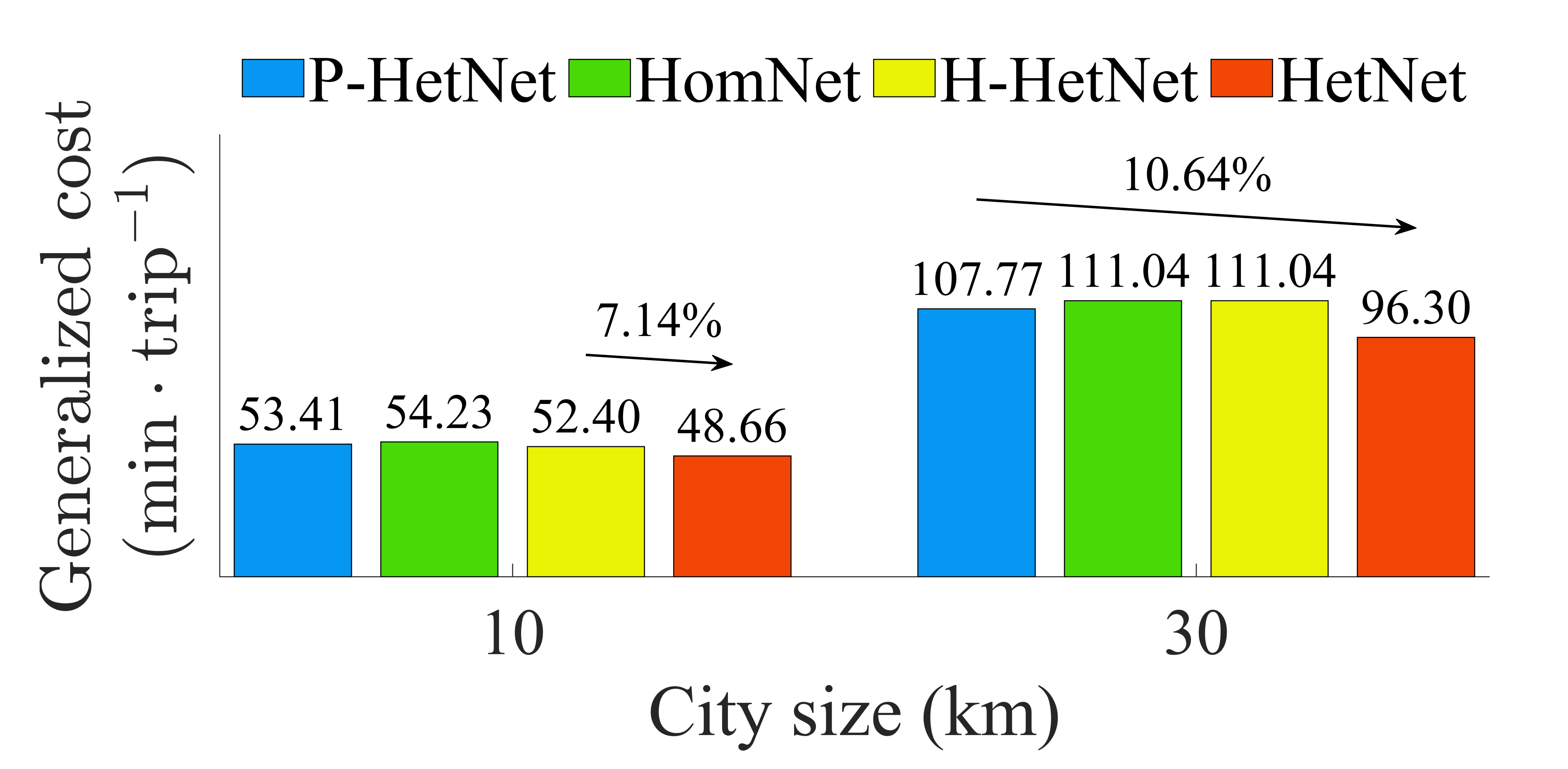}
    }
    \caption{Effects of city size ($R$) on generalized cost under different demand patterns. The percentage above each group reports HetNet’s cost reduction relative to the best-performing benchmark in the same group, i.e., $(Z_{\mathrm{BM}}^{\min}-Z_{\mathrm{HetNet}})/Z_{\mathrm{BM}}^{\min}\times 100\%$, where $Z_{\mathrm{BM}}^{\min}=\min\{Z_{\mathrm{H\text{-}HetNet}}, Z_{\mathrm{P\text{-}HetNet}}, Z_{\mathrm{HomNet}}\}$.}
    \label{fig_citysize_parametric_analysis}
\end{figure}

The results show that higher demand levels reduce the average generalized cost across all scenarios, indicating economies of scale (EOS) in all network structures. As demand increases, transit resources are utilized more efficiently and fixed infrastructure costs are distributed among a larger number of patrons, leading to lower average generalized costs. Moreover, the cost advantage of HetNet over the benchmark networks generally increases with higher demand levels. Under the four checkerboard demand patterns, the corresponding savings increase more markedly, from 3.79\%--4.60\% at $D=5{,}000$ to 9.72\%--10.26\% at $D=200{,}000$. This finding highlights the strong potential of HetNet in high-density urban environments, where demand concentrations allow the heterogeneous network structure to allocate transit services more efficiently.

The sensitivity analysis on patrons' average value of time further shows that the advantage of HetNet becomes more pronounced in cities with lower average wages. When the value of time is lower, the relative importance of operational costs in the generalized cost becomes more significant. Since HetNet can flexibly adjust route density according to spatial demand conditions, it can allocate transit resources more efficiently and reduce unnecessary infrastructure deployment. More specifically, under four checkerboard demand patterns, HetNet's cost saving increases from 7.03\%--7.94\% at $\mu=25$ to 11.62\%--13.18\% at $\mu=5$. As a result, the heterogeneous network structure provides greater cost savings under such conditions.

The analysis of city size shows a similar pattern. As the city side length increases from $R=10$ to $30$ km, HetNet's cost saving rises from 7.03\%--7.94\% to 10.53\%--10.78\% under the four checkerboard patterns. This suggests that HetNet is particularly beneficial in larger cities with spatially dispersed and heterogeneous demand, where its flexible line structure can better balance coverage and operating efficiency.

\section{Conclusions}
This paper proposes CA models for designing bus networks in grid cities under spatially heterogeneous demand. Heterogeneous layouts of bus lines are produced, in which bus lines are heterogeneous to make lateral movements from longitudinal directions, and perform line merging and diverging anywhere in the city. Note that this representation of the overlapping bus lines is the first time modeled in CA-based TND models. This flexibility brings our models significant improvements over the recent heterogeneous network model \citep{ouyangContinuumApproximationApproach2014}, particularly when patrons demand exhibits strong spatial heterogeneity.  

Methodologically, our work has extended the theory of continuum approximation models to a more general approach in transit network design that flexibly accounts for a variety of heterogeneous demand patterns. The practical applicability of the CA-based TND models in real cities is thus enhanced. This generalization, however, increases model complexity. To address this, we adopt an SGP approach, whereby the model is locally approximated around reference points and transformed into a sequence of standard GP subproblems solved iteratively until convergence. This solution framework provides an efficient approach for handling CA-based transit network design problems that contain nonlinear terms and absolute-value structures.

In numerical experiments, we firstly assess the accuracy of the proposed continuum model by comparing the model outputs with those obtained from the corresponding discretized transit networks, with estimation errors in the generalized cost remaining below 3\% across all demand patterns. We then compare the proposed HetNet with three benchmarks—H-HetNet, P-HetNet, and HomNet—and find that HetNet achieves over 7\% generalized cost savings under checkerboard demand patterns. New findings in this paper include: (i) our heterogeneous network design model demonstrates consistent superiority over other network models across all tested scenarios; (ii) the advantages of heterogeneous network structures become more pronounced under highly heterogeneous demand patterns; and (iii) heterogeneous designs will primarily benefit high-demand, low-wage, and large-area cities.


Admittedly, this paper is limited in its simplified representation of transit users' route choice behaviors, which is necessary for developing a parsimonious formulation and for applying local decomposition solution techniques. Other network forms, such as hybrid networks \citep{daganzoStructureCompetitiveTransit2010} and ring-radial networks \citep{chenOptimalTransitService2015}, are also worthy of exploration for heterogeneous designs. Future work of interest also includes extending the model to accommodate irregular city shapes and to incorporate multiple transit modes \citep{fanOptimalDesignIntersecting2018a}.

\ACKNOWLEDGMENT{%
This part will be filled after peer review.
}

%
%
%
 \begin{APPENDICES}

\numberwithin{equation}{section}
\counterwithin{figure}{section}
\counterwithin{table}{section}

\renewcommand{\theHsection}{appendix.\Alph{section}}
\renewcommand{\theHequation}{appendix.\thesection.\arabic{equation}}
\renewcommand{\theHfigure}{appendix.\thesection.\arabic{figure}}
\renewcommand{\theHtable}{appendix.\thesection.\arabic{table}}

\section{Nomenclature} \label{Appen_A}
The key notations are as follows.
\begin{table}[!ht]
\centering
\caption{Nomenclature}
\label{tab:Nomenclature}
\resizebox{\textwidth}{!}{%
\begin{tabular}{ll}
\hline
Decision functions & Description    \\
\hline
$\delta_{i}(\Vec{x})$  & line density of direction $i$ at $\Vec{x}$ ($\rm km^{-1}$) \\
$h_{i}(\Vec{x}) $    &   Headways of $i$-direction lines at $\Vec{x}$ ($\rm hr \cdot veh^{-1}$)\\
\hline
\makecell[l]{Other variables\\functions}    & \\
\hline
$R$ & Side length of a square city ($\rm km$) \\
$\lambda(\Vec{x}_o,\Vec{x}_d)$ & Demand densities of trips originating at $\Vec{x}_o$ and ending at $\Vec{x}_d$ ($\rm trip \cdot km^{-4} \cdot hr^{-1}$)\\
$\lambda_{\rm bo}(\Vec{x}),\lambda_{\rm al}(\Vec{x}),\lambda_{\rm tr}(\Vec{x})$    &  Boarding, alighting, and transferring demand densities at $\Vec{x}$ ($\rm trip \cdot km^{-2} \cdot hr^{-1}$) \\
$\lambda_{\rm fl}^{i}(\Vec{x})$    &  Directional on-board flux of $i$-direction lines at $\Vec{x}$ ($\rm trip \cdot km^{-1} \cdot hr^{-1}$) \\
$Z$ & Generalized cost ($\rm hr$)\\
$Z_A:N_l, N_s, N_k, N_h$   & \makecell[l]{Agency metrics ($\$$):  line infrastructure length ($\rm km$), number of stops, \\ vehicle-kilometers traveled per hour ($\rm veh \cdot km$), and vehicle-hours traveled per hour ($\rm veh \cdot hr$)}  \\
$\pi_l, \pi_s,\pi_k,\pi_h$    &  \makecell[l]{Unit cost of line infrastructure length ($\rm \$ \cdot km^{-1}$), number of stops ($\$$), \\ vehicle-kilometers traveled ($\rm \$ \cdot veh^{-1} \cdot km^{-1}$), vehicle-hours traveled ($\rm \$ \cdot veh^{-1} \cdot hr^{-1}$)} \\
$Z_P:T_a,T_w, T_r,T_t$    & \makecell[l]{User metrics ($\rm hr$): access/egress cost ($\rm hr$), wait time ($\rm hr$),  \\  in-vehicle travel time ($\rm hr$), transfer penalty ($\rm hr$)} \\
$q_i(\Vec{x})$   & Vehicle flow per line per km$^2$ at $\Vec{x}$ ($\rm veh \cdot hr^{-1} \cdot km^{-1}$)\\
$Q_i(\Vec{x})$   & Cumulative vehicle flow per line per km$^2$ at $\Vec{x}$ ($\rm veh \cdot hr^{-1}$)\\
$\mathrm{Q}_i $  &  Cross-sectional vehicle flows per hour of direction $i$ ($\rm veh \cdot hr^{-1}$)\\
$d_i(\Vec{x})$   & Detour distance per line per km$^2$ at $\Vec{x}$ ($\rm veh \cdot hr^{-1} \cdot km^{-1}$)\\
$D$ & Total patron demand ($\rm trip \cdot hr^{-1}$)\\
$v,v_w$    & Cruising speed of transit vehicles; walking speed of patrons ($\rm km \cdot hr^{-1}$)\\
$C$    & Capacity of a transit vehicle ($\rm trip \cdot veh^{-1}$)\\
$\mu$    & Value of time of patrons ($\rm \$ \cdot hr^{-1}$) \\
$\tau, \sigma$    & \makecell[l]{Fixed delay per stop ($\rm hr \cdot stop^{-1}$), transfer penalty ($\rm hr \cdot transfer^{-1}$)}\\
$\alpha $    & Patrons' detour impact factor\\
$\beta_w$ & Perceived walking time factor of patrons\\
$\eta$ & Patrons' demand heterogeneity\\
$\Delta $    & Discretization parameters ($\rm km$)\\

\hline
\end{tabular}%
}
\end{table}

\section{Demand functions} \label{Appen_B}
The following demand functions are derived based on the route choice behaviors as described in section \ref{sec_HetNet}. Detailed derivations can be found in \citeauthor{ouyangContinuumApproximationApproach2014} (\citeyear{ouyangContinuumApproximationApproach2014}, section 2.3, pp. 336-338). Here, we only summarize the results with minor changes in the symbols of variables/functions. 
 
(i) Boarding and alighting demand functions\\
Boarding densities of patrons originating in the neighborhood at $\Vec{x} = (x,y)$ in four directions, i.e., Eastbound (denoted by superscript $\rm E$), Westbound ($\rm W$), Northbound ($\rm N$), and Southbound ($\rm S$), are:
\begin{subequations}
    \begin{align}
        & \lambda^{\rm E}_{\rm bo}(x,y) = \frac{1}{2} \int^{R}_{y'=0} \int^{R}_{x'=x} {\lambda(x,y,x',y') \diff x'\diff y'} , \quad & \lambda^{\rm W}_{\rm bo}(x,y) =  \frac{1}{2}\int^{R}_{y'=0} \int^{x}_{x'=0} {\lambda(x,y,x',y') \diff x'\diff y'} , \\
        & \lambda^{\rm N}_{\rm bo}(x,y) = \frac{1}{2} \int^{R}_{y'=y} \int^{R}_{x'=0} {\lambda(x,y,x',y') \diff x'\diff y'} , \quad & \lambda^{\rm S}_{\rm bo}(x,y) =  \frac{1}{2}\int^{y}_{y'=0} \int^{R}_{x'=0} {\lambda(x,y,x',y') \diff x'\diff y'} .
    \end{align}
\end{subequations}

Similarly, alighting demand ($\rm trip \cdot km^{-2} \cdot hr^{-1}$) ending at $(x,y)$ from four directions are:
\begin{subequations}
    \begin{align}
        & \lambda^{\rm E}_{\rm al}(x,y) = \frac{1}{2} \int^{R}_{y'=0} \int^{x}_{x'=0} {\lambda(x',y',x,y) \diff x'\diff y'}, & \lambda^{\rm W}_{\rm al}(x,y) =  \frac{1}{2}\int^{R}_{y'=0} \int^{R}_{x'=x} {\lambda(x',y',x,y) \diff x'\diff y'}, \\
        & \lambda^{\rm N}_{\rm al}(x,y) = \frac{1}{2} \int^{y}_{y'=0} \int^{R}_{x'=0} {\lambda(x',y',x,y) \diff x'\diff y'},  & \lambda^{\rm S}_{\rm al}(x,y) =  \frac{1}{2} \int^{R}_{y'=y} \int^{R}_{x'=0} {\lambda(x',y',x,y) \diff x'\diff y'} .
    \end{align}
\end{subequations}

Thus, total boarding and alighting demand ($\rm trip \cdot km^{-2} \cdot hr^{-1}$) at $(x,y)$ are:
\begin{align}
    & \lambda_{\rm bo}(x,y) = \sum_{i} {\lambda^{i}_{\rm bo}(x,y) }, \quad & \lambda_{\rm al}(x,y) = \sum_{i} {\lambda^{i}_{\rm al}(x,y) }.
\end{align}

(ii) Directional transfer demand functions\\
Transfer demand ($\rm trip \cdot km^{-2} \cdot hr^{-1}$) at $(x,y)$ to four directions are:
\begin{subequations}
    \begin{align}
        & \lambda^{\rm E}_{\rm tr} (x,y) = \frac{1}{2} \int_{y'=0}^{R} \int_{x' = x}^{R} { \lambda(x, y', x', y) \diff x' \diff y'}, & \lambda^{\rm W}_{\rm tr} (x,y) = \frac{1}{2} \int_{y' = 0}^{R} \int_{x'=0}^{x} { \lambda(x,y', x', y) \diff x' \diff y' }, \\
        & \lambda^{\rm N}_{\rm tr} (x,y) = \frac{1}{2} \int_{y'= y}^{R} \int_{x' = 0}^{R} { \lambda(x', y, x, y') \diff x' \diff y'}, & \lambda^{\rm S}_{\rm tr} (x,y) = \frac{1}{2} \int_{y' = 0}^{y} \int_{x'=0}^{R} { \lambda(x', y, x, y') \diff x' \diff y'} .
    \end{align}
\end{subequations}

 (iii) On-board flux functions \\
The fluxes of onboard passengers ($\rm trip \cdot km^{-1} \cdot hr^{-1}$) passing through $(x,y)$ in four directions are:
\begin{subequations}
    \begin{align}
        \lambda_{\rm fl}^{\rm E} (x,y) = \frac{1}{2} \int_{x_1=0}^{x} \left[ \int_{x_2 = x}^{R} \int_{y_2 = 0}^{R} \lambda(x_1, y, x_2, y_2) \diff y_2 \diff x_2 \right] \diff x_1 + \frac{1}{2} \int_{x_2 = x}^{R} \left[ \int_{x_1 = 0}^{x} \int_{y_1 = 0}^{R} \lambda(x_1, y_1, x_2, y) \diff y_1 \diff x_1 \right] \diff x_2 \\
        \lambda_{\rm fl}^{\rm W} (x,y) = \frac{1}{2} \int_{x_1=0}^{x} \left[ \int_{x_2 = x}^{R} \int_{y_2 = 0}^{R} \lambda(x_2, y_2, x_1, y) \diff y_2 \diff x_2 \right] \diff x_1 + \frac{1}{2} \int_{x_2 = x}^{R} \left[ \int_{x_1 = 0}^{x} \int_{y_1 = 0}^{R} \lambda(x_2, y, x_1, y_1) \diff y_1 \diff x_1 \right] \diff x_2 \\
        \lambda_{\rm fl}^{\rm N} (x,y) = \frac{1}{2} \int_{y_1=0}^{y} \left[ \int_{y_2 = y}^{R} \int_{x_2 = 0}^{R} \lambda(x, y_1, x_2, y_2) \diff x_2 \diff y_2 \right] \diff y_1 + \frac{1}{2} \int_{y_2 = y}^{R} \left[ \int_{y_1 = 0}^{y} \int_{x_1 = 0}^{R} \lambda(x_1, y_1, x, y_2) \diff x_1 \diff y_1 \right] \diff y_2 \\
        \lambda_{\rm fl}^{\rm S} (x,y) = \frac{1}{2} \int_{y_1=0}^{y} \left[ \int_{y_2 = y}^{R} \int_{x_2 = 0}^{R} \lambda(x_2, y_2, x, y_1) \diff x_2 \diff y_2 \right] \diff y_1 + \frac{1}{2} \int_{y_2 = y}^{R} \left[ \int_{y_1 = 0}^{y} \int_{x_1 = 0}^{R} \lambda(x, y_2, x_1, y_1) \diff x_1 \diff y_1 \right] \diff y_2.
    \end{align}
\end{subequations}

\section{Heterogeneous demand density generation}\label{Appen_C}
For the monocentric and commute demand patterns, the demand density function is specified as

\begin{align}
    \lambda(x_o,y_o,x_d,y_d) = \frac{D}{\Lambda} \prod_{e=o,d} \left(a_1 + a_2\sum^2_{k=1} \exp\left[- (a_{3k}x_{e} - b_{ek})^2 - (a_{4k}y_{e} - \overline{b}_{ek})^2\right] \right),
\end{align}
where $\lambda(x_o,y_o,x_d,y_d)$ is the density ($\rm trip \cdot km^{-4} \cdot hr^{-1}$) of demand that originates from the neighborhood area centered at $(x_o, y_o)$ and ends at $ (x_d, y_d) $, $D$ represents the total patron
demand, and parameters $a_1, a_2, a_{3k}, a_{4k}, b_{ek}$, and $\overline{b}_{ek}$ take the values specified in the cited work, as shown in Table \ref{tab_ouyang_demand}, for monocentric and commute demand patterns, respectively. 

\begin{table}[!ht]
    \centering
    \caption{Demand distribution parameters borrowed from Ouyang et al. (2014).}
    \label{tab_ouyang_demand}
    \begin{tabular}{ccccccccccccccc}
        \hline
         Demand & $a_1$ & $a_2$ & $a_{31}$ & $a_{32}$ & $a_{41}$ & $a_{42}$ & $b_{o1}$ & $b_{o2}$ & $b_{d1}$ & $b_{d2}$ & $\overline{b}_{o1}$ & $\overline{b}_{o2}$ & $\overline{b}_{d1}$ & $\overline{b}_{d2}$ \\
         \hline
         Monocentric & 0.0016 & 0.065 & 0.5 & 0 & 0.5 & 0 & 2.5 & 0 & 2.5 & 0 & 2.5 & 0 & 2.5 & 0 \\
         Commute & 0.00044 & 0.070 & 0.5 & 0 & 0.5 & 0 & 1.0 & 0 & 4.0 & 0 & 4.0 & 0 & 1.0 & 0 \\
         \hline
    \end{tabular}
\end{table}

The $\Lambda$ is added to the original function in \citeauthor{ouyangContinuumApproximationApproach2014} (\citeyear{ouyangContinuumApproximationApproach2014}, p.340, eq. 14), a scaling parameter $\Lambda=\iiiint_{0}^R\lambda(x_o,y_o,x_d,y_d) \diff x_o \diff y_o \diff x_d \diff y_d$ is introduced here to normalize the demand density function so that the total integrated demand over the study region equals the prescribed value $D$. The generated results are shown in Figures \ref{Fig_monocentric_demand} and \ref{Fig_commute_demand}.

Regarding the checkerboard demand density, the core idea is to divide the network into high-density $H$ and low-density $L$ regions, which are distributed in a checkerboard-like pattern across the network. The patrons' demand density remains constant within each $H$ and $L$ region. The higher heterogeneity of this pattern arises from the asymmetric distribution of travel origins and destinations, which may facilitate a more realistic design. The design process for this demand distribution consists of three steps:
    
\begin{enumerate}
\item Identifying the locations of each $H$ and $L$ regions.
\item Constructing patron transfer equation set to solve for demand densities in the four transfer scenarios: high-density to high-density, $\lambda_{H \rightarrow H}$; high-density to low-density, $\lambda_{H \rightarrow L}$; low-density to high-density, $\lambda_{L \rightarrow H}$; and low-density to low-density, $\lambda_{L \rightarrow L}$. The equation set is given in

\begin{subequations}\label{Eq_checkerboard_demand_equation_set}
\begin{align}
&\left\{
\begin{aligned}
\lambda_{bo,H} &= R_H \lambda_{H \rightarrow H} + R_L \lambda_{H \rightarrow L} \\
\lambda_{bo,L} &= R_H \lambda_{L \rightarrow H} + R_L \lambda_{L \rightarrow L} \\
\lambda_{al,H} &= R_H \lambda_{H \rightarrow H} + R_L \lambda_{L \rightarrow H} \\
\lambda_{al,L} &= R_H \lambda_{H \rightarrow L} + R_L \lambda_{L \rightarrow L}
\end{aligned}
\right. ,
\label{Eq_checkerboard_demand_equation_set_1} \\
&\left\{
\begin{aligned}
R_H \lambda_{bo,H} + R_L \lambda_{bo,L} &= D \\
R_H \lambda_{al,H} + R_L \lambda_{al,L} &= D \\
\frac{R_H \lambda_{bo,H}}{D} = \frac{R_H \lambda_{al,H}}{D} &= \rho_H \\
\frac{R_H \lambda_{H \rightarrow H}}{\lambda_{bo,H}} &= \rho_{H \rightarrow H}
\end{aligned}
\right. ,
\label{Eq_checkerboard_demand_equation_set_2}
\end{align}
\end{subequations}

where Eq. (\ref{Eq_checkerboard_demand_equation_set_1}) calculates the departure and arrival demand densities, and Eq. (\ref{Eq_checkerboard_demand_equation_set_2}) constrains the total demand and relevant ratios. $R_H$ and $R_L$ represent the areas of high and low demand regions, respectively. $\lambda_{bo,H}$ and $\lambda_{bo,L}$ denote the demand densities leaving the high- and low-density regions, while $\lambda_{al,H}$ and $\lambda_{al,L}$ represent the demand densities arriving in these regions. $\rho_H$ is the proportion of demand leaving and arriving in high-density areas, and $\rho_{H \rightarrow H}$ is the proportion of demand within high-density regions that corresponds to the $H \rightarrow H$ transfer scenario. In this study, both of these ratios are set to 0.9.

The solution is given in 

\begin{equation}\label{Eq_checkerboard_demand_solution}
\left\{
\begin{aligned}
\lambda_{H \rightarrow H} &= \frac{D \rho_H}{R_H^2(2-\rho_{H \rightarrow H})} \\
\lambda_{H \rightarrow L} &= \lambda_{L \rightarrow H} = \frac{D \rho_H (1-\rho_{H \rightarrow H})}{R_H R_L(2-\rho_{H \rightarrow H})} \\
\lambda_{L \rightarrow L} &= \frac{D (2-\rho_{H \rightarrow H}-3\rho_H+2\rho_H\rho_{H \rightarrow H})}{R_L^2(2-\rho_{H \rightarrow H})} \\
\end{aligned}
\right. .
\end{equation}

Notably, $\lambda_{H \rightarrow L} = \lambda_{L \rightarrow H}$, as $\frac{R_H \lambda_{bo,H}}{D} = \frac{R_H \lambda_{al,H}}{D} = \rho_H$.

\item Assigning values to $\lambda(x_o,y_o,x_d,y_d)$ based on the results from step 2 and the locations determined in step 1.

\end{enumerate}

Finally, we obtained four distinct checkerboard demand distributions, as shown in Figures \ref{Fig_checkerboard0_demand}-\ref{Fig_checkerboard3_demand}.

\section{Demand Heterogeneity Measurement}\label{Appen_D}
This appendix provides the computational procedure for the demand heterogeneity metric $\eta$. The metric modifies the conventional JS divergence through a distance-weighting scheme to capture the impact of spatial demand distribution on transit network operational efficiency.

We first define the weighted distance function between demand points $(x_{o1},y_{o2})$ and $(x_{d1},y_{d2})$:

\begin{subequations}\label{Eq_weighted_distance_function}
\begin{align}
&\mathrm{dis}\left(x_{o_1}, y_{o_2}, x_{d_1}, y_{d_2}\right) = 
\begin{cases} 
\displaystyle\frac{\Delta \bar{d}}{\sqrt{2}R}, & (x_{o_1}, y_{o_2}) = (x_{d_1}, y_{d_2}) \\[8pt]
\displaystyle\frac{\sqrt{(x_{o_1}-x_{d_1})^2 + (y_{o_2}-y_{d_2})^2}}{\sqrt{2}|R|}, & \text{otherwise}
\end{cases},\\
&\Delta \bar{d} = \iiiint_0^\Delta \sqrt{(x_1-x_2)^2+(y_1-y_2)^2}\, dx_1 dx_2 dy_1 dy_2 \approx 0.5214\Delta,\\
\end{align}
\end{subequations}
where $\Delta \bar{d}$ represents the expected distance within a unit grid, computed through quadruple integration, and $o_1,o_2,d_1,d_2 \in \{1,2,\ldots,\mathcal{N}\}$.

Then, we construct the spatially distance-weighted demand transition probability distribution:

\begin{align}\label{Eq_demand_probability_distribution}
p\left(x_{o_i}, y_{o_j}, x_{d_l}, y_{d_k}\right) = \frac{\mathrm{dis}\left(x_{o_i}, y_{o_j}, x_{d_l}, y_{d_k}\right) \cdot \lambda\left(x_{o_i}, y_{o_j}, x_{d_l}, y_{d_k}\right)}{\displaystyle\sum_{o_i,o_j,d_l,d_k} \mathrm{dis}\left(x_{o_i}, y_{o_j}, x_{d_l}, y_{d_k}\right) \cdot \lambda\left(x_{o_i}, y_{o_j}, x_{d_l}, y_{d_k}\right)}.
\end{align}

Finally, the distance-weighted JS divergence is defined as:

\begin{subequations}\label{Eq_demand_heterogeneity_function}
\begin{align}
&
\begin{aligned}
    \eta =&\ \frac{1}{2\ln 2} \sum_{o_1,o_2,d_1,d_2} p\left(x_{o_1}, y_{o_2}, x_{d_1}, y_{d_2}\right) \ln\left(\frac{p\left(x_{o_1}, y_{o_2}, x_{d_1}, y_{d_2}\right)}{m\left(x_{o_q}, y_{o_w}, x_{d_1}, y_{d_2}\right)}\right) \\
    &+ \frac{1}{2\ln 2} \sum_{o_1,o_2,d_1,d_2} u\left(x_{o_1}, y_{o_2}, x_{d_1}, y_{d_2}\right) \ln\left(\frac{u\left(x_{o_1}, y_{o_2}, x_{d_1}, y_{d_2}\right)}{m\left(x_{o_1}, y_{o_2}, x_{d_1}, y_{d_2}\right)}\right),
\end{aligned}\\
&u\left(x_{o_1}, y_{o_2}, x_{d_1}, y_{d_2}\right) = \frac{1}{(\mathcal{N}+1)^4}, \\
&m\left(x_{o_1}, y_{o_2}, x_{d_1}, y_{d_2}\right) = \frac{1}{2}\left[p\left(x_{o_1}, y_{o_2}, x_{d_1}, y_{d_2}\right) + u\left(x_{o_1}, y_{o_2}, x_{d_1}, y_{d_2}\right)\right],
\end{align}
\end{subequations}
where $u(\cdot)$ denotes the uniform reference distribution; and $m(\cdot)$ represents the midpoint distribution between $p(\cdot)$ and $u(\cdot)$. 

This metric comprehensively considers both the structural characteristics of demand distribution and spatial distance factors, effectively quantifying the impact of demand patterns on bus network operational efficiency.

\section{Modeling P-HetNet and HomNet via Geometric Programming}\label{Appen_E}
The proposed HetNet model generalizes both the P-HetNet and HomNet formulations. Accordingly, these two models can be viewed as special cases of the proposed framework. Since vehicle detours are not considered in either model, the terms involving $d_i(\Vec{x})$ in the objective function are omitted. Moreover, the vehicle flow conservation constraints in Eqs. (\ref{cons_flow_conservation_i}) and (\ref{cons_flow_conservation_j}) are always satisfied. To enable GP, additional simplifications are introduced according to the structural properties of the two models.

For P-HetNet, the two categories of decision variables in each direction $i \in I=\left\{\rm E,W,N,S \right\}$ are assumed to remain invariant along the corresponding line direction. In addition, considering the bidirectional symmetry of the transit network, they can be specified as:

\begin{equation}\label{Eq_specialization_p_hetnet_delta}
\delta_i(\Vec{x})=
\left\{
\begin{aligned}
\delta_{\rm EW}(y), &\; i \in \{\rm E,W\}
\\
\delta_{\rm NS}(x), &\; i \in \{\rm N,S\}
\end{aligned}
\right. .
\end{equation}

Under this specification, $\delta_{\rm EW}(y)$, and $\delta_{\rm NS}(x)$ become the new decision variables. The same specialization applies to the service headway $h_i(\Vec{x})$.

The objective function must be integrated along the direction of operation for each $i \in I=\left\{\rm E,W,N,S \right\}$. Taking Eq. (\ref{eq_Nh}) as an example, the resulting expression after integration is given by

\begin{equation}\label{Eq_specialization_p_hetnet_objective}
N_h = 2\int_{y=0}^{R} q_{\rm EW}(y) \left( \frac{R}{v}+\tau \int_{x=0}^{R}\delta_{\rm NS}(x) \diff x \right) \diff y +  2\int_{x=0}^{R} q_{\rm NS}(x) \left( \frac{R}{v}+\tau \int_{y=0}^{R}\delta_{\rm EW}(y) \diff y \right) \diff x ,
\end{equation}
where $q_{\rm EW}(y) = \frac{\delta_{\rm EW}(y)}{h_{\rm EW}(y)}$, and $q_{\rm NS}(x) = \frac{\delta_{\rm NS}(x)}{h_{\rm NS}(x)}$. The transformations of the remaining objective function terms in Eqs. (\ref{Network_ZA}) and (\ref{Network_ZP}) follow the same procedure and are therefore omitted for brevity.
  
Under the above specialization, the capacity constraint in Eq. (\ref{cons_capacity}) for directions ${\rm E,W}$ can be reformulated as 

\begin{eqnarray}\label{Eq_specialization_p_hetnet_capacity}
\lambda_{fl}^{\rm EW}(y)\frac{h_{\rm EW}(y)}{\delta_{\rm EW}(y)} & \le C , \quad
\lambda_{fl}^{\rm EW}(y) & = 
\displaystyle
\max_{\substack{i \in \{{\rm E,W}\}, x \in [0,R]}}
\lambda_{fl}^i(\Vec{x}), \quad y \in [0,R].
\end{eqnarray}

The corresponding expressions for directions ${\rm N,S}$ follow analogously.

For homogeneous networks, the two categories of decision variables in each operational direction are assumed to be constant over the entire study area $R$. Consequently, $\delta_{\rm EW}, \delta_{\rm NS}, h_{\rm EW}$, and $h_{\rm NS}$ are treated as the new decision variables. The objective function is then integrated over all directions, and the capacity constraint can be simplified through a procedure similar to that used for P-HetNet.

Under these transformations, the objective functions of both P-HetNet and HomNet can be expressed as posynomials of the new decision variables. 
Moreover, the capacity constraint satisfies the requirement that inequality constraints must also be posynomials. 
As a result, the optimization problems for both network types conform to the standard form of GP and can therefore be solved efficiently to obtain globally optimal solutions. 
Additional implementation details can be found in \cite{maoDesignTransitNetworks2026}.

\section{Network discretization procedure}\label{Appen_F}
We use the optimal solutions of the CA model, namely $\delta_i^*(x_n,y_{n'})$, and $q_i^*(x_n,y_{n'})$, to guide the discretized network generation and vehicle-flow assignment, as illustrated in Fig.~\ref{fig_discretization}.

\begin{figure}[htbp]
    \centering
    \includegraphics[width=0.98\linewidth]{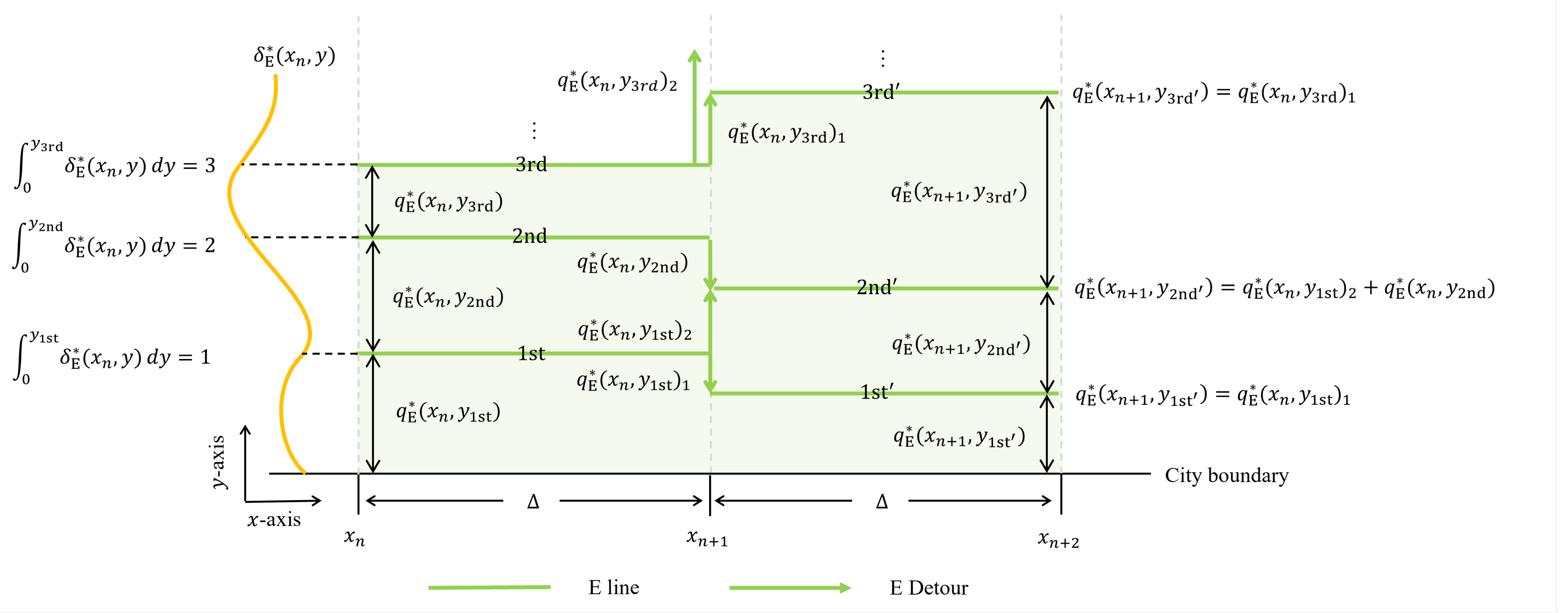}
    \caption{Illustration of the network discretization procedure.}
    \label{fig_discretization}
\end{figure}

For ease of exposition, the following procedure is described using the locations $x_n$ and $x_{n+1}$ in the E direction as an example. For line design, $\delta_i^*(x_n,y_{n'})$ is smoothed along the $\bar{i}$ direction to obtain the continuous curve $\delta_i^*(x_n,y)$, represented by the orange curve in the figure. The locations of the discrete lines are then determined by integrating $\delta_i^*(x_n,y)$ along the $\bar{i}$ direction until integer values are reached, as indicated by $y_{\rm 1st}$, $y_{\rm 2nd}$, $y_{\rm 3rd}$, $\ldots$ in the figure. Each line has a length of $\Delta$. The same design procedure applies to other locations in the E/W direction, as well as in the N/S direction.


For vehicle-flow assignment, $q_i^*(x_n,y_{n'})$ is also smoothed along the $\bar{i}$ direction to obtain $q_i^*(x_n,y)$. The vehicle flow within the interval $[y_{\rm (M-1)th},y_{\rm Mth}]$ is then assigned to the line $\rm Mth$. Specifically, the assigned flow is calculated as $\Delta\int_{y_{\rm (M-1)th}}^{y_{\rm Mth}} q_i^*(x_n,y) \diff y$. Accordingly, the headway of line $\rm Mth$ is given by the reciprocal of its assigned vehicle flow.

The vehicle flows between adjacent locations, such as from $x_n$ to $x_{n+1}$, are determined according to the flow-balance relationships among lines. As shown in Fig.~\ref{fig_discretization}, after part of the flow on line $\rm 1st$ is assigned to line $\rm 1st'$, the remaining flow, denoted by $q_{E}^*(x_n,y_{\rm 1st})_2$, is further assigned to line $\rm 2nd'$. Meanwhile, the flow on line $\rm 2nd$ is also assigned to line $\rm 2nd'$, where the flow balance is exactly satisfied. Therefore, no detoured vehicle flow is required between line $\rm 2nd$ and line $\rm 3rd'$. Following the same rule, all vehicle flows from $x_n$ to $x_{n+1}$ can be fully assigned due to the flow-conservation constraints.

After the exact locations of all discrete lines in the E/W and N/S directions are obtained, the transfer stations are determined by the intersections of lines from the two directions, consistent with the assumptions of this study.

Passenger travel routes are then computed using a shortest-path algorithm. Each passenger selects the stations closest to their origin and destination as the boarding and alighting stations, respectively, and is allowed to make at most one transfer. Finally, we verify that the vehicle capacity constraints are satisfied.

\end{APPENDICES}


\bibliographystyle{apalike}
\bibliography{refer} 


\end{document}